\documentclass[format=acmsmall, review=false, screen=true]{acmart} 

\usepackage[utf8]{inputenc}
\usepackage{latexsym}
\usepackage{graphicx}
\usepackage{pstricks}
\usepackage{pstring}
\pstrSetArrowColor{black}
\usepackage{tikz}
\usetikzlibrary{arrows,shapes,chains,matrix,automata,calc}
\usepackage{tikz-qtree}

\usepackage{xspace}
\usepackage[all]{xy}
\usepackage{manfnt}

\newcommand   \Ie{{\sl I.e.\ }}
\newcommand   \ie{{\sl i.e.\ }}
\newcommand   \eg{{\sl e.g.\ }}

\newcommand   \resp{{\sl resp.\ }}

\newcommand\textbfit[1]{{\bf\em #1}}
\newcommand{\concept}[1]{\textbfit{\boldmath #1}}
\newcommand{\defin}[1]{\concept{#1}}

        \newcommand\makeset[1]{\{\,#1\,\}}
\newcommand\anglebra[1]{\langle\, #1 \,\rangle}
\renewcommand\anglebra[1]{(#1)}

        \newcommand\domain[1]{{\sf dom}(#1)}

        \newcommand\mng[1]{{\mathopen{[\![}\,#1\,\mathclose{]\!]}}}

\newcommand\dom[1]{{\sf Dom}(#1)}
\newcommand\var{\mathit{Var}}
\newcommand\arity[1]{{\it ar}(#1)}

\newcommand{\rscheme}{\mathcal{S}}

\newcommand\Op[2]{\mathit{Op}_{#1}^{#2}}
\newcommand\topone[1]{{\mathit top}_1\,{#1}}
\newcommand\pushone[1]{{\mathit push}_1^{#1}}

\newcommand\popone{{\mathit pop}_1}

\newcommand\collapse{{\mathit collapse}}

\newcommand\pushlk[2]{{\mathit push}_{1}^{#2,#1} }
\newcommand\popn[1]{{\mathit pop}_{#1}}
\newcommand\pushn[1]{{\mathit push}_{#1}}
\newcommand\topn[1]{{\mathit top}_{#1}}
\newcommand{\toprew}[1]{\mathit rew_1^{#1}}

\newcommand\mksk[1]{\hbox{\tt{[}} #1 \hbox{\tt{]}}}

\newcommand\order[1]{{\mathit ord}(#1)}

\newcommand\id{\mathit{id}}

\newcommand\term[1]{{\mathcal{T}}(#1)}
\newcommand\tree[1]{{\mathcal{T}}^\infty(#1)}

\newcommand\stack{s}
\newcommand\varstack{t}

\renewcommand{\epsilon}{\varepsilon}
\renewcommand{\phi}{\varphi}

\newcommand{\Eloise}{\'Elo\"ise\xspace}
\newcommand{\Abelard}{Abelard\xspace}

\renewcommand{\phi}{\varphi}

\newcommand{\VE}{V_\Ei}
\newcommand{\VA}{V_\Ai}
\newcommand{\Ei}{\mathrm{E}}
\newcommand{\Ai}{\mathrm{A}}

\newcommand{\strat}{\phi}
\newcommand{\arena}{\mathcal{G}}

\newcommand{\game}{\mathbb{G}}
\newcommand{\pgame}{\mathbb{G}}

\newcommand{\play}{\lambda}
\newcommand{\WC}{\Omega}

\newcommand{\col}{\rho}
\newcommand{\colors}{C}
\newcommand{\exptime}{\textsc{ExpTime}}

\newcommand{\eat}[1]{}

\newcommand\mklksk[2]{\hbox{\tt{[}} #2 \nd(#1){\hbox{\tt{]}}}}

\newcommand{\transgraph}[1]{\mathrm{Graph}(#1)}

\renewcommand{\path}{\pplay}

\newcommand{\langweps}[1]{L(#1)}

\newcommand{\silent}{\mathfrak{e}}
\newcommand{\racine}{\mathfrak{r}}
\newcommand{\era}[1]{\overset{#1}{\longrightarrow}}

\newcommand{\eRa}[1]{\overset{#1}{\Longrightarrow}}

\newcommand{\unfold}[1]{\mathrm{Tree}(#1)}
\newcommand{\truth}[2]{|\!|#1|\!|_{#2}}
\renewcommand{\truth}[2]{[\![#2]\!]_{#1}}
\newcommand{\marked}[1]{\underline{#1}}

\newcommand{\treemarkednode}[2]{#1_{#2}}

\newcommand{\Acc}{\mathrm{Acc}}
\newcommand{\arobase}{{\symbol{64}}}

\definecolor{LimeGreen}{rgb}{0.2, 0.8, 0.2}
\definecolor{Crimson}{rgb}{0.86, 0.08, 0.24}
\definecolor{Orange}{rgb}{0.98, 0.6, 0.01}
\definecolor{Blue}{rgb}{0, 0.5, 1}

\begin{document}

\theoremstyle{acmdefinition}
\newtheorem{remark}[theorem]{Remark}
\newtheorem{property}[theorem]{Property}

\title{Higher-Order Recursion Schemes and Collapsible Pushdown Automata: Logical Properties}

\author{Christopher~H. Broadbent}
\affiliation{%
  \institution{Department of Computer Science, University of Oxford}
  \streetaddress{Wolfson Building, Parks Road}
  \city{Oxford}
  \country{UK}
}
\email{chbroadbent@gmail.com}

\author{Arnaud Carayol}
\affiliation{%
  \institution{LIGM, Univ Gustave Eiffel, CNRS}
  \streetaddress{5 boulevard Descartes — Champs sur Marne}
  \city{Marne-la-Vallée Cedex 2}
  \postcode{77454}
  \country{France}
}
\email{Arnaud.Carayol@univ-mlv.fr}

\author{C.-H. Luke Ong}
\affiliation{%
  \institution{Department of Computer Science, University of Oxford}
  \streetaddress{Wolfson Building, Parks Road}
  \city{Oxford}
  \country{UK}
}
\email{Luke.Ong@cs.ox.ac.uk}

\author{Olivier Serre}
\orcid{0000-0001-5936-240X}
\affiliation{%
  \institution{Université de Paris, IRIF, CNRS}
  \streetaddress{Bâtiment Sophie Germain, Case courrier 7014, 
8 Place Aurélie Nemours}
  \city{Paris Cedex 13}
  \postcode{75205}
  \country{France}}
\email{Olivier.Serre@cnrs.fr}

\begin{abstract}
This paper studies the logical properties of a very general class of infinite ranked trees, namely those generated by higher-order recursion schemes. We consider, for both monadic second-order logic and modal $\mu$-calculus, three main problems: model-checking, logical reflection (aka global model-checking, that asks for a finite description of the set of elements for which a formula holds) and selection (that asks, if exists, for some finite description of a set of elements for which an MSO formula with a second-order free variable holds). For each of these problems we provide an effective solution. This is obtained thanks to a known connection between higher-order recursion schemes and collapsible pushdown automata and on previous work regarding parity games played on transition graphs of collapsible pushdown automata.
\end{abstract}

\acmSubmissionID{}

%
%

%
%


\maketitle

\renewcommand{\shortauthors}{Broadbent et al.}

\section{Introduction}

In this paper we study the logical properties of a very general class of infinite ranked trees, namely those generated by higher-order recursion schemes (equivalently by collapsible pushdown automata). We consider three main problems —~model-checking, logical refection (aka global model-checking) and selection~— for both monadic second-order logic and modal $\mu$-calculus.

\subsection*{Infinite Trees with a Decidable MSO Theory}

A fundamental result of Rabin states that, for any formula expressible in monadic second-order (MSO) logic, one can decide whether it holds in the infinite complete binary tree~\cite{Rabin69}: in other words, the MSO model-checking problem is decidable. Since then, extending this result has been an important field of research. There are three main possible directions for that: the first one is to enrich MSO logic while preserving decidability; the second one is to look for structures richer than the infinite complete binary tree with a decidable MSO theory; the third one is to consider questions subsuming the model-checking problem. In this paper we are following the second and third directions.


Recursion schemes, an old model of computation, were originally designed as a canonical programming calculus for studying program transformation and control structures. In recent years, \emph{higher-order recursion schemes} have received much attention as a method of constructing rich and robust classes of possibly infinite ranked trees with strong algorithmic properties; theses are essentially finite typed deterministic term rewriting systems that generate when one applies the rewriting rules \emph{ad infinitum} an infinite tree. The interest was sparked by the discovery of Knapik, Niwi\'nski and Urzyczyn~\cite{KNU02} that recursion schemes which satisfy a syntactic constraint called \emph{safety} generate the same class of trees as higher-order pushdown automata. Remarkably these trees have decidable monadic second-order theories, subsuming earlier well-known MSO decidability results for regular (or order-0) trees \cite{Rabin69} and algebraic (or order-1) trees \cite{Courcelle95}. 

An alternative approach was developed by Caucal who introduced in \cite{Caucal02} two infinite hierarchies, one made of infinite trees and the other made of infinite graphs, defined by means of two simple transformations: {unfolding, which goes from graphs to trees, and inverse rational mapping ({or MSO-interpretation \cite{Carayol03}}), which goes from trees to graphs}. He showed that the tree hierarchy coincides with the trees generated by safe schemes, and as both unfolding and MSO-interpretation preserve MSO decidability, it follows that structures in those hierarchies have MSO decidable theories.

A major step was obtained by Ong who proved in ~\cite{Ong06a} that the modal $\mu$-calculus (local) model checking problem for trees generated by \emph{arbitrary} order-$n$ recursion schemes is $n$-\exptime-complete (hence these trees have decidable MSO theories). Note that this result was obtained using tools from innocent game semantics (in the sense of Hyland and Ong \cite{HO00}) and in particular does not rely on an equivalent automata model for generating trees.

Finding a class of automata that characterises the expressivity of higher-order recursion schemes was left open. Indeed, the results of Damm and Goerdt~\cite{DG86}, as well as those of Knapik \emph{et
  al.}~\cite{KNU01,KNU02} may only be viewed as attempts to answer the question as they have both had to impose the same syntactic constraints on recursion schemes, called of \emph{derived types} and \emph{safety} respectively, in order to establish their results. 
A partial answer was later obtained by Knapik, Niwi\'nski, Urzyczyn and Walukiewicz 
who proved that order-2 homogeneously-typed (but not necessarily safe) recursion
schemes are equi-expressive with a variant class of order-2 pushdown
automata called \emph{panic automata}~\cite{KNUW05}. 
Finally, Hague, Murawski, Ong and Serre gave a complete answer to the question in \cite{HMOS08,HMOS17} (also see \cite{CS08}). They introduced a new kind of
higher-order pushdown automata, which generalises \emph{pushdown automata with links} \cite{AdMO05a}, or equivalently panic automata, to all finite orders, called \emph{collapsible pushdown automata} (CPDA), in which every symbol in the stack has a link to a
(necessarily lower-ordered) stack situated somewhere below it.
A major result of their paper is that for every $n \geq 0$, order-$n$ recursion schemes and order-$n$ CPDA
are equi-expressive as generators of trees.

\subsection*{Main Results}

The equi-expressivity of higher-order recursion schemes and collapsible pushdown automata, as well as the connection between logic and two-player perfect-information parity games (see \eg \cite{Thomas97,Wilke2001,Walukiewicz04}), provide a roadmap to study logical properties of trees generated by recursion schemes: study collapsible pushdown games (\ie parity games played on transition graphs of CPDA) and derive logical consequences on trees generated by recursion schemes.
The companion paper~\cite{BCHMOS20} gives an in-depth study of collapsible pushdown parity games (following a series of papers~\cite{HMOS08,BCOS10,CS12} by the authors) on top of which we build in the present paper. 

Our first straightforward contribution is to note that the decidability of the model-checking problem for MSO (equivalently $\mu$-calculus) against trees generated by recursion schemes is an immediate consequence of the decidability of collapsible pushdown parity games and the equi-expressivity theorem

We then turn to the global version of the model-checking problem. Let $\mathcal{T}$ be a class of finitely-presentable infinite structures (such as trees or graphs) and $\mathcal{L}$ be a logical language for describing correctness properties of these structures. The \emph{global model checking problem} asks, given $t \in \mathcal{T}$ and $\phi \in \mathcal L$, whether 
the set $\truth{t}{\phi}$ of nodes defined by $\phi$ and $t$ is finitely describable,
and if so, whether it is effective.

An innovation of our work is a new approach to global model checking, by “internalising” the semantics $\truth{t}{\phi}$. Let $\phi \in \mathcal L$, and $\mathcal{S}$ be a recursion scheme over a ranked alphabet $\Sigma$ (i.e.~the node labels of $\mng{\mathcal{S}}$, the tree generated by $\mathcal{S}$, are elements of the ranked alphabet $\Sigma$). We say that $\mathcal{S}_\phi$, which is a recursion scheme over $\Sigma \cup \marked{\Sigma}$ (where $\marked{\Sigma}$ consists of a marked copy of each $\Sigma$-symbol), is a \emph{$\phi$-reflection}\footnote{In programming languages, \emph{reflection} is the process by which a computer program can observe and dynamically modify its own structure and behaviour.}  of $\mathcal{S}$ just if $\mathcal{S}$ and $\mathcal{S}_\phi$ generate the same underlying tree; further, suppose a node $u$ of $\mng{\mathcal{S}}$ has label $f$, then the label of the node $u$ of $\mng{\mathcal{S}_\phi}$ is $\marked{f}$ if $u$ in $\mng{\mathcal{S}}$ satisfies $\phi$, and it is $f$ otherwise. Equivalently we can think of $\mng{\mathcal{S}_\phi}$ as the tree that is obtained from $\mng{\mathcal{S}}$ by distinguishing the nodes that satisfy $\phi$. Our second contribution is the result that higher-order recursion schemes are (constructively) \emph{reflective} with respect to~the modal $\mu$-calculus (Theorem~\ref{theo:reflection-mu-calculus}). \Ie we give an algorithm that, given a modal $\mu$-calculus formula $\phi$, transforms a recursion scheme to its $\phi$-reflection.

While modal $\mu$-calculus and MSO are equivalent for expressing properties of a tree at its root, it is no longer true at other nodes (see \eg \cite{JW96}). Hence, it is natural to ask whether higher-order recursion schemes are reflective with respect to MSO logic. We answer positively (Theorem~\ref{theo:reflection-MSO}) this question by relying on the previous result for $\mu$-calculus.

We derive two consequences of the MSO reflection. The first one (Corollary~\ref{cor:divergent}) is to show how MSO reflection can be used to construct, starting from a scheme that may have non-productive rules, an equivalent one that does not have such divergent computations. The second application consists in proving (Theorem~\ref{theorem:ALaCaucal}) that the class of trees generated by recursion schemes is closed under the operation of MSO interpretation followed by tree unfolding hence, providing a result in the same flavour as the one obtained by Caucal for safe schemes in~\cite{Caucal02}.

Our third main contribution is to consider a more general problem than (both local and global) model-checking, namely the \emph{MSO selection property}. More precisely, we prove (Theorem~\ref{theo:selection}) that if $\mathcal{S}$ is a recursion scheme generating a tree $t$ satisfying a formula of the form $\exists X\phi(X)$ (where $X$ is a second-order free variable ranging over sets of nodes) then one can build another scheme that generates the tree $t$ where a set of nodes $U$ satisfying $\phi(X)$ is marked. This result is in fact quite surprising as it is known from~\cite{GurevichS83a,CarayolL07,CarayolHDR} that there exists a tree generated by an order-$3$ (safe) recursion scheme for which no MSO choice function exists, and that the selection property is closely connected to choice functions.

Note that most of the above mentioned results where previously presented by the authors in two papers at the LiCS conference~\cite{BCOS10,CS12} and that the current paper gives a unify and complete presentation of their proofs.

\subsection*{Related Work}

We already discussed the previous work on MSO model-checking again regular trees~\cite{Rabin69}, algebraic trees~\cite{Courcelle95}, trees generated by safe recursion schemes~\cite{KNU02,Caucal02}, trees generated by possibly unsafe order-2 recursion schemes~\cite{KNUW05} and general recursion schemes~\cite{Ong06a}. The proof presented in the present paper (\ie going through the connection with collapsible pushdown games) was first presented in \cite{HMOS08}. Following initial ideas in \cite{Aehlig06} and \cite{Kobayashi09}, Kobayashi and Ong gave yet another proof of Ong's decidability result: their proof \cite{KO09} consists in showing that, given a recursion scheme and an MSO formula, one can construct an intersection type system such that the scheme is typable in the type system if and only if the property is satisfied by the scheme; typability is then reduced to solving a parity game.

Piterman and Vardi \cite{PV04} studied the global model checking problem for regular trees and prefix-recognisable graphs using two-way alternating parity tree automata. Extending their results, Carayol et al.~\cite{CHMOS08} showed that the winning regions of parity games played over the transition graphs of higher-order pushdown automata (a strict subclass of CPDA) are regular. Later, using game semantics, Broadbent and Ong \cite{BO09} showed that for every order-$n$ recursion scheme $\mathcal{S}$, the set of nodes in $\mng{\mathcal{S}}$ that are definable by a given modal $\mu$-calculus formula is recognisable by an order-$n$ (non-deterministic) collapsible pushdown word automaton. The result we prove here (previously presented in~\cite{BCOS10}) is stronger as $\mu$-calculus reflection implies that the nodes are recognisable by a \emph{deterministic} CPDA.

The MSO selection property was first established in~\cite{CS12}. Alternative proofs were later given by Haddad~\cite{Haddad13,HaddadPHD}, and by Grellois and Melliès in~\cite{GrelloisM15}. 
Both proofs are very different from the one we give here. Indeed, our proof uses the equi-expressivity theorem to restate the problem as a question on CPDA, and a drawback of this approach is that once the answer is given on the CPDA side one needs to go back to the scheme side, which is not complicated but yields a scheme that is very different from the original one. The advantage of the approaches in \cite{HaddadPHD} (built on top of the intersection types approach by Kobayashi and Ong \cite{KO09}) and in \cite{GrelloisM15} (based on purely denotational arguments and connections with linear logic) is to work directly on the recursion scheme and to succeed to provide as a selector a scheme obtained from the original one by adding duplicated and annotated versions of the terminals.

In a recent work~\cite{Parys18}, Parys considered the logic WMSO+U, an extension of weak monadic second-order logic (\ie MSO logic where second-order quantification is limited to range on \emph{finite} sets) by the unbounding quantifier, expressing the fact that there exist arbitrarily large finite sets satisfying a given property its extension. This logic is incomparable with MSO logic. He showed that model-checking is decidable and that both reflection and selection hold for trees generated by recursion schemes.

\subsection*{Structure of This Paper}

The article is organised as follows. Section~\ref{section:Preliminaries} introduces the main concepts and some classical results. In Section~\ref{section:schemes} we introduce higher-order recursion schemes and in Section~\ref{section:CPDA} collapsible pushdown automata; we then give in Section~\ref{section:KnownResult} few known results that we build on in the rest of the paper. 
Section~\ref{section:modelChecking} briefly discusses the (local) model-checking problem. The global model-checking and the consequences the refection properties are respectively studied in Section~\ref{section:globalMC} and Section~\ref{section:globalMCConsequences}. Finally the selection problem is solved in Section~\ref{section:selection}.

\section{Preliminaries}\label{section:Preliminaries}

\subsection{Basic Notations}

When $f$ is a (partial) mapping, we let $\domain{f}$ denote its domain.

\subsection{Words}
An \defin{alphabet} $\Sigma$ is a (possibly infinite) set of letters. In the sequel $\Sigma^*$ denotes the set of \defin{finite words} over $\Sigma$, and $\Sigma^\omega$ the set of \defin{infinite words} over $\Sigma$. The empty word is written $\epsilon$ and the length of a word $u$ is denoted by $|u|$. Let $u$ be a finite word and $v$ be a (possibly infinite) word. Then $u\cdot v$ (or simply $uv$) denotes the concatenation of $u$ and $v$; the word $u$ is a prefix of $v$ iff there exists a word $w$ such that $v=u\cdot w$. A subset $X\subseteq \Sigma^*$ is \defin{prefix-closed} if, for every $v\in X$ one has $u\in X$ for any prefix $u$ of $v$.

\subsection{Trees}

Let $D$ be a finite set of \emph{directions}. A \concept{$D$-tree} is just a prefix-closed subset $T$ of $D^*$ whose elements are called \emph{nodes}. For a node $u\in T$, an element of the form $u\cdot d$ for some $d\in D$ is called the $d$-child (or simply a child if $d$ does not matter) of $u$. A node with no child is called a \emph{leaf} while the node $\epsilon$ is the \emph{root} of $T$.

Let $\Sigma$ be a finite alphabet. A \concept{$\Sigma$-labelled tree} is a function $t : \dom{t} \rightarrow \Sigma$ such that $\dom{t}$ is a $D$-tree for some set of directions $D$; for a node node $u \in \dom{t}$, we refer to $t(u)$ as the \emph{label} of $u$ in $t$.

If $\Sigma$ is a \concept{ranked alphabet} \ie each $\Sigma$-symbol $a$ has an arity $\arity{a} \geq 0$, a \concept{$\Sigma$-labelled ranked and ordered} tree (or simply a $\Sigma$-labelled tree if the context is clear) $t: \dom{t} \rightarrow \Sigma$ is a $\Sigma$-labelled tree such that the following holds (meaning that the label of a node determines its number of children):
\begin{itemize}
	\item $\dom{t}$ is a $\{1,\dots,m\}$-tree where $m=\max\{\arity{a}\mid a \in A\}$;
	\item for every node $u\in \dom{t}$, $\{i\mid 1\leq i\leq m \text{ and } u\cdot i\in \dom{t}\} = \{1,\dots,\arity{t(u)}\}$.
\end{itemize}

We write $\tree{\Sigma}$ for the set of (finite and infinite) $\Sigma$-labelled trees.

\subsection{Graphs}\label{subsection:graphs}

Let $A$ be a finite alphabet containing a distinguished symbol $\silent$ standing for silent transition; we let $A_{\silent}=A\setminus\{\silent\}$. An \defin{$A$-labelled graph} is a pair $G=(V,E)$ where $V$ is a set of vertices and $E\subseteq V\times A\times V$ is a set of  edges. For any $(u,a,v)\in E$ we write $u \era{a} v$ and we refer to it as an $a$-edge (\resp silent edge if $a=\silent$) with source $u$ and target $v$. Moreover, we require that for all $u\in V$, if $u$ is the source of a silent transition then $u$ is not the source of any $a$-transition with $a\neq \silent$. 

For a word $w = a_1 \cdots a_n \in A^*$, we define a binary relation $\era{w}$ on $V$ by letting $u \era{w} v$ if there exists a sequence $v_0, \ldots, v_n$ of elements in $V$ such that $v_0=u$, $v_n=v$, and for all $i \in [1,n]$, $v_{i-1} \era{a_{i}} v_{i}$. These definitions are extended to languages over $A$ by taking, for all $L \subseteq A^*$, the relation $\era{L}$ to be the union of all $\era{w}$ for $w \in L$.

For a word $w =a_1 \cdots a_k\in A_{\silent}^*$, we denote by $\eRa{w}$ the relation $\era{L_w}$ where $L_w = \silent^* a_1 \silent^* \cdots \silent^* a_k\silent^* $ is the set of words over $A$ obtained by inserting arbitrarily many occurrences of $\silent$ in $w$. 
  
The graph $G$ is said to be \defin{deterministic} if for all $u,v_1$ and $v_2$ in $V$ and all $a$ in $A$, if $u \era{a} v_1$ and $u \era{a} v_2$ then $v_1=v_2$. From now on we always assume that the graphs are deterministic.

Consider a deterministic $A$-labelled graph $G=(V,E)$ together with a distinguished vertex $r\in V$ called its \emph{root}. We associate with it a tree, denoted $\unfold{G}$, with directions in $A_{\silent}$, reflecting the possible behaviours in $G$ starting from the root. For this we let $$\unfold{G} = \{ w \in A_{\silent}^* \mid \exists v \in V,\, r \eRa{w}v \},$$ \ie
$\unfold{G}$ is obtained by unfolding $G$ from its root and contracting all $\silent$-transitions. 
Figure~\ref{fig:lts-trees} presents a deterministic $\{1,2,\silent\}$-labelled graph $G$ together with its associated tree.

In case $G=(V,E)$ is equipped with a vertex-labelling function $\rho:V\rightarrow \Sigma$ where $\Sigma$ is a finite alphabet, one can define a $\Sigma$-labelled tree from a pair $(G,r)$ by considering the tree $t:\unfold{G}\rightarrow \Sigma$ where $t(w)=\rho(v_w)$ where $v_w$ is the unique vertex in $G$ such that $r\eRa{w} v_w$ and that is not the source of an $\silent$-labelled edge{ other than an $\silent$-labelled loop}{ (recall that we assumed that a vertex that is the source of a silent transition cannot be the source of any $a$-transition with $a\neq \silent$)}. Note that, if $\Sigma$ is ranked and if $A_\silent=\{1,\cdots m\}$ with $m=\max\{\arity{a}\mid a\in \Sigma\}$, if we have $\{i\mid i\neq \silent \text{ and }\exists v',\ v\era{i}v'\}=\{1,\dots,\arity{\rho(v)}\}$ for every $v\in V$, then the tree $t$ is a $\Sigma$-labelled ranked and ordered tree. An example is given in Figure~\ref{fig:lts-trees}.

\begin{figure}
\begin{center}
\begin{tikzpicture}[scale=1,transform shape,level distance = 24pt,sibling distance = 14pt]
\tikzset{>=stealth}
\tikzset{edge from parent/.style={draw,->}}
\Tree [.\node(r){\textcolor{teal}{$r$}};
	\edge node[auto=right]{\small $1$};
	[.\node(s){\textcolor{teal}{$s$}}; 
	] 
	\edge node[auto=left]{\small $2$};
	[.{\textcolor{teal}{$t$}} \edge node[auto=right]{$\silent$}; 
		[.\node(u){\textcolor{teal}{$u$}};
			\edge[draw=none] ;[.\node{}; \edge[draw=none] ;[.\node{};]]
		] 
	]
	] 
]
\draw[->] (u) .. controls +(east:1) and +(north east:1.5) .. (r) 
node[pos=0.5,right]{\small $1$};
\path[->] (s) edge [loop below] node {$\silent$} (s);
\end{tikzpicture}
\hspace{1cm}
\begin{tikzpicture}[scale=1,transform shape,level distance = 25pt,sibling distance = 18pt]
\tikzset{>=stealth}
\tikzset{edge from parent/.style={draw,-}}
\Tree [.{$\epsilon$} 
	\edge;
	[.{$1$} 
	] 
	\edge;
	[.{$2$} 
		\edge;
		[.{$21$} 
			\edge;
			[.{$211$} ]		
			\edge;
			[.{$212$} 				
				\edge[dotted,-] ;
			[.{\phantom{$0$}} ] 
]		
		]
	] 
]
\end{tikzpicture}
\hspace{1cm}
\begin{tikzpicture}[scale=1,transform shape,level distance = 25pt,sibling distance = 18pt]
\tikzset{>=stealth}
\tikzset{edge from parent/.style={draw,-}}
\Tree [.{$\epsilon$ , \textcolor{red}{$a$}} 
	\edge;
	[.{$1$ , \textcolor{red}{$c$}} 
	] 
	\edge;
	[.{$2$ , \textcolor{red}{$b$}} 
		\edge;
		[.{$21$ , \textcolor{red}{$a$}} 
			\edge;
			[.{$211$ , \textcolor{red}{$c$}} ]		
			\edge;
			[.{$212$ , \textcolor{red}{$b$}} 				
				\edge[dotted,-] ;
			[.{\phantom{$0$}} ] 
]		
		]
	] 
]
\end{tikzpicture}
\end{center}
\caption{\label{fig:lts-trees} A deterministic $\{1,2,\silent\}$-labelled graph $G$ with root $r$ (on the left) together with its associated tree $\unfold{G}$ (in the middle) and its associated $\{a,b,c\}$-labelled ranked and order tree when one lets $\arity{a}=2$, $\arity{b}=1$, $\arity{c}=0$, $\rho(r)=a$, $\rho(s)=c$ and {$\rho(u)=b$} (on the right; node labels are written in red).}
\end{figure}
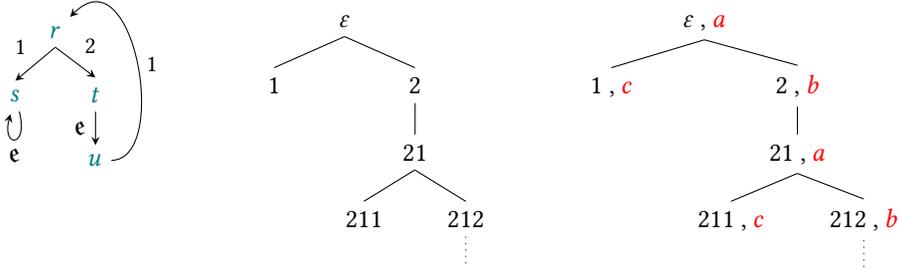

\subsection{Types}

\concept{Types} are generated by the grammar $\tau ::= o \; | \; \tau
\rightarrow \tau$.  Every type $\tau \not= o$ can be written uniquely as
$\tau_1 \rightarrow (\tau_2\rightarrow\cdots \rightarrow (\tau_n \rightarrow o)\cdots )$, for some $n \geq 1$ which is called its
\emph{arity}; the \emph{ground} type $o$ has arity 0.  We follow the convention that arrows associate to the right, and simply write $\tau_1 \rightarrow \tau_2\rightarrow\cdots \rightarrow \tau_n \rightarrow o$, which we sometimes abbreviate to $(\tau_1, \ldots, \tau_n, o)$.  
{The \concept{order} of a type measures the nesting depth on the left of $\to$. We define $\order{o} = 0$ and $\order{\tau_1 \rightarrow \tau_2} = \max (\order{\tau_1} + 1, \order{\tau_2})$. Thus $\order{\tau_1 \to \ldots \to \tau_n \to o} = 1 + \max\{\order{\tau_i}\mid 1\leq i\leq n\}$. For example, $\order{o \to o \to o \to o} = 1$ and $\order{((o \to o) \to o) \to o} = 3$.}

\subsection{Terms}

Let $\Upsilon$ be a set of typed symbols. {Let $f \in \Upsilon$ and $\tau$ be a type, we write $f : \tau$ to mean that $f$ has type $\tau$.} 

The set of \concept{(applicative) terms of type ${\tau}$ generated from ${\Upsilon}$}, written $\mathcal{T}_\tau(\Upsilon)$,
is defined by induction over the following rules. If $f : A$ is an
element of $\Upsilon$ then $f \in \mathcal{T}_\tau(\Upsilon)$; if $s \in \mathcal{T}_{\tau_1
  \rightarrow \tau_2}(\Upsilon)$ and $t \in \mathcal{T}_{\tau_1}(\Upsilon)$ then $s \; t \in {\mathcal{T}}_{\tau_2}(\Upsilon)$. For simplicity we write $\mathcal{T}(\Upsilon)$ to mean ${\mathcal{T}}_o(\Upsilon)$, the set of {terms of ground type}. {Let $t$ be a term, we write $t : \tau$ to mean that $t$ is an term of type $\tau$.} 
 
A {ranked alphabet} $\Sigma$ can be seen as a set of typed symbols: for each $\Sigma$-symbol $a$ its arity
$\arity{a} \geq 0$ determines its type $\underbrace{o
  \rightarrow \cdots \rightarrow o \rightarrow}_{\arity{a}}
o$; in particular every symbol has an order-0 or order-1 type. In this setting, terms in $\mathcal{T}(\Sigma)$ can be identified with $\Sigma$-labelled ranked and ordered trees.

\subsection{Logic}

We now give some brief background on logic (mainly for trees and graphs). More thorough introductions to the topic can be found in many textbooks and survey, \eg in \cite{EFT84,Thomas97}.

A \defin{relational structure} $\frak{A}=(D,R_1,\dots,R_k)$ is given by a (possibly \emph{infinite}) set $D$, called the \defin{domain} of $\frak{A}$, and \defin{relations} $R_i$ (if we let $r_i$ be the arity of relation $R_i$, then $R_i\subseteq D^{r_i}$). 

Many classical objects can be represented as relational structures. For instance, a $\Sigma$-labelled tree $t:\dom{t}\rightarrow \Sigma$ over a ranked alphabet $\Sigma$ whose symbols have arity at most $m$ corresponds to the relational structure $\frak{t}=(\dom{t},(S_i)_{1\leq i\leq m},\sqsubset,(Q_a)_{a\in \Sigma})$ where $S_i$ is the $i$-child relation defined by $S_i=\{(u,ui)\mid u,ui\in\dom{t}\}$, $\sqsubset$ is the strict prefix ordering, and $Q_a = t^{-1}(a)$ for each label $a\in \Sigma$.

In a similar manner a vertex- and edge-labelled graph can be represented as a relational structure: the domain coincides with the set of vertices and one has a binary relation for each edge label and a unary relation for each vertex label.

Properties of relational structures (here trees or graphs) can be expressed thanks to logical formalisms. We start with \defin{first-order logic (FO)}. First-order formulas on the relational structure $\frak{A}$ may use variables $x,y,\dots$ ranging over elements in the domain and are built up from atomic formulas of the form 
\begin{itemize}
\item $x=y$ where both $x$ and $y$ are variables, and 
\item $R(x_1,\dots,x_k)$ where $R$ is any relation in $\frak{A}$
\end{itemize}
by means of the usual Boolean connectives $\neg$, $\vee$, $\wedge$, $\Rightarrow$, $\Leftrightarrow$  and the quantifiers $\exists$ and $\forall$. The semantics is as expected and we do not give it here (see \eg \cite{EFT84}). 
A formula can contain \emph{free} variables, \ie variables that are not under the scope of a quantifier. In that case the semantics of the formula should be understood relatively to some interpretation of the free variables. 

\defin{Monadic second-order logic (MSO)} extends first-order logic by allowing second-order variables $X,Y,\dots$ which range over \emph{sets} of elements of the domain (\eg sets of nodes in trees or set of vertices in graphs). In the syntax of MSO formulas one can use the atomic formulas $x\in X$, where $x$ is a first-order variable and $X$ is a second-order variable, whose meaning is that element $x$ belongs to set $X$. 

We will also consider a fixpoint modal logic called \defin{(modal) $\mu$-calculus}. As we will write only few formulas from $\mu$-calculus and mostly rely on the fact that there exists a strong connection between $\mu$-calculus and parity games, we do not give any definition regarding to this logic and refer the reader to \cite{AN01,Wilke2001}.

For a given structure $\frak{A}$ and a formula $\phi$, one writes {$\frak{A}\models \phi$} to mean that $\phi$ is true in $\frak{A}$, and {$(\frak{A},p_1\dots,p_k)\models \phi(x_1,\dots,x_k)$} to mean that $\phi$ is true in $\frak{A}$ when one interprets free variables $x_i$ as $p_i$ for each $i=1,\dots,k$.

The \defin{local model-checking problem} is to decide for a given structure $\frak{A}$ and a formula $\phi$ without free variable, whether $\frak{A}\models \phi$ holds.

The \defin{global model-checking} is, for a given structure $\frak{A}$ and a formula $\phi(x)$ with a first-order free variable $x$, to provide a finite description of the set $\truth{\frak{A}}{\phi}=\{p\in D\mid (\frak{A},p)\models\phi(x)\}$. In case of $\mu$-calculus (that is interpreted in pointed relational structures, \ie relational structure together with  a distinguished root element), the set  $\truth{\frak{A}}{\phi}$ is intended to be the set of roots that make the formula $\phi$ true.

\subsection{Games}\label{def:paritygames}

An \defin{arena} is a triple $\arena=(G,\VE,\VA)$ where $G=(V,E)$ is a graph and $V=\VE\uplus\VA$ is a partition of the vertices among two players, \Eloise and \Abelard. For simplicity in the definitions, we assume that $G$ has no dead-end.

\Eloise and \Abelard play in $\arena$ by moving a pebble along edges. A \defin{play} from an initial vertex $v_0$ proceeds as follows: the player owning $v_0$ (\ie \Eloise if $v_0\in \VE$, \Abelard otherwise) moves the pebble to a vertex $v_1\in E(v_0)$. Then the player owning $v_1$ chooses a successor $v_2\in E(v_1)$ and so on. As we assumed that there is no dead-end, a play is an infinite word $v_0v_1v_2\cdots \in V^\omega$ such that for all $0\leq i$ one has $v_{i+1}\in E(v_i)$. A \defin{partial play} is a prefix of a play, \ie it is a finite word $v_0v_1\cdots v_\ell \in V^*$  such that for all $0\leq i<\ell$ one has $v_{i+1}\in E(v_i)$.

A \defin{strategy} for \Eloise is a function $\strat_\Ei:V^*V_\Ei\rightarrow V$ assigning, to every partial play ending in some vertex $v\in \VE$, a vertex $v'\in E(v)$. Strategies of \Abelard are defined likewise, and usually denoted $\strat_\Ai$.
In a given play $\play=v_0v_1\cdots$ we say that \Eloise (\resp \Abelard) \defin{respects a strategy} $\strat_\Ei$ (\resp $\strat_\Ai$) if whenever $v_i\in V_\Ei$ (\resp $v_i\in V_\Ai$) one has $v_{i+1} = \strat_\Ei(v_0\cdots v_i)$ (\resp $v_{i+1} = \strat_\Ai(v_0\cdots v_i)$).

A \defin{winning condition} is a subset $\WC\subseteq V^\omega$ and a (two-player perfect information) \defin{game} is a pair $\game=(\arena,\WC)$ consisting of an arena and a winning condition. A game is finite if it is played on a finite arena.

A play $\lambda$ is \defin{won} by \Eloise if and only if $\lambda\in \WC$; otherwise $\lambda$ is won by \Abelard. A strategy $\strat_\Ei$ is \defin{winning} for \Eloise in $\game$ from a vertex $v_0$ if any play starting from $v_0$ where \Eloise respects $\strat_\Ei$ is won by her. Finally a vertex $v_0$ is \defin{winning} for \Eloise in $\game$ if she has a winning strategy $\strat_\Ei$ from $v_0$. Winning strategies and winning vertices for \Abelard are defined likewise.

A \defin{parity winning condition} is defined by a \defin{colouring function} $\col$ that is a mapping $\col: V \rightarrow \colors \subset \mathbb{N}$ where $\colors$ is a \emph{finite} set of \defin{colours}. The parity winning condition associated with $\col$ is the set $\WC_\col = \{v_0 v_1 \cdots \in V^\omega \mid \liminf (\col(v_i))_{i \geq 0} \text{ is even}\}$, \ie a play is winning if and only if the smallest colour infinitely often visited is even.
A \defin{parity game} is a game of the form $\game=(\arena,\WC_\col)$ for some colouring function.

\subsection{Tree Automata}\label{section:treeAutomata}

We now introduce the usual model of automata to recognise languages of (possibly infinite) ranked trees.

Let $\Sigma$ be a ranked alphabet. A \defin{tree automaton} is a tuple $\mathcal{A}=(Q,\Sigma,q_0,\Delta,\col,\Acc)$ where $Q$ is a finite set of control states, $q_0\in Q$ is the initial state, $\Delta\subseteq \bigcup_{a\in\Sigma} Q\times \{a\}\times Q^{\arity{a}}$ is a transition relation, $\col:Q\rightarrow\colors\subset \mathbb{N}$ is a colouring function and $\Acc\subseteq Q\times \Sigma$ is an acceptance condition to handle leaves.

Let $t:\dom{t}\rightarrow\Sigma$ be a $\Sigma$-labelled tree. A \defin{run} of $\mathcal{A}$ over $t$ is a tree $r:\dom{t}\rightarrow Q$ such that the following holds:
\begin{itemize}
\item $r(\epsilon)=q_0$, \ie the root is labelled by the initial state;
\item for every node $u\in\dom{t}$ one has $(r(u),t(u),r(u\cdot 1),\cdots r(u\cdot\arity{t(u)}))\in\Delta$, \ie the local constraints imposed by the transition relation are respected.
\end{itemize}
The run $r$ is accepting if and only if the following two conditions are satisfied:
\begin{itemize}
	\item for every leaf $u\in\dom{t}$ one has $(r(u),t(u))\in\Acc$;
	\item for every infinite branch, the smallest colour of a state appearing infinitely often along it is even; formally if $u_0,u_1,u_2,\dots$ denotes the infinite sequence of nodes read along a branch one has that $\liminf (\col(r(u_i))_{i\geq 0}$ is even.
\end{itemize}

Finally, a tree is \defin{accepted} by $\mathcal{A}$ if and only if there is an accepting run over it, and we refer to the set of accepted trees as the language recognised by $\mathcal{A}$.

There are tight connections between model-checking MSO logic against ranked trees, solving parity games, and checking whether the language recognised by a tree automaton is empty. We refer the reader to~\cite{Thomas97} for details on those connections and we only recall here the ones we use in the present paper.


{
Let $\Sigma$ be a ranked alphabet, let $t$ be a $\Sigma$-labelled ranked tree and let $u_1,\cdots,u_k$ be nodes in $t$ for some $k\geq 0$. Then we write $t_{u_1,\cdots,u_k}$ for the $\Sigma\times\{0,1\}^k$-labelled ranked tree such that $\dom{t_{u_1,\cdots,u_k}}=\dom{t}$ and for every $u\in\dom{t}$,  $t_{u_1,\cdots,u_k}(u)=(t(u),(\iota_1(u),\dots,\iota_k(u)))$ where, for $i=1,\cdots,k$, one let $\iota_i(u)=1$ if $u=u_i$ and $\iota_i(u)=0$ otherwise. In other words $t_{u_1,\cdots,u_k}$ is obtained from $t$ by marking nodes $u_1,\dots,u_k$.

The first connection is the famous result by Rabin~\cite{Rabin69}. 

\begin{theorem}\label{th:Rabin69}
Let $\Sigma$ be a ranked alphabet. The following holds.\begin{itemize}
	\item For every MSO formula $\phi(x_1,\dots,x_k)$ with possibly first-order variables over $\Sigma$-labelled ranked trees, one can build an automaton $\mathcal{A}_{\phi(x_1,\dots,x_k)}$ such that the trees accepted by $\mathcal{A}_{\phi(x_1,\dots,x_k)}$ are exactly those trees $t_{u_1,\dots,u_k}$ such that $(t,u_1,\dots,u_k) \models \phi(x_1,\dots,x_k)$.
	\item For every tree automaton $\mathcal{A}$ one can build an MSO-formula $\phi_\mathcal{A}$ such that the trees accepted by $\mathcal{A}$ are exactly those where $\phi_\mathcal{A}$ holds. 
	\end{itemize}
\end{theorem}
}

The second connection used in this paper is the game-approach to the acceptance problem for tree automata. Let $\mathcal{A}=(Q,\Sigma,q_0,\Delta,\col,\Acc)$ be a tree automaton and let $t$ be a $\Sigma$-labelled tree. Consider the following (informal) parity game $\game_{\mathcal{A},t}$. The main vertices in game $\game_{\mathcal{A},t}$ are pairs $(u,q)$ made of a node $u$ in $t$ and a state $q$ in $Q$. In a node $(u,q)$, \Eloise picks a transition $(q,t(u),q_1,\dots,q_{\arity{t(u)}})$ and goes to an intermediate vertex where \Abelard picks some $i$ such that $1\leq i\leq \arity{t(u)}$, and the play proceeds then from $(u\cdot i,q_i)$. In case $\arity{t(u)}=0$ the play loops forever in $(u,q)$. The colouring function in $\game_{\mathcal{A},t}$ assigns colour $\col(q)$ to vertex $(u,q)$ if $\arity{t(u)}\neq 0$ and otherwise it assigns $0$ if $(q,t(u))\in\Acc$ and $1$ otherwise. Intermediate vertices controlled by \Abelard gets the maximal colour used in $\mathcal{A}$ (hence, have no impact on who wins a play). The following result is due to Gurevich and Harrington~\cite{GurevichH82}.

\begin{theorem}\label{th:GurevichH82}
	\Eloise has a winning strategy in the parity game $\game_{\mathcal{A},t}$ from $(\epsilon,q_0)$ if and only if the tree $t$ is accepted by $\mathcal{A}$.
\end{theorem}

\section{Recursion Schemes}\label{section:schemes}

For each type $\tau$, we assume an infinite set $\var_\tau$ of variables of
type $\tau$, such that $\var_{\tau_1}$ and $\var_{\tau_2}$ are disjoint whenever $\tau_1
\neq \tau_2$; and we write $\var$ for the union of $\var_\tau$ as $\tau$ ranges
over types. We use letters $x, y, \phi, \psi, \chi, \xi$ etc.~to range
over variables.

A (deterministic) \concept{recursion scheme} is a quadruple $\rscheme =
\anglebra{\Sigma, \mathcal{N}, \mathcal{R}, I}$ where
\begin{itemize}
\item $\Sigma$ is a ranked alphabet of \concept{terminals} (including a
  distinguished symbol $\bot : o$ that we shall omit when desciribing $\Sigma$)
\item $\mathcal{N}$ is a finite set of typed \concept{non-terminals}; we use
upper-case letters $F, H$, etc.~to range over non-terminals
\item $I \in \mathcal{N}$ is a distinguished \concept{initial symbol} of type $o$
\item $\mathcal{R}$ is a finite set of \concept{rewrite rules}, \emph{one for each}
non-terminal $F:(\tau_1, \ldots, \tau_n, o)$, of the form
\[ F \, \xi_1 \, \cdots \, \xi_n \; \rightarrow \; e \] where each
$\xi_i$ is a variable of type $\tau_i$, and $e$ is a term in
$\term{\Sigma \cup \mathcal{N} \cup \makeset{\xi_1, \cdots,
    \xi_n}}$. Note that the expressions on either side of the arrow
are terms of ground type.
\end{itemize}
The \concept{order} of a recursion scheme is defined to be the highest
order of (the types of) its non-terminals.

In this paper we use recursion schemes as generators of
$\Sigma$-labelled trees. Informally the \emph{value tree} {$\mng{\rscheme}$} of (or the tree \emph{generated} by) a recursion
scheme $\rscheme$ is a possibly infinite term (of
ground type), {constructed from the terminals in $\Sigma$}, that
is obtained, {starting from the initial symbol $I$,}
by unfolding the rewrite rules of ${\rscheme}$ \emph{ad infinitum},
replacing formal by actual parameters each time.

To define $\mng{\rscheme}$, we first introduce a map 
{$(\cdot)^\bot : \bigcup\limits_{A} \mathcal{T}_A({\Sigma \cup \mathcal{N}}) \longrightarrow  \bigcup\limits_{A: \order{A} \leq 1} \mathcal{T}_A({\Sigma})$}
that takes
a term and replaces each non-terminal, together with its
arguments, by $\bot$. We define $(\cdot)^\bot$ by structural recursion
as follows: we let $f$ range over $\Sigma$-symbols, and $F$ over
non-terminals in $\mathcal{N}$
\[\begin{array}{rll}
f^\bot & = & f \\ F^\bot & = & \bot \\ (st)^\bot & = & \left\{
\begin{array}{ll}
\bot & \hbox{if $s^\bot = \bot$}\\
(s^\bot t^\bot) & \hbox{otherwise.}
\end{array}
\right.
\end{array}\]

Clearly if $s \in \term{\Sigma \cup \mathcal{N}}$ is of ground type, so is
$s^\bot \in \term{\Sigma}$.

Next we define a {one-step reduction relation} $\rightarrow_{\rscheme}$
which is a binary relation over terms in $\term{\Sigma \cup \mathcal{N}}$.  Informally, $t \rightarrow_{\rscheme} t'$ just if $t'$ is obtained
from $t$ by replacing some occurrence of a non-terminal $F$ by the
right-hand side of its rewrite rule in which all formal parameters are
in turn replaced by their respective actual parameters, subject to the
proviso that the $F$ must occur at the head of a subterm of ground
type. Formally $\rightarrow_{\rscheme}$ is defined by induction over the
following rules:
\begin{itemize}
\item (\emph{Substitution}). $F t_1 \cdots t_n \rightarrow_{\rscheme} e[t_1 /
  \xi_1, \cdots, t_n / \xi_n]$ where $F \xi_1 \cdots \xi_n \rightarrow
  e$ is a rewrite rule of $G$.

\item (\emph{Context}). If $t \rightarrow_{\rscheme} t'$ then $(st)
  \rightarrow_{\rscheme} (st')$ and $(ts) \rightarrow_{\rscheme} (t's)$.
\end{itemize}

Note that $\tree{\Sigma}$ is a complete partial order with respect
to the approximation ordering $\sqsubseteq$ defined by: $t
\sqsubseteq t'$ just if $\dom{t} \subseteq \dom{t'}$ and for all 
$w \in \dom{t}$, we have $t(w) = \bot$ or $t(w) = t'(w)$. \Ie $t'$ is
obtained from $t$ by replacing some $\bot$-labelled nodes by
$\Sigma$-labelled trees. 
{If one views ${\rscheme}$ as a rewrite system, it is a consequence of the Church-Rosser
property~\cite{ChurchR1936}  that the set $\makeset{t^\bot \in \tree{\Sigma} :  \text{there is a finite reduction sequence  }S = t_0 \rightarrow_{\rscheme} \;
\cdots \; \rightarrow_{\rscheme} t_n = t}$ is directed. 
Hence, we can finally define the $\Sigma$-labelled
ranked tree $\mng{{\rscheme}}$, called the \concept{value tree} of (or
the tree \concept{generated} by) ${\rscheme}$:
\[\mng{{\rscheme}} \; = \; \sup \makeset{t^\bot \in \tree{\Sigma} :  \text{there is a finite reduction sequence  }S = t_0 \rightarrow_{\rscheme} \;
\cdots \; \rightarrow_{\rscheme} t_n = t}. \] }

\begin{example}\label{ex:Scheme:AnBnCnD}
Consider the order-2 recursion scheme $\mathcal{S}$ with non-terminals $I:o,\, F:(o\rightarrow o) \rightarrow (o\rightarrow o) \rightarrow o,\, C_p:(o\rightarrow o) \rightarrow (o\rightarrow o) \rightarrow o \rightarrow o$, variables $x: o,\, \phi,\psi: o\rightarrow o$, terminals $a,b,c,d$ of arity $2$, $1$, $1$ and $0$ respectively, and the following rewrite rules:

\[\left\{\begin{array}{rll}
I & \rightarrow & F \, b \, c\\ 
F \, \phi \, \psi & \rightarrow & a \, (F \, (C_p \, b \, \phi) \, (C_p \, c \, \psi)) \, (\phi \, (\psi \, d))\\
C_p \, \phi \, \psi \, x & \rightarrow & \phi \, (\psi \, x)
\end{array}\right.\]

The non-terminal $C_p$ is to be understood as a mechanism to compose its two first arguments and apply the result to the third argument. 

The first steps of rewriting of $\mathcal{S}$ are given in Figure~\ref{fig:ex-rewrittingAnBnCnD}.

\begin{figure}
\begin{center}
\begin{tikzpicture}[level distance = 30pt,sibling distance = 10pt,scale=0.8,transform shape]
\tikzstyle{every node}=[font=\normalsize]
\tikzset{>=stealth}
\tikzset{edge from parent/.style={draw,-}}
\node (n0) at (0,0) {$I$};
\node (n1) at (1.5,0){
	\Tree [.{ $F$} 
		[.{$b$} ] 
		[.{$c$} ] 
	]
};
\draw[->] (n0) -- (n1);
\node (n2) at (5,0){
	\Tree [.{ $a$} 
		[.{$F$} 			
			[.{$C_p$} [.{$b$} ] [.{$b$} ]] 
			[.{$C_p$} [.{$c$} ] [.{$c$} ]] 
		]
		[.{$b$} [.{$c$} [.{$d$} ]]]  
	]
};
\draw[->] (n1) -- (n2);

\node (n3) at (13,0){
	\Tree [.{ $a$} 
		[.{$a$} 
			[.{$F$} 			
				[.{$C_p$} [.{$b$} ] [.{$C_p$} [.{$b$} ] [.{$b$} ]]] 
				[.{$C_p$} [.{$c$} ] [.{$C_p$} [.{$c$} ] [.{$c$} ]]] 
			]
			[.{$C_p$} [.{$b$} ] [.{$b$} ] [.{$C_p$} [.{$c$} ] [.{$c$} ] [.{$d$} ]]] 
		]
		[.{$b$} [.{$c$} [.{$d$} ]]]  
	]
};
\draw[->] (n2) -- (n3);

\node (n4) at (5,-7){
	\Tree [.{ $a$} 
		[.{$a$} 
			[.{$F$} 			
				[.{$C_p$} [.{$b$} ] [.{$C_p$} [.{$b$} ] [.{$b$} ]]] 
				[.{$C_p$} [.{$c$} ] [.{$C_p$} [.{$c$} ] [.{$c$} ]]] 
			]
			[.{$b$} [.{$b$} [.{$C_p$} [.{$c$} ] [.{$c$} ] [.{$d$} ]]]] 
		]
		[.{$b$} [.{$c$} [.{$d$} ]]]  
	]
};

\node (n3b) at (0,-7){};

\draw[->] (n3b) -- (n4);

\node (n5) at (14,-7){
	\Tree [.{ $a$} 
		[.{$a$} 
			[.{$F$} 			
				[.{$C_p$} [.{$b$} ] [.{$C_p$} [.{$b$} ] [.{$b$} ]]] 
				[.{$C_p$} [.{$c$} ] [.{$C_p$} [.{$c$} ] [.{$c$} ]]] 
			]
			[.{$b$} [.{$b$} [.{$c$} [.{$c$} [.{$d$} ]]]]] 
		]
		[.{$b$} [.{$c$} [.{$d$} ]]]  
	]
};
\draw[->] (n4) -- (n5);
\end{tikzpicture}
\end{center}
\caption{First steps of rewriting of the recursion scheme from Example~\ref{ex:Scheme:AnBnCnD}.}\label{fig:ex-rewrittingAnBnCnD}
\end{figure}
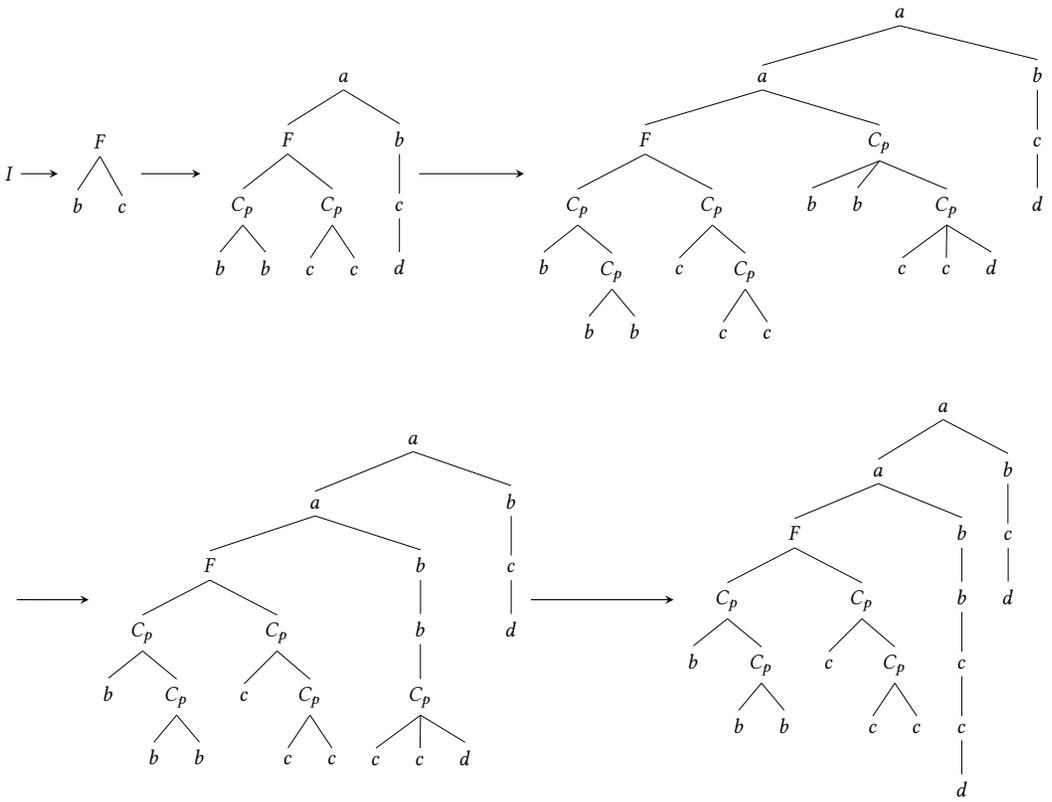

The value tree $t=\mng{\mathcal{S}}$ (depicted in Figure~\ref{fig:ex-rewrittingAnBnCnDTree}) has domain $\{0^k1^h\mid k\geq 0\text{ and } h\leq {2k+3}\}$, and is defined, for every $k\geq 0$ by $t(0^k)=a$, $t(0^k1^h) = b$ if $h\leq k+1$, $t(0^k1^h) = c$ if $k+2\leq h\leq 2k+2$ and $t(0^k1^h)=d$ if $h=2k+3$.

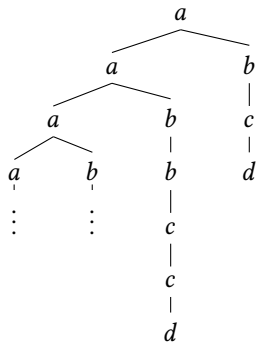
\begin{figure}
\begin{tikzpicture}[scale=1,transform shape,level distance = 20pt,sibling distance = 18pt]
\node (n) at (0,0){
	\Tree [.{ $a$} 
		[.{$a$} [.{$a$} [.{$a$} [.{$\vdots$} ]] [.{$b$} [.{$\vdots$} ]]] 
		[.{$b$} [.{$b$} [.{$c$} [.{$c$} [.{$d$} ]]] 			
		] ]]
		[.{$b$} [.{$c$} [.{$d$} ] 			
		] ]  
	]
};	
\end{tikzpicture}
\caption{The tree generated by the recursion scheme from Example~\ref{ex:Scheme:AnBnCnD} and by the CPDA from Example~\ref{Example:CPDAAnBnCndTree}.}\label{fig:ex-rewrittingAnBnCnDTree}
\end{figure}

\end{example}

\section{Collapsible Pushdown Automata}\label{section:CPDA}

\subsection{Stacks with Links and Their Operations}

Fix an alphabet $\Gamma$ of \concept{stack symbols} and a distinguished
\concept{bottom-of-stack symbol} $\bot \in \Gamma$. An \concept{order-$0$
  stack} (or simply \concept{$0$-stack}) is just a stack symbol. An
\concept{order-${(n+1)}$ stack} (or simply \concept{${(n+1)}$-stack}) $\stack$ is
a non-null sequence, written $\mksk{\stack_1 \cdots \stack_l}$, of $n$-stacks
such that every non-$\bot$ $\Gamma$-symbol $\gamma$ that occurs in $\stack$ has
a \emph{link} to a stack of some order $e$ (say, where $0 \leq e \leq n$)
situated below it in $\stack$; we call the link an \concept{${(e+1)}$-link}. The
\concept{order} of a stack $\stack$ is written $\order{\stack}$. 
The \concept{height} of a stack $\mksk{\stack_1 \cdots \stack_{l}}$ is defined as $l$.

{
\begin{remark}
One way to give a formal semantics of the stack operations is to work with appropriate numeric representations of the links as explained in \cite[Section~3.2]{HMOS17}. We believe that the informal presentation should be sufficient for this work and hence refer the reader to \cite{HMOS17} for a formal definition of stacks.
\end{remark}
}

As usual, the bottom-of-stack symbol $\bot$ cannot be popped from or
pushed onto a stack. Thus we require an \concept{order-1 stack} to be a
non-null sequence $\mksk{\gamma_1 \cdots \gamma_l}$ of elements of $\Gamma$ such
that for all $1 \leq i \leq l$, $\gamma_i = \bot$ iff $i = 1$. We inductively define
$\bot_k$, the \concept{empty ${k}$-stack}, as follows: $\bot_0 = \bot$
and $\bot_{k+1} = \mksk{\bot_k}$.

We first define the operations $\popn{i}$ and $\topn{i}$ with $i \geq
1$: $\topn{i}(\stack)$ returns the top $(i-1)$-stack of $\stack$, and
$\popn{i}(\stack)$ returns $\stack$ with its top $(i-1)$-stack
removed. Precisely let $\stack = \mksk{\stack_1 \cdots \stack_{l+1}}$ be a stack with
$1 \leq i \leq \order{\stack}$:
\[\begin{array}{rll}
\topn{i}(\underbrace{\mksk{\stack_1 \cdots \stack_{l+1}}}_{\hbox{$\stack$}}) & = &
\left\{\begin{array}{ll} \stack_{l+1} & \hbox{if $i = \order{\stack}$}\\
    \topn{i} (\stack_{l+1}) \quad & \hbox{if $i < \order{\stack}$}
\end{array}\right.\\
\popn{i}(\underbrace{\mksk{\stack_1 \cdots \stack_{l+1}}}_{\hbox{$\stack$}}) & = &  
\left\{\begin{array}{ll}
\mksk{\stack_1 \cdots \stack_l} & \hbox{if $i = \order{\stack}$ and $l \geq 1$}\\
\mksk{\stack_1 \cdots \stack_l \,
\popn{i}(\stack_{l+1})} \quad & \hbox{if $i < \order{\stack}$}
\end{array}\right.\\
\end{array}\]
By abuse of notation, we set $\topn{\order{\stack}+1}(\stack) = \stack$. Note that
$\popn{i}(\stack)$ is undefined if $\topn{i+1}(\stack)$ is a one-element
$i$-stack. For example $\popn{2}(\mksk{\mksk{\bot \, \alpha \, \beta}})$ and $\popn{1}(\mksk{\mksk{\bot \, \alpha \, \beta}\mksk{\bot}})$ are both undefined.

There are two kinds of $\mathit{push}$ operations. {We start with the
\emph{order-$1$} $\mathit{push}$}. Let $\gamma$ be a non-$\bot$ stack symbol and
$1 \leq e \leq \order{\stack}$, we define a new stack operation
$\pushlk{e}{\gamma}$ that, when applied to $\stack$, first attaches a link from
$\gamma$ to the $(e-1)$-stack \emph{immediately} below the top
$(e-1)$-stack of $\stack$, then pushes $\gamma$ (with its link) onto the top
1-stack of $\stack$. Formally for $1 \leq e \leq \order{\stack}$ and $\gamma \in
(\Gamma \setminus \makeset{\bot})$, we define
\[ \pushlk{e}{\gamma}( 
\underbrace{\mksk{\stack_1 \cdots \stack_{l+1}}}_{\hbox{$\stack$}}) = 
\left\{
\begin{array}{ll}
\mksk{\stack_1 \cdots \stack_l \, \pushlk{e}{\gamma}(\stack_{l+1})} \quad &
\hbox{if $e < \order{\stack}$}\\
\mksk{\stack_1 \cdots \stack_l \, \stack_{l+1} \, \gamma^\dag} & \hbox{if
$e = \order{\stack} = 1$}\\
\mksk{\stack_1 \cdots \stack_l \, \pushone{{\widehat{\gamma}}}(\stack_{l+1})} & \hbox{if
$e = \order{\stack} \geq 2$ and $l \geq 1$}\\
\end{array} 
\right.  \]
where 
\begin{itemize}
\item $\gamma^\dag$ denotes the symbol $\gamma$ with a link to the
0-stack $\stack_{l+1}$
\item $\widehat{\gamma}$ denotes the symbol $\gamma$ with a link to the
  $(e-1)$-stack $\stack_l$; and we define
\[\pushone{\widehat{\gamma}}(\underbrace{\mksk{\varstack_1 \cdots \varstack_{r+1}}}_{\hbox{$\varstack$}}) = 
\left\{
\begin{array}{ll}
\mksk{\varstack_1 \cdots \varstack_r \,\pushone{\widehat{\gamma}}(\varstack_{r+1})} \quad &
\hbox{if $\order{\varstack} > 1$}\\
\mksk{\varstack_1 \cdots \varstack_{r+1} \, \widehat{\gamma} } & \hbox{otherwise \ie~$\order{\varstack} = 1$}\\
\end{array}
\right.\] 
\end{itemize}

The higher-order $\pushn{j}$, where $j \geq 2$, simply duplicates the
top $(j-1)$-stack of $\stack$. Precisely, let $\stack = \mksk{\stack_1 \cdots
  \stack_{l+1}}$ be a stack with $2 \leq j \leq \order{\stack}$:
\[\begin{array}{lll}
  \pushn{j}(\underbrace{\mksk{\stack_1 \cdots \stack_{l+1}}}_{\hbox{$\stack$}}) & = & \left\{\begin{array}{ll}
      \mksk{\stack_1 \cdots \stack_{l+1} \, \stack_{l+1}} & \hbox{if $j = \order{\stack}$}\\
      \mksk{\stack_1 \cdots \stack_l \, \pushn{j} (\stack_{l+1})} \quad & \hbox{if $j < \order{\stack}$}
\end{array}\right.\\
\end{array}\]
In case $j = \order{\stack}$ above, the link structure of $\stack_{l+1}$ 
is preserved by the copy that is pushed on top by $\pushn{j}$.

We also define, for any stack symbol $\gamma$ an operation on stacks that rewrites the topmost stack symbol \emph{without modifying} its link. 
Formally: 
\[\begin{array}{lll}
  \toprew{\gamma} \, \underbrace{\mksk{s_1 \cdots s_{l+1}}}_{\hbox{$s$}} & = & \left\{\begin{array}{ll}
      \mksk{s_1 \cdots s_l \, \toprew{\gamma} s_{l+1}} \quad & \hbox{if $\order{s}>1$}\\
      \mksk{s_1 \cdots s_{l} \, \widehat{\gamma}} & \hbox{if $\order{s}=1$ and $l\geq 1$}
\end{array}\right.\\
\end{array}\]
where $\widehat{\gamma}$ denotes the symbol $\gamma$ with a link to the same target as the link from $s_{l+1}$. Note that $\toprew{\gamma}(\stack)$ is undefined if either $\topn{2}(s)$ or $s$ is the empty $1$-stack.

Finally there is an important operation called $\collapse$. We say
that the $n$-stack $\stack_0$ is a \concept{prefix} of an $n$-stack $\stack$,
written $\stack_0 \leq \stack$, just in case $\stack_0$ can be obtained from $\stack$ by a
sequence of (possibly higher-order) ${\mathit pop}$ operations.
Take an $n$-stack $\stack$ where $\stack_0 \leq \stack$, for some $n$-stack $\stack_0$, and $\topone(\stack)$ has a link to $\topn{e}(\stack_0)$. Then $\collapse \; \stack$ is defined to be $\stack_0$.

 \begin{example}\rm\label{eg:3stack}
To avoid clutter, when displaying $n$-stacks in examples, we shall omit 1-links (indeed by construction they can only point to the symbol directly below), writing e.g.~$\mksk{\mksk{\bot} \mksk{\bot \alpha \, \beta}}$ instead of 
$\pstr[.1cm][1pt]{
\mksk{
\mksk{\bot}
\mksk{
\nd(n1){\bot}\,\,
\;
\nd(n2-n1,50){\alpha}\,\,
\;
\nd(n3-n2,50){\beta}
}
}
}
$.

    Take the 3-stack $\stack = \mksk{\mksk{\mksk{ \, \bot \, \alpha}} \;
      \mksk{\mksk{ \, \bot} \mksk{ \, \bot \, \alpha}}}$. We have
    \[
       \begin{array}{rll}
         \pushlk{2}{\gamma}(s)  & = &  \pstr[.1cm]{\mksk{\mksk{\mksk{ \, \bot \, \alpha}} \;
                                            \mksk{\mklksk{n1}{ \, \bot}
                                                  \mksk{ \, \bot \, \alpha \, \nd(n2-n1){\gamma}}}}} \\

\collapse \, (\pushlk{2}{\gamma}(s)) & = & \mksk{ \mksk{\mksk{ \, \bot\, \alpha}} \; \mksk{\mksk{ \, \bot}}}
\\
         \underbrace{ \pushlk{3}{\gamma}(\toprew{\beta} ( \pushlk{2}{\gamma}(s)))}_\theta
                              &= & \pstr[.75cm]{\mksk{\mklksk{n1}{\mksk{ \, \bot \, \alpha}} \;
                                             \mksk{\mklksk{n2}{ \, \bot}
                                                   \mksk{ \, \bot \, \alpha \, \nd(n3-n2){\beta}\, \nd(n4-n1){\gamma}}}}}.
       \end{array}
    \]
    Then $\pushn{2} (\theta)$ and $\toprew{\alpha}(\pushn{3}(\theta))$ are respectively
    \[
        \begin{array}{c}
          \pstr[.6cm]{\mksk{\mklksk{n1}{\mksk{ \, \bot \, \alpha}} \;
                            \mksk{\mklksk{n2}{ \, \bot}
                                  \mksk{ \, \bot \, \alpha \, \nd(n3-n2){\beta} \, \nd(n4-n1){\gamma}}
                                  \mksk{ \, \bot \, \alpha \, \nd(n5-n2){\beta} \, \nd(n6-n1){\gamma}}}}}
          \;\hbox{and} \\

          \pstr[1cm]{\mksk{\mklksk{n1}{\mksk{ \, \bot \, \alpha}} \;
                     \mksk{\mklksk{n2}{ \, \bot}
                           \mksk{ \, \bot \, \alpha \, \nd(n3-n2,40){\beta} \, \nd(n4-n1,39){\gamma}}} \;
                     \mksk{\mklksk{n5}{}
                           \mksk{ \, \bot \, \alpha \, \nd(n6-n5,40){\beta} \, \nd(n7-n1,37){\alpha}}}}}.
        \end{array}
    \]

    We have $\collapse \, (\pushn{2}( \theta)) = \collapse \, (\toprew{\alpha}(\pushn{3}(\theta))) = \collapse( \theta)  = \mksk{\mksk{\mksk{ \, \bot \, \alpha}}}$.
  \end{example}

The set $\Op{n}{\Gamma}$ of order-$n$ CPDA \concept{stack operations} over stack alphabet $\Gamma$ (or simply $\Op{n}{}$ if $\Gamma$ is clear from the context) comprises six types of operations:

\begin{enumerate}
\item $\popn{k}$ for each $1 \leq k \leq n$,
\item $\pushn{j}$ for each $2 \leq j \leq n$, 
\item $\pushlk{e}{\gamma}$ for each $1 \leq e \leq n$ and each $\gamma \in (\Gamma
  \setminus \makeset{\bot})$,
\item $\toprew{\gamma}$ for each $\gamma \in (\Gamma
  \setminus \makeset{\bot})$,
\item $\collapse$, and
\item $\id$ for the identity operation (\ie $id(\stack)=\stack$ for all stack $\stack$).
\end{enumerate}

\subsection{Collapsible Pushdown Automata (CPDA)}

An \defin{order-$n$ (deterministic) collapsible pushdown automaton with input} ($n$-CPDA) is a 6-tuple 
$\anglebra{A, \Gamma, Q,\delta, q_0,F}$ where 
$A$ is an input alphabet containing a distinguished symbol $\silent$ standing for silent transition,
$\Gamma$ is a stack alphabet, 
$Q$ is a finite set of states, 
$q_0\in Q$ is the initial state, 
$F\subseteq Q$ is the set of final states
and $\delta \, : \, Q \times \Gamma \times A\,  \rightarrow \, Q \times \mathit{Op}_n$ is a transition (partial) function such that, for all $q\in Q$ and $\gamma \in \Gamma$,  if $\delta(q,\gamma,\silent)$ is defined then for all $a\in A$, $\delta(q,\gamma,a)$ is undefined (i.e.~if some silent transition can be taken, then no other transition is possible).

In the special case where $\delta(q,\gamma,\silent)$ is undefined for all $q\in Q$ and $\gamma\in\Gamma$ we refer to $\mathcal{A}$ as an $\silent$-free $n$-CPDA. In the special case where $\delta$ never performs a $\collapse$ (\ie links can safely be forgotten) we obtain the (weaker) model of \defin{higher-order pushdown automata}.

\defin{Configurations} of an $n$-CPDA are pairs of the form $(q, s)$
where $q \in Q$ and $s$ is an $n$-stack over $\Gamma$; the \defin{initial configuration} is $(q_0, \bot_n)$ and 
\defin{final configurations} are those whose control state belongs to $F$. Note that in some context (\eg when generating a tree or a game), final configurations are useless but we still assume a set $F$ for homogeneity in definitions.

An $n$-CPDA $\mathcal{A}=\anglebra{A, \Gamma, Q,\delta, q_0,F}$ naturally defines an $A$-labelled \emph{deterministic} (transition) graph $\transgraph{\mathcal{A}}=(V,E)$ whose
vertices $V$ are the configurations of $\mathcal{A}$ and whose edge relation $E\subseteq V\times A\times V$ is given by: 
$((q,s),a,(q',s'))\in E$ iff $\delta(q,\topn{1}(s),a)=(q',op)$ and $s'=op(s)$. Such a graph is called an \defin{$n$-CPDA graph}. 

\begin{example}\label{Example:CPDAAnBnCnd}
Consider the order-$2$ CPDA  $\mathcal{A}=\anglebra{\{1,2,\silent\}, \{\bot,\alpha\}, \{q_a,q_b,q_c,q_d,\widetilde{q}_a,\widetilde{q}_b,\widetilde{q}_c\},\delta, \widetilde{q}_a,q_d}$ with $\delta$ as follows:
\begin{itemize}
\item $\delta(\widetilde{q}_a,\bot,\silent)=\delta({q}_a,\alpha,1)=(q_a,\pushone{\alpha})$;
\item $\delta({q}_a,\alpha,2)=(\widetilde{q}_b,\pushn{2})$;
\item $\delta(\widetilde{q}_b,\alpha,\silent)=\delta({q}_b,\alpha,2)=(q_b,\popn{1})$;
\item $\delta({q}_b,\bot,2)=(\widetilde{q}_c,\popn{2})$;
\item $\delta(\widetilde{q}_c,\alpha,\silent)=\delta({q}_c,\alpha,2)=(q_c,\popn{1})$;
\item $\delta({q}_c,\bot,2)=({q}_d,id)$;
\end{itemize}

Its transition graph is depicted in Figure~\ref{Fig:Example:CPDAAnBnCnd}.
\end{example}

\begin{figure}
\begin{tikzpicture}[scale=1,transform shape]
\tikzstyle{every node}=[font=\normalsize]
\tikzset{>=stealth}
\tikzset{edge from parent/.style={draw,-}}
\node (n11) at (3,1.5) {$(\widetilde{q}_a,[[\bot]])$};
\node (n12) at (3,0) {$({q}_a,[[\bot\alpha]])$};
\node (n13) at (7,0) {$({q}_a,[[\bot\alpha\alpha]])$};
\node (n14) at (11,0) {$({q}_a,[[\bot\alpha\alpha\alpha]])$};
\node (n15) at (13,0){};
\draw[->] (n11) -- (n12) node[pos=0.5,left]{\small $\silent$};;
\draw[->] (n12) -- (n13) node[pos=0.5,above]{\small $1$};;
\draw[->] (n13) -- (n14) node[pos=0.5,above]{\small $1$};;
\draw[->,dotted] (n14) -- (n15) node[pos=0.5,above]{\small $1$};;

\node (n22) at (3,-1.5) {$(\widetilde{q}_b,[[\bot\alpha][\bot\alpha]])$};
\node (n23) at (7,-1.5) {$(\widetilde{q}_b,[[\bot\alpha\alpha][\bot\alpha\alpha]])$};
\node (n24) at (11,-1.5) {$(\widetilde{q}_b,[[\bot\alpha\alpha\alpha][\bot\alpha\alpha\alpha]])$};

\draw[->] (n12) -- (n22) node[pos=0.5,left]{\small $2$};;
\draw[->] (n13) -- (n23) node[pos=0.5,left]{\small $2$};;
\draw[->] (n14) -- (n24) node[pos=0.5,left]{\small $2$};;

\node (n32) at (3,-3) {$({q}_b,[[\bot\alpha][\bot]])$};
\node (n33) at (7,-3) {$({q}_b,[[\bot\alpha\alpha][\bot\alpha]])$};
\node (n34) at (11,-3) {$({q}_b,[[\bot\alpha\alpha\alpha][\bot\alpha\alpha]])$};

\draw[->] (n22) -- (n32) node[pos=0.5,left]{\small $\silent$};;
\draw[->] (n23) -- (n33) node[pos=0.5,left]{\small $\silent$};;
\draw[->] (n24) -- (n34) node[pos=0.5,left]{\small $\silent$};;

\node (n42) at (3,-4.5) {$(\widetilde{q}_c,[[\bot\alpha]])$};
\node (n43) at (7,-4.5) {$({q}_b,[[\bot\alpha\alpha][\bot]])$};
\node (n44) at (11,-4.5) {$({q}_b,[[\bot\alpha\alpha\alpha][\bot\alpha]])$};

\draw[->] (n32) -- (n42) node[pos=0.5,left]{\small $2$};;
\draw[->] (n33) -- (n43) node[pos=0.5,left]{\small $2$};;
\draw[->] (n34) -- (n44) node[pos=0.5,left]{\small $2$};;

\node (n52) at (3,-6) {$({q}_c,[[\bot]])$};
\node (n53) at (7,-6) {$(\widetilde{q}_c,[[\bot\alpha\alpha]])$};
\node (n54) at (11,-6) {$({q}_b,[[\bot\alpha\alpha\alpha][\bot]])$};

\draw[->] (n42) -- (n52) node[pos=0.5,left]{\small $\silent$};;
\draw[->] (n43) -- (n53) node[pos=0.5,left]{\small $2$};;
\draw[->] (n44) -- (n54) node[pos=0.5,left]{\small $2$};;

\node (n62) at (3,-7.5) {$({q}_d,[[\bot]])$};
\node (n63) at (7,-7.5) {$({q}_c,[[\bot\alpha]])$};
\node (n64) at (11,-7.5) {$(\widetilde{q}_c,[[\bot\alpha\alpha\alpha]])$};

\draw[->] (n52) -- (n62) node[pos=0.5,left]{\small $2$};;
\draw[->] (n53) -- (n63) node[pos=0.5,left]{\small $\silent$};;
\draw[->] (n54) -- (n64) node[pos=0.5,left]{\small $2$};;
\draw[->] (n63) -- (n52) node[pos=0.5,below left]{\small $2$};;

\node (n74) at (11,-9) {$({q}_c,[[\bot\alpha\alpha]])$};

\draw[->] (n64) -- (n74) node[pos=0.5, left]{\small $\silent$};
\draw[->] (n74) -- (n63) node[pos=0.5,below left]{\small $2$};;

\end{tikzpicture}
\caption{The transition graph of the order-$2$ CPDA from Example~\ref{Example:CPDAAnBnCnd}}\label{Fig:Example:CPDAAnBnCnd}
\end{figure}
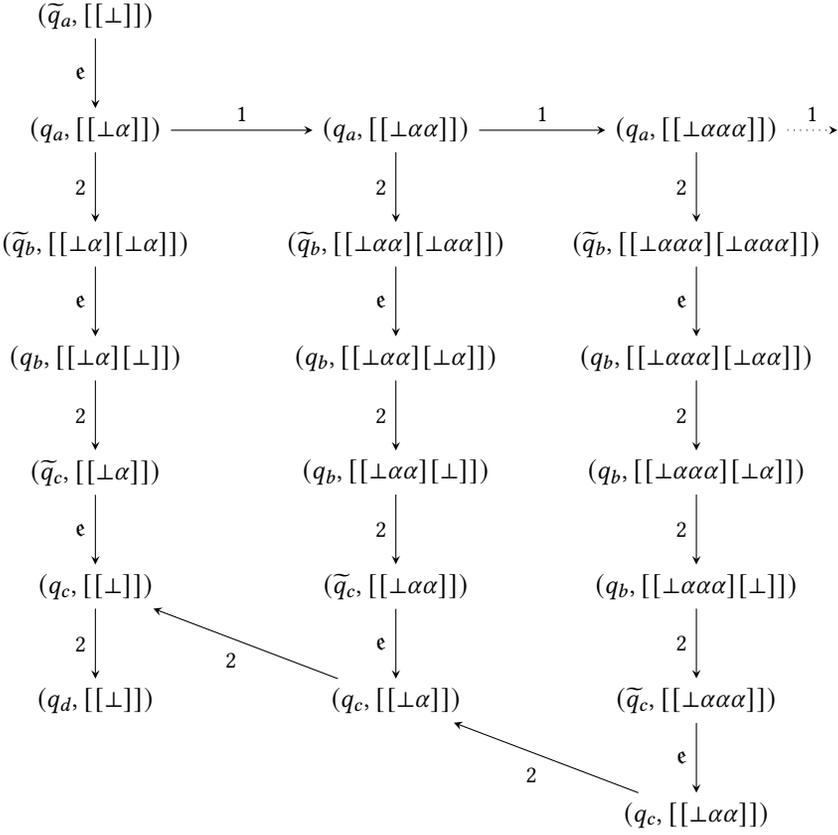

In this paper we will use $n$-CPDA for three different purposes: as words acceptors, as generators for infinite trees and as generators of parity game.

\subsection{Using an $n$-CPDA as a Words Acceptor}

A order-$n$ CPDA $\mathcal{A}=\anglebra{A, \Gamma, Q,\delta, q_0,F}$ accepts the set of words $w  \in (A\setminus\{\silent\})^*$ labelling a run from the initial configuration to a final configuration (interpreting $\silent$ as a silent move). We write
$\langweps{\mathcal{A}}$ {for} the accepted language.

Following the notations from Section~\ref{subsection:graphs}, and letting, for a word $w =a_1 \cdots a_k\in A_{\silent}^*$, $L_w = \silent^* a_1 \silent^* \cdots \silent^* a_k\silent^*$ one has $$\langweps{\mathcal{A}}=\{w\in A_{\silent}^* \mid (q_0,\bot_{n})\era{L_w} (q_f,\stack) \text{ with } q_f\in F\} $$

\begin{example}
	Consider again the order-$2$ CPDA from Example~\ref{Example:CPDAAnBnCnd}. Then its accepted language is $\{1^k 2^{2k+3}\mid k\geq 0\}$.
\end{example}

\subsection{Using an $n$-CPDA as an Infinite Tree Generator}

We now explain how to generate a $\Sigma$-labelled ranked and ordered trees using an $n$-CPDA. For this, let $\Sigma$ be a ranked alphabet, let $m=\max\{\arity{a}\mid a\in \Sigma\}$ and consider an $n$-CPDA $\mathcal{A}=\anglebra{A, \Gamma, Q,\delta, q_0,F}$ where $A=\{1,\dots,m\}\cup\{\silent\}$ together with a function $\rho:Q\rightarrow \Sigma$ that we extend as a function from configurations of $\mathcal{A}$ by letting $\rho((q,\stack))=\rho(q)$. Moreover, assume that, for all $q\in Q$ and $\gamma\in \Gamma$, $\{a\mid a \neq \silent \text{ and } (q,\gamma,a)\in\dom{\delta}\}=\{1,\dots,\arity{\rho(q)}\}$.

Then, let $G=\transgraph{\mathcal{A}}$ be the $A$-labelled deterministic transition graph associated with $\mathcal{A}$ and following Section~\ref{subsection:graphs}, consider the tree $t_\mathcal{A}: \unfold{G}\rightarrow\Sigma$ obtained by taking as domain the tree obtained by unfolding $G$ from the initial configuration $(q_0,\bot_n)$, contracting all $\silent$-transitions and labelling its nodes thanks to function $\rho$. 

\begin{example}\label{Example:CPDAAnBnCndTree}
	Consider again the order-$2$ CPDA from Example~\ref{Example:CPDAAnBnCnd} and the ranked alphabet $\Sigma=\{a,b,c,d\}$ where $\arity{a}=2$, $\arity{b}=\arity{c}=1$, and $\arity{d}=0$. Define $\rho:Q\rightarrow \Sigma$ by letting $\rho(q_a)=\rho(\widetilde{q}_a)=a$, $\rho(q_b)=\rho(\widetilde{q}_b)=b$, $\rho(q_c)=\rho(\widetilde{q}_c)=c$ and $\rho(q_d)=d$. Then, the tree $t_{\mathcal{A}}$ generated by $\mathcal{A}$ with the labelling function $\rho$ is the one depicted in Figure~\ref{fig:ex-rewrittingAnBnCnDTree}.
\end{example}

\subsection{Using an $n$-CPDA to Define a Parity Game}

Let $\mathcal{A}=\anglebra{A, \Gamma, Q,\delta, q_0,F}$ be an order-$n$ CPDA and let
$(V,E)$ be the graph obtained from $\transgraph{\mathcal{A}}$ by removing edge-labels. To stick to the definition in Section~\ref{def:paritygames} we assume that $\transgraph{\mathcal{A}}$ has no dead-end.
Let $Q_\Ei\uplus Q_\Ai$ be a partition of $Q$ and let $\col:Q \rightarrow \colors\subset\mathbb{N}$ be a colouring function (over states). Altogether they
define a partition $V_\Ei\uplus V_\Ai$ of $V$ whereby a vertex belongs
to $V_\Ei$ iff its control state belongs to $Q_\Ei$, and a colouring
function $\col:V \rightarrow \colors$ where a vertex is assigned the colour
of its control state. The structure
$\mathcal{G}=(\transgraph{\mathcal{A}},V_\Ei,V_\Ai)$ defines an arena and the pair $\mathbb{G}=(\mathcal{G},\col)$ defines a parity game (that we call an 
\defin{$n$-CPDA parity game}). 

In this context, one can also use a CPDA to define a strategy. Indeed, fix an order-$n$ CPDA $\mathcal{A}=\anglebra{A,\Gamma, Q,\delta, q_0,F}$ defining a $n$-CPDA parity game $\pgame$. 

{Consider a partial play $v_0v_1\cdots v_\ell$ in $\pgame$ where $v_0 = (q_0,\bot_n)$, together with the sequence of labels $\play\in A^*$ of the corresponding path. As $\mathcal{A}$ is deterministic, one can represent a strategy as a (partial) function $\strat:A^*\rightarrow A$.}

Now let $\mathcal{A}'=\anglebra{A,\Gamma',Q',\delta',q_0',F'}$ be an order-$n$ CPDA together with a function $\tau:Q'\times\Gamma'\rightarrow A$. Then $\mathcal{A'}$ defines a strategy $\strat_{\mathcal{A}'}$ for \Eloise in $\pgame$ by letting $\strat_{\mathcal{A}'}(\play)=\tau((q',\topn{1}(\stack')))$ where $(q',\stack')$ is the (unique) configuration reached by $\mathcal{A}'$ from its initial configuration by reading $\lambda$ (seen as an element in $A^*$ as explained above).

We say that $\mathcal{A}$ and $\mathcal{A'}$ are \concept{synchronised} if, for all $u\in A^*$ and $a\in A$, if $(q,\stack)$ and $(q',\stack')$ denote the respective configurations reached by $\mathcal{A}$ and $\mathcal{A}'$ when reading $u$ from their initial configuration, then $\delta(q,\topn{1}(\stack),a)$ and $\delta'(q',\topn{1}(\stack'),a)$ {yields the same stack operation up to a renaming of $\gamma$ in $\pushlk{e}{\gamma}$ and $\toprew{\gamma}$}. In particular it implies that the stacks $\stack$ and $\stack'$ have the same shape if one defines the \emph{shape} of a stack as the stack obtained by replacing all symbols appearing in $\stack$ by a fresh symbol $\sharp$ (but keeping the links).

\section{Known results}\label{section:KnownResult}

We give now a few known results that we will build on in the following. Complete proofs of these results can be found in the companion papers~\cite{HMOS17} and \cite{BCHMOS20}.

\subsection{The Equi-Expressivity Theorem}

In \cite{HMOS08,HMOS17} the following equi-expressivity result was proved (see also \cite{CS12} for an alternative proof).

\begin{theorem}[Equi-expressivity]\label{theo:equi} Order-$n$ recursion schemes and order-$n$ collapsible pushdown automata are equi-expressive for generating trees. That is, we have the following.
\begin{enumerate}
\item Let $\rscheme$ be an order-$n$ recursion scheme over $\Sigma$ and let $t$ be its value tree. There is an order-$n$ CPDA $\mathcal{A}=\anglebra{A, \Gamma, Q,\delta, q_0,F}$, and a function $\rho: Q \rightarrow \Sigma$ such that $t$ is the tree generated by $\mathcal{A}$ and $\rho$.
\item Let $\mathcal{A}=\anglebra{A, \Gamma, Q,\delta, q_0,F}$ be an order-$n$ CPDA, and let $t$ be the $\Sigma$-labelled tree generated by $\mathcal A$ and a function $\rho: Q \rightarrow \Sigma$. There is an order-$n$ recursion scheme over $\Sigma$ whose value tree is $t$.
\end{enumerate}
Moreover the inter-translations between schemes and CPDA are polytime computable.
\end{theorem}

\subsection{Collapsible Pushdown Parity Games And $\mu$-Calculus Definable Sets}

We refer the reader to~\cite{BCHMOS20} for a unified and self content presentation of the following results.

Collapsible pushdown parity games were first studied in~\cite{HMOS08} where it was established that one could decide the winner from a given initial vertex. 

\begin{theorem}\label{theorem:winning}
Let $\mathcal{A}=\anglebra{A,\Gamma, Q,\delta, q_0,F}$ be the $n$-CPDA and let $\pgame$ be an $n$-CPDA parity game defined from $\mathcal{A}$. Then deciding whether $(q_0,\bot_n)$ is winning for \Eloise is an $n$-\exptime~complete problem.
\end{theorem}

This was further extended in~\cite{BCOS10} where the computation of a (finite presentation) of the winning region was considered. In particular, from a game $\pgame$ one can build a new game that behaves the same but where the winning region is explicitly marked.

\begin{theorem}\label{theorem:marking}
Let $\mathcal{A}=\anglebra{A,\Gamma, Q,\delta, q_0,F}$ be an $n$-CPDA and let $\pgame$ be the $n$-CPDA parity game defined from $\mathcal{A}$. Then, one can build an order-$n$ CPDA $\mathcal{A}'=\anglebra{A,\Gamma', Q',\delta', q_0',F'}$ such that the following holds.
\begin{enumerate}
\item Restricted to the reachable configurations from their respective initial configuration, the transition graph of $\mathcal{A}$ and $\mathcal{A'}$ are isomorphic.
\item For every configuration $(q,\stack)$ of $\mathcal{A}$ that is reachable from the initial configuration, the corresponding configuration $(q',\stack')$ of $\mathcal{A'}$ is such that $(q,\stack)$ is winning for \Eloise in $\pgame$ if and only if $q'\in F$.
\end{enumerate}
\end{theorem}

Regarding $\mu$-calculus global model-checking against graphs generated by CPDA the following logical counterpart of Theorem~\ref{theorem:marking} was first proved in~\cite{BCOS10}.

\begin{theorem}\label{theorem:marking-mu-calculus}
Let $\mathcal{A}=\anglebra{\Gamma, Q,\delta, q_0,F}$ be an $n$-CPDA and let $\phi$ be a $\mu$-calculus formula defining a subset of vertices in the configuration graph of $\mathcal{A}$. Then, one can build an order-$n$ CPDA $\mathcal{A}'=\anglebra{A,\Gamma', Q',\delta', q_0',F'}$ and a mapping $\chi:Q'\rightarrow Q$ such that the following holds.
\begin{enumerate}
\item Restricted to the reachable configurations from their respective initial configuration, the transition graph of $\mathcal{A}$ and $\mathcal{A'}$ are isomorphic.
\item For every configuration $(q,\stack)$ of $\mathcal{A}$ that is reachable from the initial configuration, the corresponding configuration $(q',\stack')$ of $\mathcal{A'}$ is such that $q=\chi(q')$, and $\phi$ holds in $(q,\stack)$ if and only if $q'\in F$.
\end{enumerate}
\end{theorem}

Finally, in~\cite{CS12}, the computation of a finite description of a winning strategy was studied. 

\begin{theorem}\label{theorem:strategies}
Let $\mathcal{A}$ be an $n$-CPDA defining an $n$-CPDA parity game $\pgame$. If the initial configuration is winning for \Eloise then one can effectively construct an $n$-CPDA $\mathcal{A}'$ that is  synchronised with $\mathcal{A}$ and realises a well-defined winning strategy for \Eloise in $\pgame$ from the initial configuration.
\end{theorem}

\section{Local Model-Checking}\label{section:modelChecking}

Recall that, for a formula  $\phi$ (without free variables if one considers MSO logic; on the pointed tree if one considers $\mu$-calculus) and a tree $t$ the \emph{(local) model-checking} problem asks to decide whether $\phi$ holds in $t$. 

In the case of trees, MSO formulas without free variable and $\mu$-calculus formulas have the same expressive power (see \emph{e.g.}~\cite{JW96}). 
It is also a very standard result (obtained by combining the previous equivalence with Theorem~\ref{th:Rabin69} and Theorem~\ref{th:GurevichH82}; see also \cite{AN01,Wilke2001} for a direct construction) that $\mu$-calculus model-checking and deciding the winner in a parity game played on an arena obtained by considering a synchronised product of the tree to model-check together with a finite graph (describing the formula and its dynamics) are two equivalent problems. Hence, the equi-expressivity Theorem together with the fact that CPDA parity games are decidable (Theorem~\ref{theorem:winning}), directly imply the decidability of the MSO/$\mu$-calculus model-checking problem against trees generated by recursion schemes/CPDA. Note that historically this result was first established for trees generated by recursion schemes by Ong in~\cite{Ong06a} using tools from innocent game semantics (in the sense of Hyland and Ong \cite{HO00})

\begin{theorem}
	MSO (equivalently $\mu$-calculus) local model-checking problem is decidable for any tree generated by a recursion scheme (equivalently by a collapsible pushdown automaton). 
\end{theorem}

\section{Global Model-Checking}\label{section:globalMC}

Recall that, for a formula $\phi$ and a tree $t$ the \emph{global model-checking} asks (if there is one) for a description of the set $\truth{t}{\phi}$ of nodes of $t$ where $\phi$ holds. 

\subsection{Exogeneous \& Endogeneous Approaches to Global Model-Checking}

We present here two approaches to the global model-checking problem:  the exogenous one where the set is described by an external device (here a CPDA), and the endogenous one which is new and where the description is internalised by a recursion scheme with “polarized” labels.

\begin{itemize}
\item \emph{Exogeneous} approach: Given a $\Sigma$-labelled tree $t:\dom{t}\rightarrow \Sigma$
  and a formula $\phi$, output a description by means of a word
  acceptor device recognising $\truth{t}{\phi}\subseteq
  \dom{t}$.
\item \emph{Endogeneous} approach: Given a $\Sigma$-labelled tree
  $t:\dom{t}\rightarrow \Sigma$ and a formula $\phi$, output a finite description
  of the $(\Sigma\cup\marked{\Sigma})$-labelled tree $t_\phi:\dom{t}\rightarrow \Sigma\cup\marked{\Sigma}$, where 
  $\marked{\Sigma}=\{\marked{\sigma}\mid \sigma\in\Sigma\}$ is a
  marked copy of $\Sigma$, such that $t_\phi$ and $t$ have the same domain, and
  $t_\phi(u)=\marked{t(u)}$ if $u\in\truth{t}{\phi}$ and
  $t_\phi(u)=t(u)$ otherwise. 
\end{itemize}

\begin{example}\label{ex:GMC}
  Let $\mathcal{S}$ be the order-2 recursion scheme 
  with non-terminals {$I :o,\, F : ((o,o),o,o)$}
  variables $x : o,\, \zeta : (o,o)$, 
  terminals $f, g, a$ of arity $2, 1, 0$ respectively,
  and the following production rules:
  
 $$\left\{\begin{array}{rll}
  I & \rightarrow & F \, g\, a\\
  F\, \zeta \, x & \rightarrow & f \, (F\, \zeta \, (\zeta \, x))\, (\zeta \, x)\\
  \end{array}\right. 
  $$
 The value tree $t=\mng{\mathcal{S}}$
 is the following $\Sigma$-labelled tree:

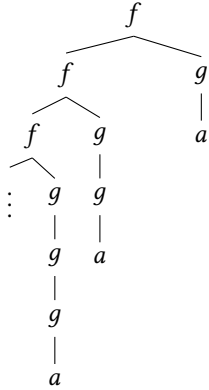
\begin{center}
\begin{tikzpicture}[level distance=8mm,level/.style={sibling distance=18mm/#1}]
\node {$f$}
   child {node {$f$}
         child {node {$f$} 
             child {node {$\vdots$}}
             child {node {$g$} child { node {$g$} child {node {$g$} child {node {$a$}}}}}}
         child {node {$g$} child { node {$g$} child {node {$a$}}}}}
   child {node {$g$} child {node {$a$}}};
\end{tikzpicture} 
\end{center}

{Let $\phi$ be the
 $\mu$-calculus formula defining the nodes
which are labelled by $g$ such that the length of the (unique) path to an
 $a$-labelled node is odd. Formally (we refer to \cite{AN01,Wilke2001} for syntax and semantics of $\mu$-calculus), $\phi = p_g \wedge \mu X.(\diamond_1 p_a\vee \diamond_1\diamond_1
 X)$, where $p_g$ (resp. $p_a$) is a propositional variable {asserting} that the current node is labelled by $g$ (resp. $a$).}

An \emph{exogenous} approach to the global model-checking problem is to
output a description (\eg by means of a finite-state automaton) of the set
$\truth{t}{\phi}=\{1^n21^k\mid n+k \text{ is odd}\}$, which in this special case is a regular language.

An \emph{endogenous} approach to this problem is to output the following recursion scheme:
 $$\left\{\begin{array}{rll}
  I & \rightarrow & H \, \marked{g}\, a\\
  H\, z & \rightarrow & f \, (\marked{H}\, g \, z)\, z\\
  \marked{H}\, z & \rightarrow & f \, (H\, \marked{g} \, z)\, z\\
 \end{array}\right. 
  $$
  with non-terminals $I :o,\, H : (o,o)$ and a variable $z : o$. The value tree of this new scheme is as desired:

\begin{center}
\begin{tikzpicture}[level distance=8mm,level/.style={sibling distance=18mm/#1}]
\footnotesize
\node {$f$}
   child {node {$f$}
         child {node {$f$}  
             child {node {$\vdots$}}
             child {node {$\marked{g}$} child { node {$g$} child {node {$\marked{g}$} child {node {$a$}}}}}}
         child {node {$g$} child { node {$\marked{g}$} child {node {$a$}}}}}
   child {node {$\marked{g}$} child {node {$a$}}};
\end{tikzpicture} 
\end{center}
\end{example}

We now define a general concept, called \emph{reflection}, and which expresses the ability to perform the endogenous approach within a class of trees for a given logic.

\begin{definition}{(Reflection)}
Let $\mathcal{C}$ be a class of trees, and let $\mathcal{L}$ be some logical formalism. 
Let $t$ be a tree in $\mathcal{C}$ and let $\phi$ be an $\mathcal{L}$-formula. We say that a tree $t'\in \mathcal{C}$ is a \concept{$\phi$-reflection} of $t$ just if $t'=t_\phi$. We say that the class $\mathcal{C}$ is \concept{$\mathcal{L}$-reflective} in case for all $t\in\mathcal{C}$ and all $\phi\in\mathcal{L}$ one has $t_\phi\in\mathcal{C}$.

In case the class $\mathcal{C}$ is finitely presented (\ie each element of $\mathcal{C}$ comes with a finite presentation, \eg given by a recursion scheme or by a collapsible pushdown automaton), we say that $\mathcal{C}$ is \concept{$\mathcal{L}$-effectively-reflective} if $\mathcal{C}$ is $\mathcal{L}$-reflective and moreover one can effectively construct, for any (presentation of) $t\in\mathcal{C}$ and any $\phi\in\mathcal{L}$, (a presentation of)  $t_\phi$: \ie there is an algorithm that, given a formula $\phi\in\mathcal{L}$, transforms a presentation of an element in $\mathcal{C}$ into a presentation of its $\phi$-reflection.
\end{definition}

In the sequel, we prove that the class of trees generated by recursion schemes/CPDAs is $\mu$-calculus-effectively-reflective as well as MSO-effectively-reflective. 


\subsection{$\mu$-Calculus Reflection}

Regarding $\mu$-calculus, the following result providing an exogeneous and an endogeneous description of $\mu$-calculus definable sets is a simple consequence of the equi-expressivity theorem together with Theorem~\ref{theorem:marking-mu-calculus}.

\begin{theorem}\label{theo:reflection-mu-calculus}
Let $t$ be a $\Sigma$-labelled tree generated by an order-$n$ recursion scheme $\rscheme$ and $\phi$ be a $\mu$-calculus formula. 
\begin{enumerate}
\item\label{item:reflection-mu-calculus-1} There is an algorithm that takes $(\rscheme, \phi)$ as its input and outputs an order-$n$ CPDA $\mathcal{A}$ such that $L(\mathcal{A})=\truth{t}{\phi}$.
\item\label{item:reflection-mu-calculus-2} There is an algorithm that takes $(\rscheme, \phi)$ as its input and outputs an order-$n$ recursion scheme that generates $t_\phi$.
\end{enumerate}
\end{theorem}

\begin{proof}
First remark that (\ref{item:reflection-mu-calculus-2}) implies (\ref{item:reflection-mu-calculus-1}).  To see why this is so, assume that we can construct an order-$n$ recursion scheme generating $t_\phi$.
Thanks to Theorem~\ref{theo:equi}, we can construct 
an order-$n$ CPDA $\mathcal{A}$ which, together with a mapping $\rho: Q \mapsto \Sigma \cup \marked{\Sigma}$, generates $t_\phi$. Taking $\{ q \in Q \mid \rho(q) \in \marked{\Sigma} \}$ as
a set of final states, and using $\mathcal{A}$ as a finite words acceptor it immediately follows that $L(\mathcal{A})=\truth{t}{\phi}$.

We now concentrate on (\ref{item:reflection-mu-calculus-2}).  Fix an order-$n$ recursion scheme
  $\rscheme = \anglebra{\Sigma, \mathcal{N}, \mathcal{R}, I}$ and let $t$ be
  its value tree; let $m=\max\{\arity{a} \mid a\in\Sigma\}$. Let $\phi$ be a $\mu$-calculus formula. Using
  Theorem~\ref{theo:equi}, 
  we can construct an $n$-CPDA
  $\mathcal{A}=\anglebra{A, \Gamma, Q,\delta, q_0,F}$ with $A_\silent=\{1,\dots,m\}$
  and a mapping $\rho: Q \rightarrow \Sigma$ such that $t$
  is the tree generated by $\mathcal{A}$ and $\rho$. 
 
Let $G=(V,E)$ be the transition graph of $\mathcal{A}$. From $\transgraph{\mathcal{A}}$ we define a new (deterministic) graph $G'=(V,E')$ obtained by contracting the $\silent$-labelled edges in $G$, \ie $E'=\{(v_1,a,v_2)\mid a\in A_\silent \text{ and } v_1 \era{\silent^*a} v_2\}$. From the definitions it directly follows that $G'$ equipped with the labelling function induced by $\rho$ defines the same tree as $G$, namely $t$.

  Assume that we have, for every state $q$ of $\mathcal{A}$, a predicate
  $p_q$ that holds at a node $(q',\stack)\in V$ iff $q=q'$. Then the
  formula $\phi$ can be translated to a formula $\phi'$ on $G'$ as follows: for each $a \in \Sigma$, replace every
  occurrence of the predicate $p_a$ in $\phi$ by the
  disjunction $\bigvee_{q \in Q, \rho(q)=a} p_q$.
  Then, by definition of $\phi'$ and because $G'$ generates $t$, one has $$\truth{t}{\phi} = \{u\in A_\silent^*\mid (G',v'_u)\models \phi'\}$$ where we denote by $v'_u$ the (unique, if exists) vertex in $G'$ that is reachable from the initial configuration by a path labeled by $u$.

  In turn $\phi'$ can be translated to a formula $\phi_\silent$ on $G$ by  replacing in $\phi'$ every sub-formula of the form $\diamond_a \psi$ by $\mu X.(\diamond_a\psi\vee\diamond_\silent X)$,
  \ie we replace the assertion “take an $a$-edge to a vertex where $\psi$ holds” by the assertion “take a (possibly empty) finite sequence of $\silent$-edges to a vertex from which there is an $a$-edge to a vertex where $\psi$ holds”.
  Then, by definition of $\phi_\silent$ and from how $G'$ was defined from $G$, one has, for every $u\in A$, that $$(G,v_u)\models \phi_\silent \quad \text{iff} \quad (G',v'_{u_\silent})\models \phi'$$ where we denote by $v_u$ the (unique, if exists) vertex in $G$ that is reachable from the initial configuration by a path labeled by $u$ and where we denote by $v'_{u_\silent}$ the (unique, if exists) vertex in $G'$ that is reachable from the initial configuration by a path labeled by the word $u_\silent\in A_\silent^*$ obtained from $u$ by removing all occurrences of $\silent$.

Now use Theorem~\ref{theorem:marking-mu-calculus} for $\mathcal{A}$ and $\phi_\silent$, leading to a new order-$n$ CPDA $\mathcal{A'}=\anglebra{A,\Gamma', Q',\delta', q_0',F'}$ and a mapping $\chi:Q'\rightarrow Q$. As the transitions graphs (when restricted to reachable configurations from the initial configuration) of $\mathcal{A}$ and $\mathcal{A'}$ are isomorphic, it follows that $\mathcal{A'}$ and $\rho\circ \chi:Q'\rightarrow \Sigma$ generates $t$.

It follows at once that the tree $t_\phi$ is generated by $\mathcal{A}'$ and the mapping $\rho':Q'\rightarrow \Sigma\cup \marked{\Sigma}$ defined by 
$$
\rho'(q') = 
\begin{cases}
\rho(\chi(q')) & \text{if }	q' \not\in F\\
\marked{\rho(\chi(q'))} & \text{otherwise.}
\end{cases}
$$

{According to Theorem~\ref{theo:equi}, one can construct from $\mathcal{A}'$ an order-$n$ recursion scheme generating $t_\phi$.}
\end{proof}

\begin{remark}
There are natural questions concerning complexity. The first one concerns the algorithm in Theorem \ref{theo:reflection-mu-calculus}: it is $n$-time exponential in both the size of the scheme and the size of the formula. This is because we need to solve an order-$n$ CPDA parity game when invoking Theorem~\ref{theorem:marking-mu-calculus} (see \cite{BCHMOS20}) built by taking a product of an order-$n$ CPDA equi-expressive with $\rscheme$ (thanks to Theorem \ref{theo:equi} its size is polynomial in the one of $\rscheme$) with a finite transition system of polynomial size in that of $\phi$. 
The second issue concerning complexity is how the size of the new scheme (obtained in the second point of Theorem \ref{theo:reflection-mu-calculus}) relates to that of $\rscheme$ and $\phi$.
For similar reasons, it is $n$-time exponential in the size of $\rscheme$ and $\phi$.
The last one concerns the size of the order-$n$ CPDA in the first point of Theorem \ref{theo:reflection-mu-calculus}: because it is constructed from the new scheme  given by the second point, it is also $n$-time exponential in the size of $\rscheme$ and $\phi$.
\end{remark}

\subsection{MSO Reflection}

It is natural to ask if trees generated by recursion schemes are reflective with respect to MSO. Indeed, $\mu$-calculus and MSO are equivalent for expressing properties of a {deterministic} tree at the \emph{root}, but not other nodes; see \emph{e.g.}~\cite{JW96}. In fact one would need backwards modalities to express all of MSO in $\mu$-calculus. 

\begin{example}\label{example:reflectionMSO}
Consider the following property $\phi(x)$ (definable in MSO but not in $\mu$-calculus) on {nodes $x$ of a tree}: “$x$ is the right son of an $f$-labelled node, and there is a path from $x$ to an $a$-labelled node which contains an odd number of occurrences of $g$-labelled nodes”. 

Returning to the scheme of Example \ref{ex:GMC}, an \emph{exogenous} approach to the global model-checking problem is to output a description (\eg by means of a finite-state automaton) of the language $\truth{t}{\phi}=\{1^n2\mid n \text{ is even}\}$, which in this special case is a regular language.

An \emph{endogenous} approach to this problem is to output the following recursion scheme:
$$\left\{\begin{array}{rll}
 I & \rightarrow & \marked{F} \, g\, a\\
 \marked{F}\, \phi \, x & \rightarrow & f \, (F\, g \, (\phi \, x))\, (\marked{g} \, x)\\
 F\, \phi \, x & \rightarrow & f \, (\marked{F}\, g \, (\phi \, x))\, ({g} \, x)\\
\end{array}\right. 
 $$ 
  with non-terminals $I :o,\, F : ((o,o),o,o),\, \marked{F} : ((o,o),o,o)$ and variables $x : o,\, \phi : (o,o)$. The value tree of this new scheme is as desired:

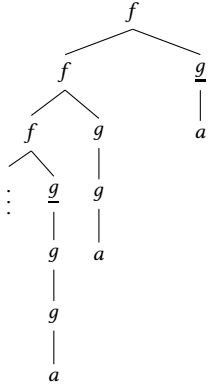
\begin{center}
\begin{tikzpicture}[level distance=8mm,level/.style={sibling distance=18mm/#1}]
\footnotesize
\node {$f$}
   child {node {$f$}
         child {node {$f$}  
             child {node {$\vdots$}}
             child {node {$\marked{g}$} child { node {$g$} child {node {${g}$} child {node {$a$}}}}}}
         child {node {$g$} child { node {$g$} child {node {$a$}}}}}
   child {node {$\marked{g}$} child {node {$a$}}};
\end{tikzpicture} 
\end{center}
\end{example}

Relying on the $\mu$-calculus reflection, one can prove that the class of trees generated by recursion schemes is MSO-reflective.

\begin{theorem}\label{theo:reflection-MSO}
Let $t$ be a $\Sigma$-labelled tree generated by an order-$n$ recursion scheme $\rscheme$ and $\phi(x)$ be an MSO formula. 
\begin{enumerate}
\item There is an algorithm that takes $(\rscheme, \phi)$ as its input and outputs an order-$n$ CPDA $\mathcal{A}$ such that $L(\mathcal{A})=\truth{t}{\phi}$.
\item There is an algorithm that takes $(\rscheme, \phi)$ as its input and outputs an order-$n$ recursion scheme that generates $t_\phi$.
\end{enumerate}
\end{theorem}

\begin{proof}
We only concentrate on $(2)$ as it implies $(1)$ using the same argument as in the proof of Theorem~\ref{theo:reflection-mu-calculus}.

For any node $u$, we let $\treemarkednode{t}{u}$ denote the tree obtained from $t$ by marking the node $u$ (and no other node). {Recall (see definition on  page~\pageref{th:Rabin69}) that formally  $\treemarkednode{t}{u}$ is a $\Sigma \times \{0,1\}$-labelled tree. The second component of the alphabet is used to indicate the position of marked node. In addition, for a $\Sigma$-labelled tree $t$, we let $\tilde{t}$ denote the $\Sigma \times \{0,1\}$-labelled tree such that $\dom{t}=\dom{\tilde{t}}$ and for all $u \in \dom{\tilde{t}}$, $\tilde{t}(u)=(t(u),0)$ which corresponds to the tree $t$ in which no node is marked. }

Using Theorem~\ref{th:Rabin69}, one can construct a tree automaton $\mathcal{B}_{\phi(x)}$ that accepts $\treemarkednode{t}{u}$ iff $(t,u)\models \phi(x)$. We let $S$ denote the set of control states of $\mathcal{B}_{\phi(x)}$.

In order to construct $t_\phi$, we first annotate $t$ with
informations on the behaviour of $\mathcal{B}_{\phi(x)}$ on the subtrees of
{$\tilde{t}$}. 

We mark $t$ by $\mu$-calculus definable sets to obtain an
enriched tree denoted $\bar{t}$.  With each pair $(q,d)\in S\times \{1,\dots,m\}$, where $m=\max\{\arity{a}\mid a\in \Sigma\}$, we associate {a $\mu$-calculus formula} $\psi_{q,d}$ such that, for every node $u$, {$u \in \truth{t}{\psi_{q,d}}$} iff the $d$-child of $u$ exists and $\mathcal{B}_{\phi(x)}$ has an accepting run on $\tilde{t}[ud]$ starting from $q$, where {$\tilde{t}[ud]$} denotes the subtree of $\tilde{t}$ rooted at {$ud$}. {The existence of $\psi_{q,d}$ is due to fact that acceptance by parity tree automata is expressible in $\mu$-calculus \cite{StreettE89}.}
The tree $\bar{t}$ is the ${\Sigma'=\Sigma\times 2^{S \times\{1,\dots,m\}}}$-labelled ranked tree with domain $\dom{t}$ where for every $u\in\dom{t}$,
$$\bar{t}(u)=(t(u),\{(q,d)\mid\mathcal{B}_{\phi(x)}[q] \text{ accepts } {\tilde{t}[ud]\}})$$ where $\mathcal{B}_{\phi(x)}[q]$ denotes the automaton obtained from $\mathcal{B}_{\phi(x)}$ by changing the initial state to $q$.
Then, by (successive applications of) Theorem~\ref{theo:reflection-mu-calculus}, $\bar{t}$ can be generated by an order-$n$ collapsible automaton.

Using the annotations on $\bar{t}$, for any node $u$, one can decide, only considering the path from the root to $u$, whether $\mathcal{B}_{\phi(x)}$ accepts $\treemarkednode{t}{u}$. More precisely, there exists a regular words language $L$ over $\Sigma' \cup \{1,\dots,m\}$ such that a node $u$ of $t$ satisfies $\phi$ if and only if the word obtained by reading in $\bar{t}$ the labels and directions from the root to the node $u$ belongs to $L$. 

{To prove the existence of the regular language $L$, we introduce the notion of partial accepting run of $\mathcal{B}_{\phi(x)}$ on $\tilde{t}$ stopping at $u$ in state $p$ which is defined similarly to an accepting run except that it is required to assign the state $p$ to the node $u$ and that it is undefined for all nodes below $u$. Note that all infinite branches for which it is defined are required to satisfy the acceptance condition.

For a node $u$, we let $P(u)$ denote the set of states $p \in S$ for which $\mathcal{B}_{\phi(x)}$ has a partial accepting run on $\tilde{t}$ stopping at $u$ in state $p$. Remark that, by definition, $P(\varepsilon)$ is the set of initial states of $\mathcal{B}_{\phi(x)}$.

For two nodes $u$ and $ud$ in $\bar{t}$, the set $P(ud)$ can be computed from $P(u)$, the label of $u$ in $\tilde{t}$ and the direction $d$. Indeed, for a node $u$ labelled by $(\sigma,Q) \in \Sigma \times 2^{S \times\{1,\dots,m\}}$ in $\bar{t}$ and for all states $q \in S$, $q$ belongs to $P(ud)$ if and only if there exists a state $p \in P(u)$ and a transition $(p,(\sigma,0),p_1,\ldots,p_k)$ of $\mathcal{B}_{\phi(x)}$ such that for all $i \in [1,k]\setminus \{d\}$, $(p_i,i)$ belongs to $Q$ and $p_d$ is equal to $q$.

Using this property, one easily constructs a finite automaton over $\Sigma' \cup \{1,\ldots,m\}$ which after reading the path from the root to a node $u$ computes the set $P(u)$.

The final remark is that for a node $u$, we can decide whether $\mathcal{B}_{\phi(x)}$ accepts $\treemarkednode{t}{u}$ by only considering the set $P(u)$ and the label of $u$ in $\bar{t}$. Indeed, for a node $u$ labeled by $(\sigma,Q) \in \Sigma \times 2^{S \times\{1,\dots,m\}}$ in $\bar{t}$, $\treemarkednode{t}{u}$ is accepted by $\mathcal{B}_{\phi(x)}$ if and only if there exists a node $p \in P(u)$ and a transition  $(p,(\sigma,1),p_1,\ldots,p_k)$ such that for all $i \in [1,k]$, $(p_i,i)$ belongs to $Q$.

Hence, one easily constructs a finite automaton accepting the language $L$.}

Finally, an order-$n$ CPDA generating $t_\phi$ is obtained by taking a synchronised product between an order-$n$ CPDA generating $\bar{t}$ and a finite \emph{deterministic} automaton recognising $L$: a node in the generated tree is marked iff the associated control state in the automaton recognising $L$ is final. 
\end{proof}
\section{Some Consequences of MSO Reflection}\label{section:globalMCConsequences}

We derive two consequences of the MSO reflection. The first one (Corollary~\ref{cor:divergent}) is to show how MSO reflection can be used to construct, starting from a scheme that may have non-productive rules, an equivalent one that does not have such divergent computations. The second application consists in proving (Theorem~\ref{theorem:ALaCaucal}) that the class of trees generated by recursion schemes is closed under the operation of MSO interpretation followed by tree unfolding, hence providing a result in the same flavour as the one obtained by Caucal for safe schemes in~\cite{Caucal02}.

\subsection{Avoiding Divergent Computations}\label{sec:divergentComputations}

Let $\mathcal{S}$ be the order-2 recursion scheme with non-terminals $I :o,\, H:(o,o), \, F : ((o,o,o),o)$, variables $z : o,\, \phi : (o,o,o)$, terminals $\Sigma=\{f, a\}$ of arity $2$ and $0$ respectively, and the following production rules:
\[\begin{array}{rll}
I & \rightarrow & f (H\,a) (F\, f) \\
H\, z &\rightarrow & H\, (H\,z) \\
F\,\varphi & \rightarrow &\varphi\, a\,(F\,\varphi) \\
\end{array}\]

The second rule $H\, z \,\, \rightarrow \,\, H\, (H\,z) $ is \emph{divergent} (or \emph{non-productive}), meaning that it generates a node labelled $\bot$ in the value tree (see Figure~\ref{fig:ex-value-tree}).

\begin{figure}[htb]
\begin{center}
\begin{tikzpicture}[level distance = 24pt,sibling distance = 7pt,scale=1,transform shape]
\tikzstyle{every node}=[font=\normalsize]
\tikzset{>=stealth}
\tikzset{edge from parent/.style={draw,-}}
\node(limite) at (0,0) {
\Tree [.{ $f$} 
	[.{$\bot$} ] 
	[.{$f$} 
		[.{$a$} ]
		[.{$f$} 
			[.{$a$} ]
			[.{$f$} 
			[.{$a$} ]
			[.{$f$} 
				[.{$a$} ]
				\edge[dashed,-] ;
			[.{\null} ] 
			] 
			] 
		]	 
	] 
]
};
\node(egale) at (3,0) {\large$= \sup$};
\node(egale) at (4,0) {\Huge$\{$};
\tikzset{level distance = 26pt,sibling distance = 12pt}
\node(arbres) at (7.2,-0.8) {
\Tree [.{ $\bot$} ] ,
\Tree [.{ $f$}  [.{ $\bot$} ]  [.{ $\bot$} ]] ,
\Tree [.{ $f$}  [.{ $\bot$} ]  [.{ $f$} [.{$a$} ] [.{$\bot$} ] ]] , \ldots
};
\node(egale) at (10.5,0) {\Huge$\}$};

\end{tikzpicture}
\caption{The value tree of the recursion scheme $\mathcal{S}$}\label{fig:ex-value-tree}
\end{center}
\end{figure}

Consider now a scheme with production rules
\[\begin{array}{rll}
I & \rightarrow & f (\boxtimes) (F\, f) \\
F\,\varphi & \rightarrow &\varphi\, a\,(F\,\varphi) \\
\end{array}\]
Then it produces the same tree as the previous scheme up to relabelling nodes labelled by $\boxtimes$ by $\bot$. Note that this latter scheme does not lead any divergent computation.

Actually, MSO effective reflection leads to the following general result, showing that one can always design an “equivalent” non-divergent scheme. See also \cite{Haddad12,SW13,HaddadPHD} for alternative proofs of this statement.

\begin{corollary}\label{cor:divergent}
Let $\mathcal{S}$ be an order-$n$ recursion scheme with terminals $\Sigma$ and generating a tree $t$. Then one can construct another order-$n$ scheme $\mathcal{S}_\bot$ with terminals $\Sigma\cup\{\boxtimes\}$ where $\boxtimes:o$ is a fresh symbol of arity $0$ and such that 
\begin{enumerate}
\item The trees $t$ and $t_\bot$ have the same domain;
\item the tree $t_\bot$ generated by $\mathcal{S}_\bot$ does not contain any node labelled $\bot$;
\item for every node $u$, $t_\bot(u)=t(u)$ if $t(u)\neq\bot$ and $t_\bot(u)=\boxtimes$ otherwise.
\end{enumerate}
\end{corollary}

\begin{proof}
Let $\mathcal{S}=\anglebra{\Sigma, \mathcal{N}, \mathcal{R}, I}$ and let $\arobase:o\rightarrow o$ be a fresh terminal symbol of arity $1$. Define $\mathcal{S}_{\arobase}=\anglebra{\Sigma\cup\{\arobase\}, \mathcal{N}, \mathcal{R}_\arobase, I}$ where $\mathcal{R}_\arobase$ is the set of production rules $\{Fx_1\cdots x_n\rightarrow \arobase\, m\mid Fx_1\cdots x_n\rightarrow m \text{ belongs to }\mathcal{R}\}$, \ie we append a symbol $\arobase$ whenever performing a rewriting step. Denote by $t^\arobase$ the tree generated by $\mathcal{S}_\arobase$.

It is fairly simple to notice the following: 
\begin{itemize}
\item $t^{\arobase}$ does not contain any node labelled $\bot$;
\item $t$ is obtained from $t^{\arobase}$ by
\begin{enumerate}
\item replacing any infinite piece of branch (possibly starting from another node than the root) made only of nodes labelled by $\arobase$ by a single node labelled by $\bot$;
\item contracting any finite path made of nodes labelled by $\arobase$.
\end{enumerate}
\end{itemize}

Now consider an MSO formula $\phi$ stating that a node is labelled by $\arobase$, is the source of an infinite sequence of nodes labelled by $\arobase$ and is the child of a node not labelled by $\arobase$ (\ie the node is the first one in an infinite piece of  branch labelled by $\arobase$).
Thanks to Theorem~\ref{theo:reflection-MSO}, we can build a new recursion scheme $\mathcal{S}_{\arobase,\phi}$ that generates $t^{\arobase}_\phi$.

Now one obtains $\mathcal{S}_\bot$ by doing the following modification from $\mathcal{S}_{\arobase,\phi}$:
\begin{enumerate}
\item in any production rule, replace any occurrence of a ground subterm of the form $\underline{\arobase}\, s$ by the ground term $\boxtimes$;
\item in any production rule, replace any occurrence of a ground subterm of the form ${\arobase}\, s$ by the ground term $s$.
\end{enumerate}
Hence, the tree $t_\bot$ produced by $\mathcal{S}_\bot$ is obtained from $t^{\arobase}$ by (1)~replacing any infinite branch made of nodes labelled by $\arobase$ by a single node labelled by $\boxtimes$ and (2)~by contracting any finite path made of nodes labelled by $\arobase$. Therefore $\mathcal{S}_\bot$ is as expected.
\end{proof}

\subsection{An à la Caucal Result for General Schemes}\label{section:unfolding}

A natural extension of the MSO reflection is to use MSO to define new edges in the structure and not simply to mark certain nodes.  This
corresponds to the well-known mechanism of MSO interpretations
\cite{Courcelle94}. Furthermore to obtain back a tree from this new graph, we choose a vertex as a root and we unfold the graph from it. As both MSO interpretation and
unfolding are graph transformations that preserve the decidability of
MSO, combining these two transformations provides a very powerful mechanism for constructing infinite trees with a decidable MSO model-checking problem.

\begin{remark}
If we only use MSO interpretations followed by unfolding to produce trees
(and graphs) starting from the class of finite trees, we obtain the Caucal hierarchy~\cite{Caucal02}. The trees in this hierarchy are exactly the trees generated by \emph{safe} recursion schemes.
\end{remark}

\subsubsection{Main Result}

We first present a definition of MSO interpretations which is tailored
to our setting. An \defin{MSO interpretation} $\mathcal{I}$ over $\Sigma$-labelled ranked trees is
given by a domain formula $\phi_\delta(x)$, a formula $\phi_a(x)$ for
each $a \in \Sigma$ and a formula $\phi_d(x,y)$ for each
direction $d \in \{1,\dots,m\}$ where $m=\max\{\arity{a} \mid a\in\Sigma\}$.

When applied to a $\Sigma$-labelled ranked tree $t$, an MSO interpretation $\mathcal{I}$ produces a graph, denoted $\mathcal{I}(t)$,
whose vertices are the vertices of $t$ satisfying $\phi_\delta(x)$. A
vertex $u$ of $\mathcal{I}(t)$ is labelled by $a$ iff $u$
satisfies $\phi_{a}(x)$ in $t$.  Similarly there exists an edge
labelled by $d \in \{1,\ldots,m\}$ from a vertex $u$ to a vertex
$v$ iff $(t,u,v)\models \phi_d(x,y)$.


We say that the interpretation $\mathcal{I}$ is \emph{well-formed} if for all $\Sigma$-labelled
{trees} $t$, every vertex $u$ of $\mathcal{I}(t)$ is labelled by exactly
one $a \in \Sigma$ and has exactly one out-going edge for each
direction in $\{1,\ldots\arity{a}\}$. 
Here, we restrict our attention to well-formed interpretations, which ensures that, when selecting a root in that graph, the generated tree (by unfolding from the root as explained in Section~\ref{section:unfolding}) is a $\Sigma$-labelled ranked tree. {Given an  MSO interpretation $\mathcal{I}$, one can decide if it is  well-formed. Indeed, the fact that the graph obtained after applying $\mathcal{I}$ is well-formed can be expressed by a first-order formula $\psi$ hence, it follows that there exists an MSO formula $\psi^*$ such that for all tree $t$, $I(t) \models \psi$ if and only if $t \models \psi^*$ (see for instance \cite{Courcelle94}). Using Theorem~\ref{th:Rabin69}, we can construct a parity tree automaton $\mathcal{A}_{\neg \psi^*}$ accepting the trees that do not satisfy $\psi^*$ and, as emptiness is decidable for parity tree automata, the decidability follows.
}

\begin{example}
We revisit examples~\ref{ex:GMC} and \ref{example:reflectionMSO}. Consider the scheme of Example \ref{ex:GMC} and the formula $\phi$ of Example~\ref{example:reflectionMSO} (recall that $\phi$ holds in a node $u$ iff $u$ is the right son of an $f$-labelled node and there is a path from $u$ to an $a$-labelled node which contains an odd number of occurrences of $g$-labelled nodes). MSO reflection consisted in “marking” the nodes where $\phi$ holds.

Consider the MSO interpretation $\mathcal{I}$ which removes
all nodes below a node where $\phi$ holds.
All node labels are preserved. 
Finally all edges are preserved and a loop labelled by {$g$} is 
added to every node where $\phi$ holds.
It is easily seen that $\mathcal{I}$ is a well-formed interpretation.
By applying $\mathcal{I}$ to {the} tree $t$ of example~\ref{ex:GMC} and
then unfolding it from its root, we obtain the tree on the right which is
generated by the scheme on the left.

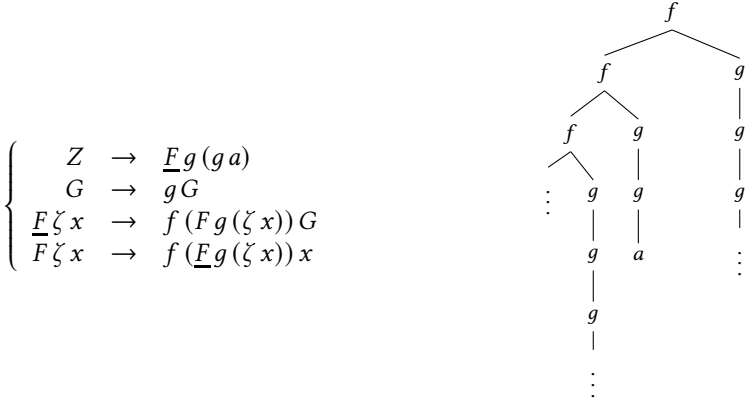
\begin{center}
\begin{minipage}{5.5cm}
$$\left\{\begin{array}{rll}
 Z & \rightarrow & \marked{F} \, g\, (g\,a)\\
 G & \rightarrow & g\, G \\
 \marked{F}\, \zeta \, x & \rightarrow & f \, (F\, g \, (\zeta \, x))\, G\\
 F\, \zeta \, x & \rightarrow & f \, (\marked{F}\, g \, (\zeta \, x))\, x\\
\end{array}\right.  $$ 
\end{minipage}
\hspace{2cm}
\begin{minipage}{4.5cm}
\begin{tikzpicture}[level distance=8mm,level/.style={sibling distance=18mm/#1}]
\footnotesize
\node {$f$}
   child {node {$f$}
         child {node {$f$}  
             child {node {$\vdots$}}
             child {node {$g$} child { node {$g$} child {node {$g$} child {node {$\vdots$}}}}}}
         child {node {$g$} child { node {${g}$} child {node {$a$}}}}}  
   child {node {$g$} child { node {$g$} child {node {$g$} child {node {$\vdots$}}}}};
\end{tikzpicture} 
\end{minipage}\end{center}

\end{example}

More generally, we have the following result (whose proof is given in the next section).

\begin{theorem}\label{theorem:ALaCaucal}
  Let $t$ be a $\Sigma$-labelled ranked tree generated by an order-$n$
  recursion scheme and let $\mathcal{I}$ be a well-formed
  MSO interpretation. For any vertex $r$ in $\mathcal{I}(t)$, the tree generated by $\mathcal{I}(t)$ from the root $r$ can
  be generated by an order-$(n+1)$ recursion scheme.
\end{theorem}

\begin{remark} A natural question is whether every tree
generated by order-$(n+1)$ recursion scheme can be obtained by
unfolding a well-formed MSO interpretation of a tree generated by an
order-$n$ recursion scheme. This is for instance true when considering
the subfamily of safe recursion schemes \cite{KNU02,Caucal02}. 
The answer is negative. {Indeed, a positive answer for possibly unsafe schemes would imply that safe schemes
of any given order are as expressive (for generating trees) as unsafe
ones of the same level, contradicting a result of Parys stating that unsafe schemes are strictly more expressive than safe schemes for generating trees~\cite{Parys12}.}
\end{remark}

\subsubsection{Proof of Theorem~\ref{theorem:ALaCaucal}}

   Let $t$ be a $\Sigma$-labelled tree given by some order-$n$
  recursion scheme $\rscheme$ and let $\mathcal{I}$ be a well-formed
  MSO interpretation given by formulas $\phi_\delta(x)$,
  $\phi_a(x)$ for each $a \in \Sigma$ and 
  $\phi_\ell(x,y)$ for each direction $\ell \in \{1,\dots m\}$ where $m=\max\{\arity{a}\mid a\in\Sigma\}$. Let $r$ be a vertex in $\mathcal{I}(t)$ and let $t'$ be the tree generated by $\mathcal{I}(t)$ from the root $r$. 
  We want to show that $t'$ is generated by  some order-$(n+1)$ scheme. By Theorem~\ref{theo:equi}, it is enough to show that $\mathcal{I}(t)$ restricted to the vertices reachable from $r$ is isomorphic to the $\silent$-closure\footnote{Formally, this means that any sequence of transitions $v_0\era{a}v_1\era{\silent} v_2\era{\silent}\cdots \era{\silent}v_{k-1}\era{\silent}v_k$ such that $v_k$ is not the source of an $\silent$-transition  is replaced by the transition $v_0\era{a}v_k$.} of the transition graph of some order-$(n+1)$ CPDA restricted to its reachable configurations.

The proof starts by annotating $t$ with MSO definable sets, leading to a tree $\bar{t}$ that, thanks to Theorem~\ref{theo:reflection-MSO}, can be generated by a recursion scheme. We then show that tree-walking automata (a class of finite memory device that walk around an input tree) can be used to accept pairs of nodes that are connected by a new edge after applying $\mathcal{I}$ to $t$. Finally, we build the announced $(n+1)$-CPDA $\mathcal{A}$ by mimicking an $n$-CPDA generating $\bar{t}$ and simulating the previous tree-walking automata.
  
  \paragraph{Notations}
 For all $\Sigma$-labelled trees $t$, all nodes $u$ and $v$ of $t$, we let
  $\treemarkednode{t}{u,v}$ be the tree obtained from $t$ by marking
  the pair $(u,v)$. Formally, $\treemarkednode{t}{u,v}$ is the $(\Sigma
  \times 2^{\{1,2\}})$-labelled tree such that
  $\dom{\treemarkednode{t}{u,v}}=\dom{t}$ and for $w \in
  \dom{\treemarkednode{t}{u,v}}$,
  $\treemarkednode{t}{u,v}(w)=(t(w),X)$ where $1 \in X$ iff $w=u$ and
  $2 \in X$ iff $w=v$.

   Similarly
  we define $\treemarkednode{t}{u,\bullet}$ (\resp $\treemarkednode{t}{\bullet,v}$, 
  \resp $\treemarkednode{t}{\bullet,\bullet}$), for a fresh symbol $\bullet$, as the tree obtained by marking $u$ by $1$ (\resp $v$ by $2$, \resp marking no node). \Ie $\bullet$ means here that no node is marked with the corresponding index ($1$ and/or $2$).
  
  Using Theorem~\ref{th:Rabin69}, one can construct for any formula
  $\phi(x_1,x_2)$ a parity tree automaton $\mathcal{B}_\phi$ that accepts
  $\treemarkednode{t}{u,v}$ iff $(t,u,v) \models \phi(x_1,x_2)$. For all $\ell \in \{1,\dots,m\}$, we write $\mathcal{B}_\ell$ for the parity
 tree automaton corresponding to the formula $\phi_\ell(x,y)$ of $\mathcal{I}$ and we let $Q_\ell$ be its finite set of control states and $\Delta_\ell$ be its transition relation.

\paragraph{Annotation of $t$ by MSO definable sets}

 The first step of the construction is to annotate
 $t$ with information concerning essentially the behaviour of the
 automata $\mathcal{B}_\ell$ on the subtrees of $t$. The resulting annotated version of $t$ is denoted $\bar{t}$. More precisely,
 the annotated tree $\bar{t}$ has, for each node $u$ such that $(t,u)\models\phi_\delta(x)$, the following finite information:
 \begin{itemize}
 \item the unique $a \in \Sigma$ such that $(t,u) \models \phi_a(x)$.
  Unicity is by definition of a well-formed interpretation.
 \item  $d_{\uparrow} \in \{1,\dots,m\} \cup \{\racine\}$ which is the 
direction from the father of the curent node to the current node, \emph{\ie} $u$ is of the form $u' d_{\uparrow}$ 
 for some $u' \in \dom{t}$, and $\racine$ if the current
node is the root.
 \item for each $\ell \in \{1,\dots,m\}$, we have:
 \begin{itemize}
 \item $i_\ell \in \{\uparrow,\downarrow,\circlearrowleft,\bot\}$ such that
   \begin{itemize}
     \item $i_\ell = \bot$ iff there is no node $v$ such that $(t,u,v) \models \phi_\ell(x,y)$,
     \item $i_\ell = \downarrow$ iff there is a unique node $v$ such that $(t,u,v) \models \phi_\ell(x,y)$ and $v$ is below $u$,
     \item $i_\ell = \uparrow$ iff there is a unique node $v$ such that $(t,u,v) \models \phi_\ell(x,y)$ and $v$ is not below $u$,
     \item $i_\ell = \circlearrowleft$ iff $(t,u,u) \models \phi_\ell(x,y)$.
    \end{itemize}

  \item the set $R_\ell$ of states  $q \in Q_\ell$  such that there exists
an accepting run of $\mathcal{B}_\ell$ starting from $q$ on the subtree of $t$ rooted at $u$.
\item the set $S_\ell$ of pairs $(d,q) \in \{1,\dots,m\} \times
  Q_\ell$ such that there exists an accepting run of $\mathcal{B}_\ell$
  starting from $q$ on the subtree of $\treemarkednode{t}{\bullet,\bullet}$ rooted at $ud$,

\item the set $T_\ell$ of pairs $(d,q) \in \{1,\dots,m\} \times
  Q_\ell$ such that there exists a node $v$ below $ud$ such that   $\mathcal{B}_\ell$ has an accepting run starting from $ud$ 
  on the subtree of $\treemarkednode{t}{\bullet,v}$ rooted at $ud$.
 \end{itemize}
\end{itemize}
 Let $\Sigma'$ be the resulting  labelling alphabet of $\bar{t}$.
As the information annotated on $\bar{t}$ is MSO definable in $t$,
 we know from Theorem~\ref{theo:reflection-MSO}
that $\bar{t}$ is generated by some order-$n$ recursion scheme.

\paragraph{Replacing MSO formulas on $t$ by tree-walking automata working on $\bar{t}$.}

Fix a direction $\ell \in \{1,\dots,m\}$.  Thanks to the extra
information available on $\bar{t}$, it is possible to decide if a pair
of nodes $(u,v)$ satisfies the formula $\phi_\ell(x,y)$ on $t$ using
a deterministic tree-walking automaton running on $\bar{t}$.
Intuitively a tree-walking automaton is a finite memory device that walks around an input tree, choosing what move to make according to its current state and the node label.

Formally, a \defin{deterministic tree-walking automaton} (introduced in the early seventies by Aho and Ullman~\cite{AhoU71}; see also~\cite{EngelfrietHB99}) working on
$\Sigma$-labelled trees is a tuple $\mathcal{W} = (Q,q_0,F,\delta)$ where $Q$ is
a finite set of states, $q_0 \in Q$ is the initial state, $F$ is
a set of final states and $\delta$ is a transition function.  The
transition function associates to a pair $(p,a) \in Q\times
\Sigma$, corresponding respectively to the current state and node label,
 a pair $(q,x) \in Q \times
(\{\uparrow,-\,\}\cup\{1,\dots,m\})$ where $q$ is the new
state and $x$ is a movement to perform.  Intuitively $-$
corresponds to “stay in the current node”, $\uparrow$ to “go to
the parent node” and $d \in \{1,\dots,m\}$ corresponds to “go
to the $d$-child”.
A pair of nodes $(u,v)$ is \emph{accepted} by $\mathcal{W}$ if it can reach $v$ in a final state starting from $u$ in the initial state.

 We claim that
there exists a deterministic tree-walking automaton $\mathcal{W}_\ell$ such that, for
any pair $(u,v)$ of nodes of $t$, we have:
\begin{center}
 $\mathcal{W}_\ell$ accepts $(u,v)$ in $\bar{t}$ iff $(t,u,v) \models
\phi_\ell(x,y)$.
\end{center}

The automaton $\mathcal{W}_\ell$ works in two phases: during the first phase the automaton only goes 
up in the tree and during the second phase it only goes down 
in the tree. Both phases can potentially be empty. In fact,
to accept a pair $(u,v)$ the automaton will first go up to the greatest
common ancestor of $u$ and $v$ and down to $v$.
 
Assume that $\mathcal{W}_\ell$ started at a node $u$ and 
denote by  $v$ the unique node (if it exists) such 
that $(t,u,v) \models \phi_\ell(x,y)$.

\smallskip\noindent\textbf{Initialisation.} The automaton is in its initial
state $q_0$ at node $u$ that is labelled by the tuple  $(a,d_\uparrow,(i_k,R_k,S_k,T_k)_{k \in \{1,\dots,m\}})$.
The automaton checks in which of the following four cases it is:

\begin{itemize}
\item \underline{The node $v$ does not exists (\ie $i_\ell=\bot$).}
No transition is defined.
\item \underline{The node $v$ is equal to $u$ (\ie $i_\ell=\circlearrowleft$).}  
The automaton goes to the accepting state.
\item \underline{The node $v$ is not below $u$ (\ie $i_\ell=\uparrow$).}
The automaton begins the first phase while memorising
the set $X$ of states $q \in Q_\ell$ such that $\mathcal{B}_\ell$ admits an accepting run starting from $q$ on the subtree of $\treemarkednode{t}{u,\bullet}$ rooted at $u$. This set can easily be computed from $\Delta_\ell$
and $S_\ell$.
\item \underline{The node $v$ is below $u$ (\ie $i_\ell=\downarrow$).} 
The automaton  begins the second phase. It computes the unique direction $d$ and the set $Y$ of
states $q \in Q_\ell$ such that:
\begin{itemize}
\item  $\mathcal{B}_\ell$ admits an
accepting run on the tree $\treemarkednode{t}{u,\bullet}$ deprived of the
nodes below $ud$ and assigning state $q$ to $ud$,
\item  $(d,q) \in T_\ell$.
\end{itemize}

\end{itemize}

\smallskip\noindent\textbf{First phase.}
The automaton is at some node $w$ and stores the set $X$ of states $q
\in Q_\ell$ such that $\mathcal{B}_\ell$ admits an accepting run starting from $q$ on
the subtree of $\treemarkednode{t}{u,\bullet}$ rooted at $w$. The label of the node
$w$ is $(a,d_\uparrow,(i_k,R_k,S_k,T_k)_{k \in
  \{1,\dots,m\}})$.  The automaton goes up in the tree (while
remembering $d_\uparrow$ and $X$) to a node $w'$, {the father of $w$,} whose label is
$(a',d_\uparrow',(i_k',R_k',S_k',T_k')_{k \in
  \{1,\dots,m\}})$. 

There are now three possible cases:  
\begin{itemize}
\item \underline{The node $v$ is the current node.} This is the case iff there exists a
  state $q \in R_\ell'$ and a transition in $\Delta_\ell$ starting in
  state $q$ with $(a,\{2\})$ as label and associating state
  $q_i$ to the $i$-child, such that:
\begin{itemize} 
\item  $q_{d_\uparrow}$ belongs to $X$ ,
\item  for all $i
  \neq d_\uparrow$, $(i,q_i) \in S_\ell$.
\end{itemize}
 In this case, the automaton goes to 
the accepting state. 
\item \underline{The node $v$ is below the $j$-child of $w'$ for some $j \in \{1,\dots,\arity{a'}\}$.}
This is the case iff there exists a
  state $q \in R_\ell'$, a transition in $\Delta_\ell$ starting in
  state $q$ with $(a',\emptyset)$ as label and associating state
  $q_i$ to the $i$-child, and such that there exists $j \neq d_\uparrow \in \{1,\dots,\arity{a'}\}$ with
\begin{itemize}
\item  $q_{d_\uparrow}$ belongs to $X$,
\item for all $i
  \neq j,d_\uparrow$ one has $(i,q_i) \in S_\ell'$,
\item $({j},q_{j}) \in T_\ell'$.
\end{itemize}
 In this case, the automaton begins the second
phase while memorising the set $Y$ of all states $q_{j}$ 
matching this definition together with the direction \footnote{Due to
the restriction imposed on $\phi_\ell(x,y)$ by the fact that $\mathcal{I}$
is a well-formed MSO interpretation, there cannot be two different directions.
Otherwise, we would have $v_1 \neq v_2$ such that $(t,u,v_1) \models \phi(x,y)$
and $(t,u,v_2) \models \phi(x,y)$.} $j$. 
\item \underline{The node $v$ is not below {$w'$}.} This is the case when the two
  previous cases do not apply. The automaton updates $X$
  using $d_\uparrow$, $\Delta_\ell$ and the former value of $X$ and goes on in the first phase.
\end{itemize}

\smallskip\noindent\textbf{The second phase}
The automaton is at some node $w$ and stores a direction $d$ and the
set $Y$ of states $q \in Q_\ell$ such that:
\begin{itemize}
\item  $\mathcal{B}_\ell$ admits an
accepting run on the tree $\treemarkednode{t}{u,\bullet}$ deprived of the
nodes below $wd$ and assigning state $q$ to $wd$
\item  there exists a node $v$ below  $wd$ such that $\mathcal{B}_\ell$
admits an accepting run starting in state $q$ on the subtree of $\treemarkednode{t}{\bullet,v}$
rooted at $wd$.
\end{itemize}

  The automaton goes down in direction $d$ (while remembering $Y$)
  to a node $w'=wd$ whose label is
  $(a',d,(i_k',R_k',S_k',T_k')_{k \in
    \{1,\dots,m\}})$. 

There are now two cases:
\begin{itemize}
\item\underline{The node $v$ is below the $d'$-child of $w'$ for some $d' \in \{1,\dots,\arity{a'}\}$.}
This is the  case iff there exists a state  $q \in Y$ and a transition in $\Delta_\ell$ starting in
  state $q$ with $(a,\emptyset)$ as label and associating states
  $q_i$ to the $i$-child such that
\begin{itemize}
\item for all $i
  \neq d'$, $(i,q_i) \in S_\ell'$,
\item $({d'},q_{d'}) \in T_\ell'$.
\end{itemize}
 In this case, the automaton updates the set $Y$ with all states $q_{d'}$ matching this condition, stores the direction $d'$ and goes on in the second phase.
\item\underline{The node $v$ is the current node.} This is precisely when the previous
case does not hold. The automaton moves to an accepting state.
\end{itemize}

\paragraph{Construction of the order-${(n+1)}$ CPDA $\mathcal{A}$ generating $t'$.}

  By Theorem~\ref{theo:equi}, there exists an order-$n$ CPDA
  $\mathcal{C}=\anglebra{\{1,\dots,m\} \cup\{\silent\}, \Gamma,
    Q,\delta, q_\iota,F}$, and a mapping $\rho: Q \rightarrow
  \Sigma'$ such that $\bar{t}$ is the tree generated by $\mathcal{C}$ and
  $\rho$. 

Hence, for every node $u=d_1 \cdots d_k \in \dom{t}$, there exists  a  unique 
sequence of configurations $(q_0,s_0),\ldots,(q_k,s_k)$ of $\mathcal{C}$ such that:
\begin{itemize}
\item there exists a path in $\transgraph{\mathcal{C}}$ labelled by $\silent^*$
from the initial configuration  to $(q_0,s_0)$,
\item for all $i \in \{0,\dots,k\}$, $(q_i,s_i)$ is not the source of an $\silent$-labelled arc in $\transgraph{\mathcal{C}}$,
\item for all $i \in \{0,\dots,m-1\}$, there exists a path labelled by a word in $d_{i+1} \silent^*$
from $(q_i,s_i)$ to $(q_{i+1},s_{i+1})$ in $\transgraph{\mathcal{C}}$.
\end{itemize}

Such a sequence can be encoded as an order-$(n+1)$ stack $s_u$ using symbols from $\Gamma\cup A$ (recall that the $s_i$ are order-$n$ stacks) in the following way:
\[
s_u = [{\tilde{s_0}}, \tilde{s_1}, \ldots, \tilde{s_{k-1}}, \pushlk{1}{q_{k}}(s_{k}) ]
\]
\noindent
where for all $i \in \{0,\dots,k-1\}$, $\tilde{s_i}= \pushlk{1}{d_{i+1}}(\pushlk{1}{q_i}(s_i))$.

The automaton $\mathcal{A}$ works on stacks corresponding to some
$s_u$ for some $u \in \dom{t}$. Its set of control states contains a distinguished state $q_\star$ and all states of the tree-walking automata
$(\mathcal{W}_\ell)_{\ell \in \{1,\dots,m\}}$ which we assumed to be disjoint.
The configurations of $\mathcal{A}$ that are source of 
non-$\silent$-labelled arcs will be of the form $(q_\star,s_u)$ for
some $u \in \dom{t}$. The intended behaviour of $\mathcal{A}$
is that, for some $\ell \in \{1,\dots,m\}$, if $(t,u,v) \models \phi_{\ell}(x,y)$
then $\mathcal{A}$ can go from the configuration $(q_\star,s_u)$ to 
the configuration $(q_\star,s_v)$ by a path labelled by $\ell \silent^*$.

First 
$\mathcal{A}$  moves by an $\ell$-labelled transition
to  the configuration
$(q^\ell_0,s_u)$ where $q^\ell_0$ is the initial state of the
tree-walking automaton $\mathcal{W}_\ell$. {Recall that the information about the location of the
vertex $v$ is given by the annotations in $\rho(\topn{1}(s_u))$.}

In a configuration of the form $(p,s_u)$ with $p$ a state of $\mathcal{W}_\ell$,
$\mathcal{A}$ simulates the behaviour of $\mathcal{W}_\ell$ on $\bar{t}$ at node
$u$ in state $p$ by a sequence of $\silent$-transitions. As
$\bar{t}(u)=\rho(\topn{1}(s_u))$, $\mathcal{A}$ can compute the
transition taken by the automaton $\mathcal{W}_\ell$ on $\bar{t}$ at node $u$ in
state $p$. The behaviour of $\mathcal{A}$ will be such that if
$\mathcal{W}_\ell$ goes from $(p,u)$ to $(q,u')$ in one step then $\mathcal{A}$
will go through a sequence of $\silent$-transitions from $(p,s_u)$ to $(q,{s_{u'}})$.

We distinguish several cases depending on the movement performed by $\mathcal{W}_\ell$.
\begin{itemize}
\item $\mathcal{W}_\ell$ stays in the current node in state $q$. Then $\mathcal{A}$
changes its state to $q$ by an $\silent$-transition.
\item $\mathcal{W}_\ell$ goes to its parent node in state $q$ (\ie $u=u'd$ and 
$\mathcal{W}_\ell$ ends up in $u'$ in state $q$). Then $\mathcal{A}$
performs  $\popn{n+1}$ followed by a $\popn{1}$ and moves to state $q$.
The stack content of $\mathcal{A}$ is now $\popn{1}(\popn{n+1}(s_u))=s_{u'}$ hence $\mathcal{A}$ is in configuration $(q,s_{u'})$.
\item $\mathcal{W}_\ell$ goes to its $d$-child in state $q$ (\ie $u'=ud$ and 
$\mathcal{W}_\ell$ ends up in $ud$ in state $q$). 
Assume that $s_u$ is equal to:
\[ [{\tilde{s_0}},\tilde{s_1}, \ldots, \tilde{s_{k-1}}, \pushlk{1}{q_{k}}(s_{k})] \]
and that $s_{ud}$ is of the from:
\[ [{\tilde{s_0}},\tilde{s_1}, \ldots, \pushlk{1}{d}(\pushlk{1}{q_{k}}(s_{k})),\pushlk{1}{q_{k+1}}(s_{k+1})]   \]
As $\mathcal{C}$ generates $\bar{t}$, there exists a path $\pi$ in $\transgraph{\mathcal{C}}$ from $(q_{k},s_k)$
to $(q_{k+1},s_{k+1})$ labelled by $d \silent^*$.

Then $\mathcal{A}$
starts by performing a $\pushlk{1}{d}$ followed by $\pushn{n+1}$, $\popn{1}$ and $\popn{1}$.
At this point the stack is:
\[ [{\tilde{s_0}},\tilde{s_1}, \ldots, \pushlk{1}{d}(\pushlk{1}{q_{k}}(s_{k})),s_{k}]. \]
Then $\mathcal{A}$ simulates the order-$n$ operations of $\mathcal{C}$ along the path $\pi$ using $\silent$-transitions.
When no $\silent$-transition of $\mathcal{C}$ can be applied, $\mathcal{A}$ performs a last $\silent$-transition where it performs a $\pushlk{1}{q_{k+1}}$ and goes to state  $q_{k+1}$ hence, reaching configuration $(q,s_{ud})$.
\end{itemize}

Eventually $\mathcal{A}$ will reach a configuration of the form $(q_f^\ell,s_v)$ where
$q_f^\ell$ is the accepting state of $\mathcal{W}_\ell$. It then goes to the state $q_\star$.

From its initial configuration, $\mathcal{A}$ deterministically 
 builds the stack $s_{u_0}$
(which correspond to the vertex $u_0$ from which $\mathcal{I}(t)$ is
unfolded) by using sequence of $\silent$-transitions  and goes to the state $q_\star$.

By construction, we have that the $\silent$-closure of $\mathcal{A}$
restricted to the vertices reachable from its initial configuration
is isomorphic to $\mathcal{I}(t)$ restricted to the vertices reachable
from $u_0$. The isomorphism simply maps a configuration $(q_\star,s_u)$ of $\mathcal{A}$ to $u \in \dom{t}$.

To generate $t'$ it remains to retrieve the node label, \ie define a function from the set of control states of $\mathcal{A}$ into $\Sigma$. As defined so far $\mathcal{A}$ does not have the necessary information for that: indeed, relevant configurations all have $q_\star$ as their control state. But it is not hard to replace $q_\star$ by a variant in $\{q_\star^a\mid a\in \Sigma\}$: instead of going to $(q_\star,s_u)$, $\mathcal{A}$ goes to $(q_\star^a,s_u)$ where $a$ is the first component labelling $u$ in $\bar{t}$ which can easily be recover by $\mathcal{A}$ (it suffices to do a $\popone$ operation and the information is then in the topmost symbol). Then, if $\rho$ is defined by $\rho(q^a_\star)=a$, the tree $t'$ is generated by $\mathcal{A}$ and $\rho$, which concludes the proof.

\section{Selection}\label{section:selection}

We now focus on a more general problem than global model-checking, often known as the \emph{synthesis} problem. For simplicity we start by a case study motivated by the famous question of the definability of choice functions on the infinite binary tree \cite{GurevichS83a}.

\subsection{A Case Study: Choice Functions}

Consider an infinite binary tree $t$ in which some nodes are coloured (\ie labelled) in red; we also assume that the domain of $t$ is included in $\{1,2\}^*$ and we identify direction $1$ with the left and direction $2$ with the right. Assume that $t$ satisfies the following extra property (we say that $t$ is \emph{well-formed}): every subtree contains at least one red node. Our goal is for all nodes $x$ to \emph{choose/select} a red node $y$ in the subtree $t[x]$ rooted at $x$ (see Figure~\ref{fig:choice} for an illustration).

\begin{figure}
\begin{center}
\begin{tikzpicture}
\fill[bottom color=LimeGreen!100!black!20,top color=LimeGreen!100!black!70] (0,0) -- (-3,-6) -- (3,-6) -- cycle ;
\draw[LimeGreen,thick] (0,0) -- (-3,-6);
\draw[LimeGreen,thick] (0,0) -- (3,-6);

\node[draw=black,circle,fill=black,scale=0.5] (n) at (-0,0) {};
\node[draw=black,circle,fill=black,scale=0.5] (n0) at (-0.3,-0.6) {};
\node[draw=black,circle,fill=black,scale=0.5] (n1) at (0.3,-0.6) {};
\node[draw=black,circle,fill=black,scale=0.5] (n00) at (-0.6,-1.2) {};
\node[draw=black,circle,fill=black,scale=0.5] (n01) at (-0.1,-1.2) {};
\node[draw=black,circle,fill=black,scale=0.5] (n10) at (0.1,-1.2) {};
\node[draw=black,circle,fill=black,scale=0.5] (n11) at (0.6,-1.2) {};
\draw[thick] (n) -- (n0);\draw[thick] (n) -- (n1);
\draw[thick] (n0) -- (n00);\draw[thick] (n0) -- (n01);
\draw[thick] (n1) -- (n10);\draw[thick] (n1) -- (n11);
\node[] (n100) at (0.3,-1.8) {};
\node[] (n101) at (-0.1,-1.8) {};
\draw[thick,dotted] (n10) -- (n100);\draw[thick,dotted] (n10) -- (n101);

\node[draw=black,circle,scale=0.5,fill=Crimson] (f1) at (-0.65,-5.3) {};
\node[draw=black,circle,scale=0.5,fill=Crimson] (f2) at (0.1,-1.2) {};
\node[draw=black,circle,scale=0.5,fill=Crimson] (f3) at (1,-3.3) {};
\node[draw=black,circle,scale=0.5,fill=Crimson] (f4) at (0.7,-5.3) {};
\node[draw=black,circle,scale=0.5,fill=Crimson] (f5) at (-0.7,-2) {};
\node[draw=black,circle,scale=0.5,fill=Crimson] (f6) at (-1.95,-5.8) {};

{
\fill[
bottom color=Orange!100!black!20,
top color=Orange!100!black!70] (-1,-3.1) -- (-2.45,-6) -- (0.45,-6) -- cycle ;
\draw[Orange,thick] (-1,-3.1) -- (-2.45,-6);
\draw[Orange,thick] (-1,-3.1) -- (0.45,-6);
\node[draw=black,circle,scale=0.5] (n) at (-1,-3.5) {};
\node[draw=black,circle,scale=0.5] (n0) at (-1.4,-4.1) {};
\node[draw=black,circle,scale=0.5] (n1) at (-0.6,-4.1) {};
\node[draw=black,circle,scale=0.5] (n00) at (-1.65,-4.7) {};
\node[draw=black,circle,scale=0.5] (n01) at (-1.15,-4.7) {};
\node[draw=black,circle,scale=0.5] (n10) at (-0.85,-4.7) {};
\node[draw=black,circle,scale=0.5] (n11) at (-0.35,-4.7) {};
{\draw[thick,->]  (n) -- (n0);}
\draw[thick,->]  (n) -- (n1);
{\draw[thick,->]  (n0) -- (n00);}
\draw[thick,->]  (n0) -- (n01);
\draw[thick,->]  (n1) -- (n10);
{\draw[thick,->]  (n1) -- (n11);}
\draw[thick,->]  (n10) -- (f1);
\node[draw=black,circle,scale=0.5,fill=Crimson] (f1) at (-0.65,-5.3) {};
\node[draw=black,circle,scale=0.5,fill=Crimson] (f6) at (-1.95,-5.8) {};
{
	\node (every) at (0,-7) {\textcolor{Crimson}{Specification: }for all nodes $x$ \emph{choose} a \textcolor{Crimson}{red} node $y$ in $t[x]$};
	\draw[->,color=black,thick,dotted] (-0.1,-6.8) to[out=135,in=-90]  (n);
	\draw[->,color=black,thick,dotted] (3.1,-6.8) to[out=135,in=-15]  (f1);
}
}

\end{tikzpicture}
\caption{Choice function}\label{fig:choice}
\end{center}
\end{figure}
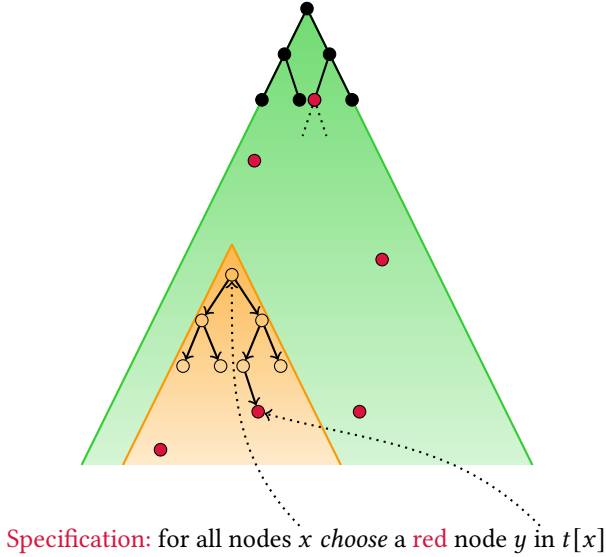

A \defin{choice function} is a function that with any node $x$ associates such a $y$. An \defin{MSO choice function} for $t$  is a formula describing a choice function, \ie a formula $\phi(x,y)$ with two first-order free variables such that $$\forall x\in t\ \ \exists !\,y\text{ s.t. } y \text{ is red, }x< y\text{ and } \phi(x,y)$$

{Guverich and Shelah proved that there does not exist an MSO formula $\varphi(x,X)$ such that for all non-empty set of nodes $U$ of the infinite complete binary tree $t_2$, $(t_2,u,U) \models \varphi(x,X)$ for exactly one node $u$ in $U$. In other terms, there is no MSO formula selecting (on the infinite complete binary tree) exactly one element for each non-empty subset. This result was later re-proven in a more elementary way by Carayol and Löding. Moreover they also exhibited in \cite[Theorem~6]{CarayolL07} a tree generated by a (safe) recursion scheme on which finding an MSO choice function (in the sense of the present article) fails.}

\begin{theorem}{\cite{GurevichS83a,CarayolL07}\cite[Proposition~4.3]{CarayolHDR}}\label{theo:choice}
There is a well-formed tree generated by an order-$3$ (safe) recursion scheme for which no MSO choice function exists.
\end{theorem}

Instead of using an MSO formula to describe a choice function, one can do the following (see Figure~\ref{fig:choice1} for an illustration).
Consider a partition $X_1\uplus X_2$ of the set of nodes of $t$, and think of the nodes in $X_1$ (\emph{resp.} $X_2$) as those where one should first go down to the left (\emph{resp.} right) in order to find a red node. A partition $(X_1,X_2)$ describes a choice function iff the following holds. 
For any node $u\in t$, define the sequence $u_0,u_1,u_2,\ldots$ by letting $u_0=u$ and $u_{i+1} = u_i d_i$ where $d_i=1$ if $u_{i}\in X_1$ and $d_i=2$ if $u_{i}\in X_2$: \ie $u_0,u_1,u_2,\ldots$ is the sequence of nodes visited starting from $u$ and following the directions indicated by $X_1\uplus X_2$. Then for some $k$, $u_k$ is red.

\begin{figure}
\begin{center}
\begin{tikzpicture}
\fill[
bottom color=LimeGreen!100!black!20,
top color=LimeGreen!100!black!70] (0,0) -- (-3,-6) -- (3,-6) -- cycle ;
\draw[LimeGreen,thick] (0,0) -- (-3,-6);
\draw[LimeGreen,thick] (0,0) -- (3,-6);

\node[draw=black,circle,scale=0.5,fill=Crimson] (f1) at (-0.65,-5.3) {};
\node[draw=black,circle,scale=0.5,fill=Crimson] (f2) at (0.1,-1.2) {};
\node[draw=black,circle,scale=0.5,fill=Crimson] (f3) at (1,-3.3) {};
\node[draw=black,circle,scale=0.5,fill=Crimson] (f4) at (0.7,-5.3) {};
\node[draw=black,circle,scale=0.5,fill=Crimson] (f5) at (-0.7,-2) {};
\node[draw=black,circle,scale=0.5,fill=Crimson] (f6) at (-1.95,-5.8) {};

{
\fill[
bottom color=Orange!100!black!20,
top color=Orange!100!black!70] (-1,-3.1) -- (-2.45,-6) -- (0.45,-6) -- cycle ;
\draw[Orange,thick] (-1,-3.1) -- (-2.45,-6);
\draw[Orange,thick] (-1,-3.1) -- (0.45,-6);
\node[draw=black,circle,scale=0.5] (n) at (-1,-3.5) {};
\node[draw=black,circle,scale=0.5] (n0) at (-1.4,-4.1) {};
\node[draw=black,circle,scale=0.5] (n1) at (-0.6,-4.1) {};
\node[draw=black,circle,scale=0.5] (n00) at (-1.65,-4.7) {};
\node[draw=black,circle,scale=0.5] (n01) at (-1.15,-4.7) {};
\node[draw=black,circle,scale=0.5] (n10) at (-0.85,-4.7) {};
\node[draw=black,circle,scale=0.5] (n11) at (-0.35,-4.7) {};
\draw[thick,->]  (n) -- (n1);
\draw[thick,->]  (n0) -- (n01);
\draw[thick,->]  (n1) -- (n10);
\draw[thick,->]  (n10) -- (f1);
\node[draw=black,circle,scale=0.5,fill=Crimson] (f1) at (-0.65,-5.3) {};
\node[draw=black,circle,scale=0.5,fill=Crimson] (f6) at (-1.95,-5.8) {};
{\node[draw=black,circle,scale=0.5] (n000) at (-1.8,-5.3) {};\draw[thick,->]  (n00) -- (n000); \draw[thick,->]  (n000) -- (f6);}
{\node (n111) at (-0.05,-5.45) {};\draw[dotted,thick,->]  (n11) -- (n111);}
{\node (n010) at (-1.3,-5.45) {};\draw[dotted,thick,->]  (n01) -- (n010);}
{\node (n1010) at (-0.85,-5.95) {};\draw[dotted,thick,->]  (f1) -- (n1010);}

\node (X0) at (-4,-4) {$X_2$};\draw[blue,->,thick]  (X0) -- (n);\draw[blue,->,thick]  (X0) -- (n10);\draw[blue,->,thick]  (X0) -- (n0);
}

\end{tikzpicture}
\caption{Partition defining a choice function}\label{fig:choice1}
\end{center}
\end{figure}
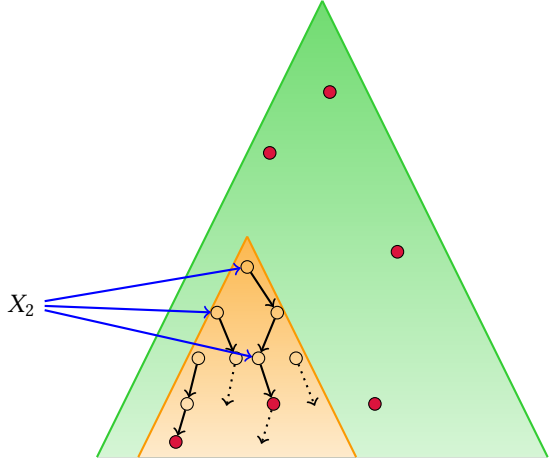

\begin{remark}
Note that a partition $X_1\uplus X_2$ always exists on a well-formed tree. Indeed, for every node $u$, consider the minimal depth of a red node in the left subtree and the minimal depth of a red node in the right subtree: if the smallest depth is in the left subtree one lets $u\in X_1$ otherwise one lets $u\in X_2$.
\end{remark}

It directly follows from Theorem~\ref{theo:choice} that a partition defining a choice function cannot be captured by an MSO formula.

\begin{corollary}{}\label{cor:choice}
There is a well-formed tree generated by an order-$3$ (safe) recursion scheme  such that, for all MSO formula $\phi(x)$, the sets $X_1=\{u\mid\phi(x)\text{ holds in }u\}$ and $X_2=\{u\mid\phi(x)\text{ does not hold in }u\}$ do not define a choice function on $t$.
\end{corollary}

In the same spirit as for global model-checking we propose an exogeneous and an endogeneous approach to the previous problem.
\begin{itemize}
\item \emph{Exogenous} approach: Given a $\Sigma$-labelled well-formed tree $t:\dom{t}\rightarrow \Sigma$, output a description by means of a word acceptor device of a subset $X_1\subseteq \dom{t}$ of nodes such that $(X_1,\dom{t}\setminus X_1)$ defines a choice function for $t$.
\item \emph{Endogenous} approach: Given a $\Sigma$-labelled well-formed tree $t:\dom{t}\rightarrow \Sigma$, output a finite description of a $
(\Sigma\times\{1,2\})$-labelled tree $t_{ch}:\dom{t}\rightarrow\Sigma\times\{1,2\}$ such that $t_{\mathrm{ch}}$ and ${t}$ have the same domain, and such that $X_1=\{u\mid u\in \dom{t} \text{ and }t_{\mathrm{ch}}(u)=(a,1) \text{ for some } a\in\Sigma\}$ and $X_2=\{u\mid u\in \dom{t} \text{ and }t_{ch}(u)=(a,2) \text{ for some } a\in\Sigma\}$ define a choice function for $t$. 
\end{itemize}

Contrasting with the impossibility result stated in Corollary~\ref{cor:choice} we have the following result that follows from a more general result (see Theorem~\ref{theo:selection} below).

\begin{corollary}\label{cor:choice-cpda}
Let $t$ be a well-formed tree generated by an order-$n$ recursion scheme $\rscheme$.
\begin{enumerate}
\item There is an algorithm that builds from $\rscheme$ an order-$n$ CPDA $\mathcal{A}$ such that $(L(\mathcal{A}),\dom{t}\setminus L(\mathcal{A}))$ defines a choice function for $t$.
\item There is an algorithm that builds from $\rscheme$ an order-$n$ recursion scheme $\rscheme_{\mathrm{ch}}$ that generates a tree $t_{\mathrm{ch}}$ defining a choice function for $t$.
\end{enumerate}
\end{corollary}

\begin{remark}
As for the reflection property we note that, in the previous statement, item (2) implies item (1).
\end{remark}

 {The corollaries~\ref{cor:choice} and~\ref{cor:choice-cpda} might seem, at first sight, contradictory. However they concern two orthogonal representations of choice functions: via MSO definability and via CPDA respectively. Applying Corollary~\ref{cor:choice-cpda} to the tree generated by a recursion scheme of Corollary~\ref{cor:choice} leads to a new recursion scheme of the same order but enriched with additional information.}

\subsection{Effective Selection Property}

We now introduce the effective selection property and first present an exogenous approach.

\begin{definition}{(MSO selection problem: exogeneous approach)}
Let $\phi(X_1,\cdots,X_\ell)$ be an MSO formula with $\ell$ second-order free variables, and let $t:\dom{t}\rightarrow \Sigma$ be a $\Sigma$-labelled ranked tree. The \defin{MSO selection problem} is to decide whether the formula $\exists X_1\ldots \exists X_\ell\ \phi(X_1,\cdots,X_\ell)$ holds in $t$, and in this case to output a description, by means of word acceptor devices, of $\ell$ sets $U_1,\cdots,U_l\subseteq \dom{t}$ such that $(t,U_1,\dots,U_\ell)\models\phi(X_1,\dots,X_\ell)$.
\end{definition}

We now give the endogenous approach. The idea is to describe subsets of nodes $U_1,\cdots U_\ell$, such that $(t,U_1,\dots,U_\ell)\models\phi(X_1,\dots,X_\ell)$, by marking every node with a tuple of $\ell$ Booleans (a node $u$ belongs to $U_i$ iff the $i$-th Boolean is $1$ in the tuple labelling $u$). See Figure~\ref{fig:selection} for an illustration of the following definition when $\ell=2$.

\begin{definition}{(MSO selection problem: endogeneous approach)}
Let $\phi(X_1,\cdots,X_\ell)$ be an MSO formula with $\ell$ second-order free variables, and let $t:\dom{t}\rightarrow\Sigma$ be a $\Sigma$-labelled ranked tree. Call $\Xi$ the ranked alphabet $\Sigma\times \{0,1\}^\ell$ defined by letting $\arity{(a,b_1,\ldots,b_\ell)}=\arity{a}$. The \defin{MSO selection problem} is to decide whether the formula $\exists X_1\ldots \exists X_\ell\ \phi(X_1,\cdots,X_\ell)$ holds in $t$, and if so to output a $\Xi$-labelled ranked tree $t_\phi:\dom{t_\phi}\rightarrow \Xi$ such that the following holds:
{
\begin{enumerate}
\item The trees $t$ and $t_\phi$ have the same domain and for every node $u$, one has $t(u)=\pi(t_\phi(u))$ where $\pi$ is the projection from ${\Xi}$ to $\Sigma$ defined by $\pi((a,b_1,\dots,b_\ell))=a$.
\item Defining, for every $1\leq i\leq \ell$, $$U_i=\{u\in \dom{t}\mid t_\phi(u) \text{ is of the form }(a,b_1,\dots,b_{i-1},1,b_{i+1},\dots,b_\ell)\},$$ one has $(t,U_1,\dots,U_\ell)\models\phi(X_1,\dots,X_\ell)$.
\end{enumerate}
Intuitively, the second point states that this marking exhibits a valuation of the $X_i$ for which $\phi$ holds in $t$. We refer to $t_\phi$ as a \defin{selector} for $\phi$ in $t$ .
}
\end{definition}

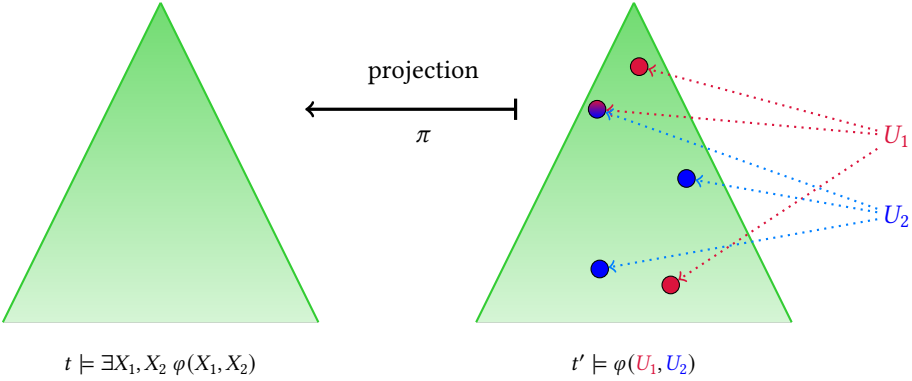
\begin{figure}
\begin{center}
\begin{tikzpicture}[scale=0.7,transform shape]
\fill[
bottom color=LimeGreen!100!black!20,
top color=LimeGreen!100!black!70] (-9,0) -- (-12,-6) -- (-6,-6) -- cycle ;
\draw[LimeGreen,thick] (-9,0) -- (-12,-6);
\draw[LimeGreen,thick] (-9,0) -- (-6,-6);
\node at (-9,-6.8) {\Large $t\models \exists X_1,X_2 \ \phi(X_1,X_2)$};

{
\fill[
bottom color=LimeGreen!100!black!20,
top color=LimeGreen!100!black!70] (0,0) -- (-3,-6) -- (3,-6) -- cycle ;
\draw[LimeGreen,thick] (0,0) -- (-3,-6);
\draw[LimeGreen,thick] (0,0) -- (3,-6);

\node[draw=black,circle,scale=1,fill=blue] (f1) at (-0.65,-5) {};
\node[draw=black,circle,scale=1,fill=Crimson] (f2) at (0.1,-1.2) {};
\node[draw=black,circle,scale=1,fill=blue] (f3) at (1,-3.3) {};
\node[draw=black,circle,scale=1,fill=Crimson] (f4) at (0.7,-5.3) {};
\node[draw=black,circle,scale=1,bottom color=blue,top color=Crimson] (f5) at (-0.7,-2) {};
\draw[<-,color=black,very thick] (-6.25,-2) to  (-2.25,-2);
\draw[color=black,very thick] (-2.25,-2.2) to  (-2.25,-1.8);
\node at (-4,-2.5) {\LARGE \textcolor{black}{$\pi$}};
\node at (-4,-1.3) {\LARGE \textcolor{black}{projection}};
\node at (0,-6.8) {\Large $t'\models  \phi(\textcolor{Crimson}{U_1},\textcolor{blue}{U_2})$};
\node(U1) at (5,-2.5) {\LARGE $\textcolor{Crimson}{U_1}$};
\node(U2) at (5,-4) {\LARGE $\textcolor{blue}{U_2}$};
\draw[->,color=Crimson,thick,dotted] (U1) to  (f2);
\draw[->,color=Crimson,thick,dotted] (U1) to  (f4);
\draw[->,color=Crimson,thick,dotted] (U1) to  (f5);
\draw[->,color=Blue,thick,dotted] (U2) to  (f5);
\draw[->,color=Blue,thick,dotted] (U2) to  (f1);
\draw[->,color=Blue,thick,dotted] (U2) to  (f3);
}

\end{tikzpicture}
\caption{The MSO selection problem}\label{fig:selection}
\end{center}
\end{figure}

\begin{example}{(Local model-checking)}\label{example:selectionvsgmc}
Obviously the MSO selection problem captures the MSO model-checking problem. Indeed, it suffices to take $\ell=0$ (\ie there is no second-order free variable).
\end{example}

\begin{example}{(Global model-checking)}\label{ex:GMCSelection}
The MSO selection problem captures the MSO global model-checking problem. Indeed, consider a tree $t$ and an MSO formula $\phi(x)$ with a first-order free variable and recall that the global model-checking asks for a description (endogenous or exogenous) of the set $\truth{t}{\phi}=\{u\in\dom{t}\mid (t,u)\models\phi(x)\}$. 
Now let $\psi(X)=  x\in X \Leftrightarrow \phi(x) $ Then $\exists X\ \psi(X)$ is always true and there is a unique $U$ such that $(t,U)\models \psi(X)$, and this $U$ is equal to $\truth{t}{\phi}$. Hence, an answer to the MSO selection problem for $\psi(X)$ on $t$ leads to a solution to the global model-checking problem for $\phi$ on $t$.
\end{example}

\begin{example}{(Choice function)}\label{ex:choicevsselection}
We now explain how to use the MSO selection problem to obtain a partition defining a choice function on a well-formed tree $t$.
Consider an MSO formula $\phi(X_1,X_2)$ that expresses the following
\begin{itemize}
\item $X_1$ and $X_2$ form a partition of the nodes of $t$.
\item For all node $x$, there exists a red node $z$ below $x$ and a (finite) subset $U$ of nodes that forms a path from $x$ to $z$, and moreover for all $y\in U$ that is different from $z$ the left-successor (\resp right-successor) of $y$ belongs to $U$ iff $y\in X_1$ (\resp $y\in X_2$).
\end{itemize}

One easily verifies that a solution $(U_1,U_2)$ of the selection problem for $\phi(X_1,X_2)$ on $t$ provides a partition defining a choice function on $t$.
\end{example}

We now define the effective MSO selection property that characterises the classes of generators of trees for which an endogenous approach of the selection problem can be performed.

\begin{definition}{(Effective MSO selection property)}
Let $\mathcal{R}$ be a class of generators of trees. We say that $\mathcal{R}$ has the \defin{effective MSO selection property}  if there is an algorithm that transforms any pair $(R,\phi(X_1,\ldots, X_\ell))$ with $R\in\mathcal{R}$ into some $R_\phi\in\mathcal{R}$ (if exists) such that the tree generated by $R_\phi$ is a selector for $\phi$ in the tree generated by $R$. 
\end{definition}

Quite surprisingly (think of the negative result in Corollary~\ref{cor:choice}), recursion schemes (equivalently CPDA) have the effective MSO selection property. The first proof~\cite{CS12} of this result (which is the one we give here) relies on the connection with the computation of a winning strategy in a CPDA parity game together with the fact that such a strategy can be defined by a CPDA that is \emph{synchronised} with the one defining the arena. Alternative proofs were later given by Haddad in~\cite{Haddad13,HaddadPHD} and by Grellois and Melliès in~\cite{GrelloisM15}. 
Both proofs are very different from the one we give here. Indeed, our proof uses the equi-expressivity theorem to restate the problem as a question on CPDA, and a drawback of this approach is that once the answer is given on the CPDA side one needs to go back to the scheme side, which is not complicated but yields a scheme (that in a sense has been normalised) that is very different from the original one. The advantage of the approaches in \cite{HaddadPHD} (built on top of the intersection types approach by Kobayashi and Ong \cite{KO09}) and in \cite{GrelloisM15} (based on purely denotational arguments and connections with linear logic) is to work directly on the recursion scheme and to succeed to provide as a selector a scheme obtained from the original one by adding duplicated and annotated versions of the terminals.

\begin{theorem}\label{theo:selection}
Higher-order recursion schemes as well as collapsible pushdown automata have the effective MSO selection property.
\end{theorem}

\begin{proof}
Thanks to Theorem~\ref{theo:equi}, it suffices to prove the property for collapsible pushdown automata. 

Let $\phi(X_1,\cdots,X_\ell)$ be an MSO formula with $\ell$ second-order free variables, and let $\mathcal{A}$ be a collapsible pushdown automaton generating a tree $t$. 

Thanks to the well-known equivalence between MSO logic and tree automata (see \eg~\cite{Thomas97}), there is a tree automaton $\mathcal{B}_\phi$ working on $\Sigma\times \{0,1\}^\ell$ trees such that a tree $t_\phi$ is accepted by $\mathcal{B}_\phi$ iff $t_\phi$ is a selector for $\phi$ in $t$.

Recall (see Section~\ref{section:treeAutomata}) that acceptance of a tree by a tree automaton can be seen as existence of a winning strategy in a parity game that is (informally) played as follows. The two players, \Eloise and \Abelard move down the tree a pebble to which is attached a state of the automaton; the play starts at the root (with initial state attached to the pebble); at each round \Eloise provides a valid transition (w.r.t the current state and the current node label) of the automaton and \Abelard moves the pebble to some child and updates the state attached to the pebble according to the transition chosen by \Eloise. In case the pebble reaches a leaf, the play loops forever and \Eloise wins iff the state and the node label are consistent with the acceptance condition on leaves; otherwise the play is infinite and \Eloise wins iff the smallest infinitely often visited priority is even.

For the tree $t_\phi$, the underlying arena of the previous game is essentially a synchronised product of $t_\phi$ with the finite graph corresponding to $\mathcal{B}_\phi$. Now consider a variant of this game where instead of checking whether a given tree $t$ is accepted by $\mathcal{B}_\phi$ the players want to check, for a given tree $t$, whether there exists some $t_\phi$ such that $t_\phi$ is accepted by $\mathcal{B}_\phi$ and $t_\phi$ is a marking of $t$. The game is essentially the same, except that now \Eloise is also giving the marking of the current vertex (\emph{i.e.} $\pi^{-1}$). In this case, the game is a collapsible pushdown game (the arena being obtained as the synchronised product of $t$ defined by a CPDA and the $\mathcal{B}_\phi$ component together with one component where \Eloise is guessing the marking) and one directly checks that \Eloise wins from the root if and only if there is an annotation $t_\phi$ of $t$ that is accepted by $\mathcal{B}_\phi$, \ie $t_\phi$ is a selector for $\phi$ in $t$. Call $\game$ this game and call $\mathcal{A}'$ the underlying CPDA.

Now, apply Theorem~\ref{theorem:strategies} to $\game$. Then either \Eloise has no winning strategy from the initial configuration  and we are done (there is no selector). Otherwise one can effectively construct an $n$-CPDA $\mathcal{B}$ that is \emph{synchronised} with $\mathcal{A}'$ and realises a well-defined \emph{winning} strategy for \Eloise in $\game$ from the initial configuration. As $\mathcal{A}'$ and $\mathcal{B}$ are synchronised, we can consider their synchronised product; call it $\mathcal{A}''$. Hence, in $\mathcal{A}''$ the configurations contain extra informations (coming from $\mathcal{B}$); in particular, for any configuration, if the control state from the $\mathcal{A}'$-component is controlled by \Eloise, then the control state from the $\mathcal{B}$-component indicates the next move \Eloise should play: in particular, it provides a transition of the tree automaton, together with information regarding the marking. Transform $\mathcal{A}''$ by removing every transition that is not consistent with the strategy described by $\mathcal{B}$: then the tree generated by this new CPDA is isomorphic to some $t_\phi$ (that is a marking of $t$) together with an accepting run of $\mathcal{B}_\phi$ on it. Now, if we forget the component from $\mathcal{B}_\phi$ we obtain an $n$-CPDA $\mathcal{A}_\phi$ that generates a selector $t_\phi$ for $\phi$ in $t$.
\end{proof}

\begin{remark}\label{rk:selection_safe}
A similar statement for \emph{safe} schemes as well as for higher-order pushdown automata (\ie collapsible pushdown automata that never perform a collapse) can be deduced from \cite{FrataniPHD,CarayolPHD,CS08}. However, the machinery for general schemes is much more involved.
\end{remark}

\begin{remark}[Selection \emph{vs} Reflection]
In Example~\ref{ex:GMCSelection} we explained how one can reduce the MSO reflection to the MSO selection. In particular Theorem~\ref{theo:selection}  directly implies MSO reflection for recursion  schemes (Theorem~\ref{theo:reflection-MSO}).

One may wonder whether a converse reduction exists, \ie whether one can transform any instance of the selection problem into (possibly several) instance(s) of the reflection problem. 

Think of the simplest case where one deals with a single second-order  free variable. If a reduction from selection to reflexion would exist it would mean that, given any formula $\phi(X)$ and any recursion scheme $\rscheme$ such that $\exists X\, \phi(X)$ holds in the tree $t$ generated by $\rscheme$, then there exists another (possibly more complicated) formula $\psi(x)$ such that if one lets $U=\truth{t}{\psi}$ one has $(t,U)\models \phi(X)$. Note here that  we do not even ask for effectivity in constructing $\psi$ from $\phi$ and that $\psi$ may depend on both $\phi$ and $\rscheme$.

Now remember that selection gives a framework to express the existence of a choice function (see Example~\ref{ex:choicevsselection}; note by the way that one only needs a single free variable here as $X_2$ can safely be defined as the complement of $X_1$). Therefore, if one could reduce selection to reflection, it would mean that there would exist an MSO formula $\psi(x)$ that defines a choice function, which contradicts Corollary~\ref{cor:choice}.

Hence, one can consider that the selection property is strictly more general than the reflection property.
\end{remark}

\bibliographystyle{ACM-Reference-Format}
\bibliography{abbrevs,main}


\newcommand{\noopsort}[1]{} \newcommand{\singleletter}[1]{#1}
  \newcommand{\etal}{et al.}
\begin{thebibliography}{47}


\ifx \showCODEN    \undefined \def \showCODEN     #1{\unskip}     \fi
\ifx \showDOI      \undefined \def \showDOI       #1{#1}\fi
\ifx \showISBNx    \undefined \def \showISBNx     #1{\unskip}     \fi
\ifx \showISBNxiii \undefined \def \showISBNxiii  #1{\unskip}     \fi
\ifx \showISSN     \undefined \def \showISSN      #1{\unskip}     \fi
\ifx \showLCCN     \undefined \def \showLCCN      #1{\unskip}     \fi
\ifx \shownote     \undefined \def \shownote      #1{#1}          \fi
\ifx \showarticletitle \undefined \def \showarticletitle #1{#1}   \fi
\ifx \showURL      \undefined \def \showURL       {\relax}        \fi
\providecommand\bibfield[2]{#2}
\providecommand\bibinfo[2]{#2}
\providecommand\natexlab[1]{#1}
\providecommand\showeprint[2][]{arXiv:#2}

\bibitem[\protect\citeauthoryear{Aehlig}{Aehlig}{2006}]%
        {Aehlig06}
\bibfield{author}{\bibinfo{person}{Klaus Aehlig}.}
  \bibinfo{year}{2006}\natexlab{}.
\newblock \showarticletitle{A Finite Semantics of Simply-Typed Lambda Terms for
  Infinite Runs of Automata}. In \bibinfo{booktitle}{\emph{"Proceedings of
  Computer Science Logic, 15th Annual Conference of the EACSL (CSL~2006)"}}
  \emph{(\bibinfo{series}{Lecture Notes in Computer Science})},
  Vol.~\bibinfo{volume}{4207}. \bibinfo{publisher}{Springer-Verlag},
  \bibinfo{pages}{104--118}.
\newblock


\bibitem[\protect\citeauthoryear{Aehlig, de~Miranda, and Ong}{Aehlig
  et~al\mbox{.}}{2005}]%
        {AdMO05a}
\bibfield{author}{\bibinfo{person}{Klaus Aehlig}, \bibinfo{person}{Jolie de
  Miranda}, {and} \bibinfo{person}{C.-H.~Luke Ong}.}
  \bibinfo{year}{2005}\natexlab{}.
\newblock \showarticletitle{Safety is not a Restriction at Level 2 for String
  Languages}. In \bibinfo{booktitle}{\emph{Proceedings of the 8th International
  Conference on Foundations of Software Science and Computational Structures
  (FoSSaCS 2005)}} \emph{(\bibinfo{series}{Lecture Notes in Computer
  Science})}, Vol.~\bibinfo{volume}{3411}.
  \bibinfo{publisher}{Springer-Verlag}, \bibinfo{pages}{490--501}.
\newblock


\bibitem[\protect\citeauthoryear{Aho and Ullman}{Aho and Ullman}{1971}]%
        {AhoU71}
\bibfield{author}{\bibinfo{person}{Alfred~V. Aho} {and}
  \bibinfo{person}{Jeffrey~D. Ullman}.} \bibinfo{year}{1971}\natexlab{}.
\newblock \showarticletitle{Translations on a Context-Free Grammar}.
\newblock \bibinfo{journal}{\emph{Information and Computation}}
  \bibinfo{volume}{19}, \bibinfo{number}{5} (\bibinfo{year}{1971}),
  \bibinfo{pages}{439--475}.
\newblock


\bibitem[\protect\citeauthoryear{Arnold and {Niwi{\'n}ski}}{Arnold and
  {Niwi{\'n}ski}}{2001}]%
        {AN01}
\bibfield{author}{\bibinfo{person}{Andr\'e Arnold} {and}
  \bibinfo{person}{Damian {Niwi{\'n}ski}}.} \bibinfo{year}{2001}\natexlab{}.
\newblock \bibinfo{booktitle}{\emph{Rudiments of mu-Calculus}}.
  \bibinfo{series}{Studies in Logic and the Foundations of Mathematics},
  Vol.~\bibinfo{volume}{146}.
\newblock \bibinfo{publisher}{Elsevier}.
\newblock


\bibitem[\protect\citeauthoryear{Broadbent, Carayol, Hague, Ong, and
  Serre}{Broadbent et~al\mbox{.}}{2020}]%
        {BCHMOS20}
\bibfield{author}{\bibinfo{person}{Christopher~H. Broadbent},
  \bibinfo{person}{Arnaud Carayol}, \bibinfo{person}{Matthew Hague},
  \bibinfo{person}{Andrzej S. Murawski C.-H.~Luke Ong}, {and}
  \bibinfo{person}{Olivier Serre}.} \bibinfo{year}{2020}\natexlab{}.
\newblock \bibinfo{title}{Collapsible Pushdown Games}.  (\bibinfo{year}{2020}).
\newblock
\urldef\tempurl%
\url{https://www.irif.fr/~serre//PublisMisc/BCHMOS20.pdf}
\showURL{%
\tempurl}


\bibitem[\protect\citeauthoryear{Broadbent, Carayol, Ong, and Serre}{Broadbent
  et~al\mbox{.}}{2010}]%
        {BCOS10}
\bibfield{author}{\bibinfo{person}{Christopher~H. Broadbent},
  \bibinfo{person}{Arnaud Carayol}, \bibinfo{person}{C.-H.~Luke Ong}, {and}
  \bibinfo{person}{Olivier Serre}.} \bibinfo{year}{2010}\natexlab{}.
\newblock \showarticletitle{Recursion Schemes and Logical Reflexion}. In
  \bibinfo{booktitle}{\emph{Proceedings of the 25th Annual IEEE Symposium on
  Logic in Computer Science (LiCS~2010)}}. \bibinfo{publisher}{IEEE Computer
  Society}, \bibinfo{pages}{120--129}.
\newblock


\bibitem[\protect\citeauthoryear{Broadbent and Ong}{Broadbent and Ong}{2009}]%
        {BO09}
\bibfield{author}{\bibinfo{person}{Christopher~H. Broadbent} {and}
  \bibinfo{person}{C.-H.~Luke Ong}.} \bibinfo{year}{2009}\natexlab{}.
\newblock \showarticletitle{On Global Model Checking Trees Generated by
  Higher-Order Recursion Schemes}. In \bibinfo{booktitle}{\emph{Proceedings of
  the 12th International Conference on Foundations of Software Science and
  Computational Structures (FoSSaCS 2009)}} \emph{(\bibinfo{series}{Lecture
  Notes in Computer Science})}, Vol.~\bibinfo{volume}{5504}.
  \bibinfo{publisher}{Springer-Verlag}, \bibinfo{pages}{107--121}.
\newblock


\bibitem[\protect\citeauthoryear{Carayol}{Carayol}{2006}]%
        {CarayolPHD}
\bibfield{author}{\bibinfo{person}{Arnaud Carayol}.}
  \bibinfo{year}{2006}\natexlab{}.
\newblock \emph{\bibinfo{title}{Automates infinis, logiques et langages}}.
\newblock \bibinfo{thesistype}{Ph.D. Dissertation}.
  \bibinfo{school}{Universit\'e de {R}ennes 1}.
\newblock


\bibitem[\protect\citeauthoryear{Carayol}{Carayol}{2019}]%
        {CarayolHDR}
\bibfield{author}{\bibinfo{person}{Arnaud Carayol}.}
  \bibinfo{year}{2019}\natexlab{}.
\newblock \emph{\bibinfo{title}{Automata, Logics and Games for Infinite
  Trees}}.
\newblock habilitation. \bibinfo{school}{Universit\'e Paris Est}.
\newblock


\bibitem[\protect\citeauthoryear{Carayol and L{\"o}ding}{Carayol and
  L{\"o}ding}{2007}]%
        {CarayolL07}
\bibfield{author}{\bibinfo{person}{Arnaud Carayol} {and}
  \bibinfo{person}{Christof L{\"o}ding}.} \bibinfo{year}{2007}\natexlab{}.
\newblock \showarticletitle{MSO on the Infinite Binary Tree: Choice and Order}.
  In \bibinfo{booktitle}{\emph{"Proceedings of Computer Science Logic, 21st
  Annual Conference of the EACSL (CSL~2007)"}} \emph{(\bibinfo{series}{Lecture
  Notes in Computer Science})}, Vol.~\bibinfo{volume}{4646}.
  \bibinfo{publisher}{Springer-Verlag}, \bibinfo{pages}{161--176}.
\newblock


\bibitem[\protect\citeauthoryear{Carayol, Meyer, Hague, Ong, and Serre}{Carayol
  et~al\mbox{.}}{2008}]%
        {CHMOS08}
\bibfield{author}{\bibinfo{person}{Arnaud Carayol}, \bibinfo{person}{Antoine
  Meyer}, \bibinfo{person}{Matthew Hague}, \bibinfo{person}{C.-H.~Luke Ong},
  {and} \bibinfo{person}{Olivier Serre}.} \bibinfo{year}{2008}\natexlab{}.
\newblock \showarticletitle{Winning Regions of Higher-Order Pushdown Games}. In
  \bibinfo{booktitle}{\emph{Proceedings of the 23rd Annual IEEE Symposium on
  Logic in Computer Science (LiCS~2008)}}. \bibinfo{publisher}{IEEE Computer
  Society}, \bibinfo{pages}{193--204}.
\newblock


\bibitem[\protect\citeauthoryear{Carayol and Serre}{Carayol and Serre}{2012}]%
        {CS12}
\bibfield{author}{\bibinfo{person}{Arnaud Carayol} {and}
  \bibinfo{person}{Olivier Serre}.} \bibinfo{year}{2012}\natexlab{}.
\newblock \showarticletitle{Collapsible Pushdown Automata and Labeled Recursion
  Schemes: Equivalence, Safety and Effective Selection}. In
  \bibinfo{booktitle}{\emph{Proceedings of the 27th Annual IEEE Symposium on
  Logic in Computer Science (LiCS~2012)}}. \bibinfo{publisher}{IEEE Computer
  Society}, \bibinfo{pages}{165--174}.
\newblock


\bibitem[\protect\citeauthoryear{Carayol and Slaats}{Carayol and
  Slaats}{2008}]%
        {CS08}
\bibfield{author}{\bibinfo{person}{Arnaud Carayol} {and}
  \bibinfo{person}{Michaela Slaats}.} \bibinfo{year}{2008}\natexlab{}.
\newblock \showarticletitle{Positional Strategies for Higher-Order Pushdown
  Parity Games}. In \bibinfo{booktitle}{\emph{Proceedings of the 33rd
  Symposium, Mathematical Foundations of Computer Science (MFCS~2008)}}
  \emph{(\bibinfo{series}{Lecture Notes in Computer Science})},
  Vol.~\bibinfo{volume}{5162}. \bibinfo{publisher}{Springer-Verlag},
  \bibinfo{pages}{217--228}.
\newblock


\bibitem[\protect\citeauthoryear{Carayol and W{\"o}hrle}{Carayol and
  W{\"o}hrle}{2003}]%
        {Carayol03}
\bibfield{author}{\bibinfo{person}{Arnaud Carayol} {and}
  \bibinfo{person}{Stefan W{\"o}hrle}.} \bibinfo{year}{2003}\natexlab{}.
\newblock \showarticletitle{The {C}aucal Hierarchy of Infinite Graphs in Terms
  of Logic and Higher-Order Pushdown Automata}. In
  \bibinfo{booktitle}{\emph{Proceedings of the 23rd International Conference on
  Foundations of Software Technology and Theoretical Computer Science
  (FST\&TCS~2003)}} \emph{(\bibinfo{series}{Lecture Notes in Computer
  Science})}, Vol.~\bibinfo{volume}{2914}.
  \bibinfo{publisher}{Springer-Verlag}, \bibinfo{pages}{112--123}.
\newblock


\bibitem[\protect\citeauthoryear{Caucal}{Caucal}{2002}]%
        {Caucal02}
\bibfield{author}{\bibinfo{person}{Didier Caucal}.}
  \bibinfo{year}{2002}\natexlab{}.
\newblock \showarticletitle{On Infinite Terms Having a Decidable Monadic
  Theory}. In \bibinfo{booktitle}{\emph{Proceedings of the 27th Symposium,
  Mathematical Foundations of Computer Science (MFCS~2002)}}
  \emph{(\bibinfo{series}{Lecture Notes in Computer Science})},
  Vol.~\bibinfo{volume}{2420}. \bibinfo{publisher}{Springer-Verlag},
  \bibinfo{pages}{165--176}.
\newblock


\bibitem[\protect\citeauthoryear{Church and Rosser}{Church and Rosser}{1936}]%
        {ChurchR1936}
\bibfield{author}{\bibinfo{person}{Alonzo Church} {and}
  \bibinfo{person}{J.~Barkley Rosser}.} \bibinfo{year}{1936}\natexlab{}.
\newblock \showarticletitle{Some Properties of Conversion}.
\newblock \bibinfo{journal}{\emph{Trans. Amer. Math. Soc.}}
  \bibinfo{volume}{39}, \bibinfo{number}{3} (\bibinfo{date}{mar}
  \bibinfo{year}{1936}), \bibinfo{pages}{472--472}.
\newblock


\bibitem[\protect\citeauthoryear{Courcelle}{Courcelle}{1994}]%
        {Courcelle94}
\bibfield{author}{\bibinfo{person}{Bruno Courcelle}.}
  \bibinfo{year}{1994}\natexlab{}.
\newblock \showarticletitle{Monadic Second-Order Definable Graph Transductions:
  A Survey.}
\newblock \bibinfo{journal}{\emph{Theoretical Computer Science}}
  \bibinfo{volume}{126}, \bibinfo{number}{1} (\bibinfo{year}{1994}),
  \bibinfo{pages}{53--75}.
\newblock


\bibitem[\protect\citeauthoryear{Courcelle}{Courcelle}{1995}]%
        {Courcelle95}
\bibfield{author}{\bibinfo{person}{Bruno Courcelle}.}
  \bibinfo{year}{1995}\natexlab{}.
\newblock \showarticletitle{The Monadic Second-Order Logic of Graphs {IX}:
  Machines and their Behaviours}.
\newblock \bibinfo{journal}{\emph{Theoretical Computer Science}}
  \bibinfo{volume}{151} (\bibinfo{year}{1995}), \bibinfo{pages}{125--162}.
\newblock


\bibitem[\protect\citeauthoryear{Damm and Goerdt}{Damm and Goerdt}{1986}]%
        {DG86}
\bibfield{author}{\bibinfo{person}{Werner Damm} {and} \bibinfo{person}{Andreas
  Goerdt}.} \bibinfo{year}{1986}\natexlab{}.
\newblock \showarticletitle{An Automata-Theoretical Characterization of the
  OI-Hierarchy}.
\newblock \bibinfo{journal}{\emph{Information and Computation}}
  \bibinfo{volume}{71} (\bibinfo{year}{1986}), \bibinfo{pages}{1--32}.
\newblock


\bibitem[\protect\citeauthoryear{Ebbinghaus, Flum, and Thomas}{Ebbinghaus
  et~al\mbox{.}}{1996}]%
        {EFT84}
\bibfield{author}{\bibinfo{person}{Heinz-Dieter Ebbinghaus},
  \bibinfo{person}{J{\"o}rg Flum}, {and} \bibinfo{person}{Wolfgang Thomas}.}
  \bibinfo{year}{1996}\natexlab{}.
\newblock \bibinfo{booktitle}{\emph{Mathematical Logic}
  (\bibinfo{edition}{second edition} ed.)}.
\newblock \bibinfo{publisher}{Springer-Verlag}.
\newblock


\bibitem[\protect\citeauthoryear{Engelfriet, Hoogeboom, and van
  Best}{Engelfriet et~al\mbox{.}}{1999}]%
        {EngelfrietHB99}
\bibfield{author}{\bibinfo{person}{Joost Engelfriet},
  \bibinfo{person}{Hendrik~Jan Hoogeboom}, {and} \bibinfo{person}{Jan{-}Pascal
  van Best}.} \bibinfo{year}{1999}\natexlab{}.
\newblock \showarticletitle{Trips on Trees}.
\newblock \bibinfo{journal}{\emph{Acta Cybernetica}} \bibinfo{volume}{14},
  \bibinfo{number}{1} (\bibinfo{year}{1999}), \bibinfo{pages}{51--64}.
\newblock


\bibitem[\protect\citeauthoryear{Fratani}{Fratani}{2006}]%
        {FrataniPHD}
\bibfield{author}{\bibinfo{person}{S\'everine Fratani}.}
  \bibinfo{year}{2006}\natexlab{}.
\newblock \emph{\bibinfo{title}{Automates \`a piles de piles \ldots de piles}}.
\newblock \bibinfo{thesistype}{Ph.D. Dissertation}.
  \bibinfo{school}{Universit\'e de {B}ordeaux}.
\newblock


\bibitem[\protect\citeauthoryear{Grellois and Melli{\`{e}}s}{Grellois and
  Melli{\`{e}}s}{2015}]%
        {GrelloisM15}
\bibfield{author}{\bibinfo{person}{Charles Grellois} {and}
  \bibinfo{person}{Paul{-}Andr{\'{e}} Melli{\`{e}}s}.}
  \bibinfo{year}{2015}\natexlab{}.
\newblock \showarticletitle{Finitary Semantics of Linear Logic and Higher-Order
  Model-Checking} \emph{(\bibinfo{series}{Lecture Notes in Computer Science})},
  Vol.~\bibinfo{volume}{9234}. \bibinfo{publisher}{Springer-Verlag},
  \bibinfo{pages}{256--268}.
\newblock


\bibitem[\protect\citeauthoryear{Gurevich and Harrington}{Gurevich and
  Harrington}{1982}]%
        {GurevichH82}
\bibfield{author}{\bibinfo{person}{Yuri Gurevich} {and} \bibinfo{person}{Leo
  Harrington}.} \bibinfo{year}{1982}\natexlab{}.
\newblock \showarticletitle{Trees, Automata, and Games}. In
  \bibinfo{booktitle}{\emph{Proceedings of the Fourteenth Annual ACM Symposium
  on the Theory of Computing (STOC'82)}}. \bibinfo{publisher}{ACM},
  \bibinfo{pages}{60--65}.
\newblock


\bibitem[\protect\citeauthoryear{Gurevich and Shelah}{Gurevich and
  Shelah}{1983}]%
        {GurevichS83a}
\bibfield{author}{\bibinfo{person}{Yuri Gurevich} {and}
  \bibinfo{person}{Saharon Shelah}.} \bibinfo{year}{1983}\natexlab{}.
\newblock \showarticletitle{Rabin's Uniformization Problem}.
\newblock \bibinfo{journal}{\emph{Journal of Symbolic Logic}}
  \bibinfo{volume}{48}, \bibinfo{number}{4} (\bibinfo{year}{1983}),
  \bibinfo{pages}{1105--1119}.
\newblock


\bibitem[\protect\citeauthoryear{Haddad}{Haddad}{2012}]%
        {Haddad12}
\bibfield{author}{\bibinfo{person}{Axel Haddad}.}
  \bibinfo{year}{2012}\natexlab{}.
\newblock \showarticletitle{IO vs OI in Higher-Order Recursion Schemes}. In
  \bibinfo{booktitle}{\emph{Proceedings 8th Workshop on Fixed Points in
  Computer Science}} \emph{(\bibinfo{series}{Electronic Proceedings in
  Theoretical Computer Science})}, Vol.~\bibinfo{volume}{77}.
  \bibinfo{pages}{23--30}.
\newblock


\bibitem[\protect\citeauthoryear{Haddad}{Haddad}{2013a}]%
        {Haddad13}
\bibfield{author}{\bibinfo{person}{Axel Haddad}.}
  \bibinfo{year}{2013}\natexlab{a}.
\newblock \showarticletitle{Model Checking and Functional Program
  Transformations}. In \bibinfo{booktitle}{\emph{Proceedings of the 33rd
  International Conference on Foundations of Software Technology and
  Theoretical Computer Science (FST\&TCS~2013)}}
  \emph{(\bibinfo{series}{LIPIcs})}, Vol.~\bibinfo{volume}{24}.
  \bibinfo{publisher}{Schloss Dagstuhl - Leibniz-Zentrum f{\"{u}}r Informatik},
  \bibinfo{pages}{115--126}.
\newblock


\bibitem[\protect\citeauthoryear{Haddad}{Haddad}{2013b}]%
        {HaddadPHD}
\bibfield{author}{\bibinfo{person}{Axel Haddad}.}
  \bibinfo{year}{2013}\natexlab{b}.
\newblock \emph{\bibinfo{title}{Shape-Preserving Transformations of
  Higher-Order Recursion Schemes}}.
\newblock \bibinfo{thesistype}{Ph.D. Dissertation}.
  \bibinfo{school}{Universit\'e Paris Diderot - Paris 7}.
\newblock


\bibitem[\protect\citeauthoryear{Hague, Murawski, Ong, and Serre}{Hague
  et~al\mbox{.}}{2008}]%
        {HMOS08}
\bibfield{author}{\bibinfo{person}{Matthew Hague}, \bibinfo{person}{Andrzej~S.
  Murawski}, \bibinfo{person}{C.-H.~Luke Ong}, {and} \bibinfo{person}{Olivier
  Serre}.} \bibinfo{year}{2008}\natexlab{}.
\newblock \showarticletitle{Collapsible Pushdown Automata and Recursion
  Schemes}. In \bibinfo{booktitle}{\emph{Proceedings of the 23rd Annual IEEE
  Symposium on Logic in Computer Science (LiCS~2008)}}.
  \bibinfo{publisher}{IEEE Computer Society}, \bibinfo{pages}{452--461}.
\newblock


\bibitem[\protect\citeauthoryear{Hague, Murawski, Ong, and Serre}{Hague
  et~al\mbox{.}}{2017}]%
        {HMOS17}
\bibfield{author}{\bibinfo{person}{Matthew Hague}, \bibinfo{person}{Andrzej~S.
  Murawski}, \bibinfo{person}{C.{-}H.~Luke Ong}, {and} \bibinfo{person}{Olivier
  Serre}.} \bibinfo{year}{2017}\natexlab{}.
\newblock \showarticletitle{Collapsible Pushdown Automata and Recursion
  Schemes}.
\newblock \bibinfo{journal}{\emph{ACM Transactions on Computational Logic}}
  \bibinfo{volume}{18}, \bibinfo{number}{3} (\bibinfo{year}{2017}),
  \bibinfo{pages}{25:1--25:42}.
\newblock


\bibitem[\protect\citeauthoryear{Hyland and Ong}{Hyland and Ong}{2000}]%
        {HO00}
\bibfield{author}{\bibinfo{person}{J.~Martin~E. Hyland} {and}
  \bibinfo{person}{C.-H.~Luke Ong}.} \bibinfo{year}{2000}\natexlab{}.
\newblock \showarticletitle{{On Full Abstraction for PCF: I. Models,
  Observables and the Full Abstraction Problem, II. Dialogue Games and Innocent
  Strategies, III. A fully Abstract and Universal Game Model}}.
\newblock \bibinfo{journal}{\emph{Information and Computation}}
  \bibinfo{volume}{163} (\bibinfo{year}{2000}), \bibinfo{pages}{285--408}.
\newblock


\bibitem[\protect\citeauthoryear{Janin and Walukiewicz}{Janin and
  Walukiewicz}{1996}]%
        {JW96}
\bibfield{author}{\bibinfo{person}{David Janin} {and} \bibinfo{person}{Igor
  Walukiewicz}.} \bibinfo{year}{1996}\natexlab{}.
\newblock \showarticletitle{On the Expressive Completeness of the Propositional
  mu-Calculus with Respect to Monadic Second Order Logic}. In
  \bibinfo{booktitle}{\emph{Proceedings of the 7th International Conference on
  Concurrency Theory (CONCUR~1996)}} \emph{(\bibinfo{series}{Lecture Notes in
  Computer Science})}, Vol.~\bibinfo{volume}{1119}.
  \bibinfo{publisher}{Springer-Verlag}, \bibinfo{pages}{263--277}.
\newblock


\bibitem[\protect\citeauthoryear{Knapik, Niwi{\'n}ski, and Urzyczyn}{Knapik
  et~al\mbox{.}}{2001}]%
        {KNU01}
\bibfield{author}{\bibinfo{person}{Teodor Knapik}, \bibinfo{person}{Damian
  Niwi{\'n}ski}, {and} \bibinfo{person}{Pawe{\l} Urzyczyn}.}
  \bibinfo{year}{2001}\natexlab{}.
\newblock \showarticletitle{Deciding Monadic Theories of Hyperalgebraic Trees}.
  In \bibinfo{booktitle}{\emph{Proceedings of the 5th conference on Typed
  Lambda Calculi and Applications (TLCA~2001)}} \emph{(\bibinfo{series}{Lecture
  Notes in Computer Science})}, Vol.~\bibinfo{volume}{2044}.
  \bibinfo{publisher}{Springer-Verlag}, \bibinfo{pages}{253--267}.
\newblock


\bibitem[\protect\citeauthoryear{Knapik, Niwi{\'n}ski, and Urzyczyn}{Knapik
  et~al\mbox{.}}{2002}]%
        {KNU02}
\bibfield{author}{\bibinfo{person}{Teodor Knapik}, \bibinfo{person}{Damian
  Niwi{\'n}ski}, {and} \bibinfo{person}{Pawe{\l} Urzyczyn}.}
  \bibinfo{year}{2002}\natexlab{}.
\newblock \showarticletitle{Higher-Order Pushdown Trees Are Easy}. In
  \bibinfo{booktitle}{\emph{Proceedings of the 5th International Conference on
  Foundations of Software Science and Computation Structures (FoSSaCS 2002)}}
  \emph{(\bibinfo{series}{Lecture Notes in Computer Science})},
  Vol.~\bibinfo{volume}{2303}. \bibinfo{publisher}{Springer-Verlag},
  \bibinfo{pages}{205--222}.
\newblock


\bibitem[\protect\citeauthoryear{Knapik, Niwi{\'n}ski, Urzyczyn, and
  Walukiewicz}{Knapik et~al\mbox{.}}{2005}]%
        {KNUW05}
\bibfield{author}{\bibinfo{person}{Teodor Knapik}, \bibinfo{person}{Damian
  Niwi{\'n}ski}, \bibinfo{person}{Pawel Urzyczyn}, {and} \bibinfo{person}{Igor
  Walukiewicz}.} \bibinfo{year}{2005}\natexlab{}.
\newblock \showarticletitle{Unsafe Grammars and Panic Automata}. In
  \bibinfo{booktitle}{\emph{Proceedings of the 32nd International Colloquium on
  Automata, Languages, and Programming (ICALP~2005)}}
  \emph{(\bibinfo{series}{Lecture Notes in Computer Science})},
  Vol.~\bibinfo{volume}{3580}. \bibinfo{publisher}{Springer-Verlag},
  \bibinfo{pages}{1450--1461}.
\newblock


\bibitem[\protect\citeauthoryear{Kobayashi}{Kobayashi}{2009}]%
        {Kobayashi09}
\bibfield{author}{\bibinfo{person}{Naoki Kobayashi}.}
  \bibinfo{year}{2009}\natexlab{}.
\newblock \showarticletitle{Types and Higher-Order Recursion Schemes for
  Verification of Higher-Order Programs}. In
  \bibinfo{booktitle}{\emph{Proceedings of the 36th ACM SIGPLAN-SIGACT
  Symposium on Principles of Programming Languages (POPL~2009)}}.
  \bibinfo{publisher}{ACM}, \bibinfo{pages}{416--428}.
\newblock


\bibitem[\protect\citeauthoryear{Kobayashi and Ong}{Kobayashi and Ong}{2009}]%
        {KO09}
\bibfield{author}{\bibinfo{person}{Naoki Kobayashi} {and}
  \bibinfo{person}{C.-H.~Luke Ong}.} \bibinfo{year}{2009}\natexlab{}.
\newblock \showarticletitle{A Type System Equivalent to the Modal Mu-Calculus
  Model Checking of Higher-Order Recursion Schemes}. In
  \bibinfo{booktitle}{\emph{Proceedings of the 24th Annual IEEE Symposium on
  Logic in Computer Science (LiCS~2009)}}. \bibinfo{publisher}{IEEE Computer
  Society}, \bibinfo{pages}{179--188}.
\newblock


\bibitem[\protect\citeauthoryear{Ong}{Ong}{2006}]%
        {Ong06a}
\bibfield{author}{\bibinfo{person}{C.-H.~Luke Ong}.}
  \bibinfo{year}{2006}\natexlab{}.
\newblock \showarticletitle{On Model-Checking Trees Generated by Higher-Order
  Recursion Schemes}. In \bibinfo{booktitle}{\emph{Proceedings of the 21st
  Annual IEEE Symposium on Logic in Computer Science (LiCS~2006)}}.
  \bibinfo{publisher}{IEEE Computer Society}, \bibinfo{pages}{81--90}.
\newblock


\bibitem[\protect\citeauthoryear{Parys}{Parys}{2012}]%
        {Parys12}
\bibfield{author}{\bibinfo{person}{Pawe{\l} Parys}.}
  \bibinfo{year}{2012}\natexlab{}.
\newblock \showarticletitle{On the Significance of the Collapse Operation}. In
  \bibinfo{booktitle}{\emph{Proceedings of the 27th Annual IEEE Symposium on
  Logic in Computer Science (LiCS~2012)}}. \bibinfo{publisher}{IEEE Computer
  Society}, \bibinfo{pages}{521--530}.
\newblock


\bibitem[\protect\citeauthoryear{Parys}{Parys}{2018}]%
        {Parys18}
\bibfield{author}{\bibinfo{person}{Pawel Parys}.}
  \bibinfo{year}{2018}\natexlab{}.
\newblock \showarticletitle{Recursion Schemes and the {WMSO+U} Logic}. In
  \bibinfo{booktitle}{\emph{Proceedings of the 35th Symposium on Theoretical
  Aspects of Computer Science (STACS~2018)}} \emph{(\bibinfo{series}{LIPIcs})},
  Vol.~\bibinfo{volume}{96}. \bibinfo{publisher}{Schloss Dagstuhl -
  Leibniz-Zentrum f{\"{u}}r Informatik}, \bibinfo{pages}{53:1--53:16}.
\newblock


\bibitem[\protect\citeauthoryear{Piterman and Vardi}{Piterman and
  Vardi}{2004}]%
        {PV04}
\bibfield{author}{\bibinfo{person}{Nir Piterman} {and}
  \bibinfo{person}{Moshe~Y. Vardi}.} \bibinfo{year}{2004}\natexlab{}.
\newblock \showarticletitle{Global Model-Checking of Infinite-State Systems}.
  In \bibinfo{booktitle}{\emph{Proceedings of the 16th International Conference
  on Computer Aided Verification (CAV~2004)}} \emph{(\bibinfo{series}{Lecture
  Notes in Computer Science})}, Vol.~\bibinfo{volume}{3114}.
  \bibinfo{publisher}{Springer-Verlag}, \bibinfo{pages}{387--400}.
\newblock


\bibitem[\protect\citeauthoryear{Rabin}{Rabin}{1969}]%
        {Rabin69}
\bibfield{author}{\bibinfo{person}{Michael~O. Rabin}.}
  \bibinfo{year}{1969}\natexlab{}.
\newblock \showarticletitle{Decidability of Second-Order Theories and Automata
  on Infinite Trees}.
\newblock \bibinfo{journal}{\emph{Trans. Amer. Math. Soc.}}
  \bibinfo{volume}{141} (\bibinfo{year}{1969}), \bibinfo{pages}{1--35}.
\newblock


\bibitem[\protect\citeauthoryear{Salvati and Walukiewicz}{Salvati and
  Walukiewicz}{2013}]%
        {SW13}
\bibfield{author}{\bibinfo{person}{Sylvain Salvati} {and} \bibinfo{person}{Igor
  Walukiewicz}.} \bibinfo{year}{2013}\natexlab{}.
\newblock \showarticletitle{Using Models to Model-Check Recursive Schemes}. In
  \bibinfo{booktitle}{\emph{Proceedings of the 11th International Conference on
  Typed Lambda Calculi and Applications (TLCA'13)}}
  \emph{(\bibinfo{series}{Lecture Notes in Computer Science})},
  Vol.~\bibinfo{volume}{7941}. \bibinfo{publisher}{Springer-Verlag},
  \bibinfo{pages}{189--204}.
\newblock


\bibitem[\protect\citeauthoryear{Streett and Emerson}{Streett and
  Emerson}{1989}]%
        {StreettE89}
\bibfield{author}{\bibinfo{person}{Robert~S. Streett} {and}
  \bibinfo{person}{E.~Allen Emerson}.} \bibinfo{year}{1989}\natexlab{}.
\newblock \showarticletitle{An Automata Theoretic Decision Procedure for the
  Propositional Mu-Calculus}.
\newblock \bibinfo{journal}{\emph{Inf. Comput.}} \bibinfo{volume}{81},
  \bibinfo{number}{3} (\bibinfo{year}{1989}), \bibinfo{pages}{249--264}.
\newblock
\urldef\tempurl%
\url{https://doi.org/10.1016/0890-5401(89)90031-X}
\showDOI{\tempurl}


\bibitem[\protect\citeauthoryear{Thomas}{Thomas}{1997}]%
        {Thomas97}
\bibfield{author}{\bibinfo{person}{Wolfgang Thomas}.}
  \bibinfo{year}{1997}\natexlab{}.
\newblock \showarticletitle{Languages, Automata, and Logic}.
\newblock In \bibinfo{booktitle}{\emph{Handbook of Formal Language Theory}},
  \bibfield{editor}{\bibinfo{person}{G.~Rozenberg} {and}
  \bibinfo{person}{A.~Salomaa}} (Eds.). Vol.~\bibinfo{volume}{III}.
  \bibinfo{publisher}{Springer-Verlag}, \bibinfo{pages}{389--455}.
\newblock


\bibitem[\protect\citeauthoryear{Walukiewicz}{Walukiewicz}{2004}]%
        {Walukiewicz04}
\bibfield{author}{\bibinfo{person}{Igor Walukiewicz}.}
  \bibinfo{year}{2004}\natexlab{}.
\newblock \showarticletitle{A Landscape with Games in the Background}. In
  \bibinfo{booktitle}{\emph{Proceedings of the 19th Annual IEEE Symposium on
  Logic in Computer Science (LiCS~2004)}}. \bibinfo{publisher}{Computer Society
  Press}, \bibinfo{pages}{356--366}.
\newblock


\bibitem[\protect\citeauthoryear{Wilke}{Wilke}{2001}]%
        {Wilke2001}
\bibfield{author}{\bibinfo{person}{Thomas Wilke}.}
  \bibinfo{year}{2001}\natexlab{}.
\newblock \showarticletitle{Alternating Tree Automata, Parity Games and Modal
  $\mu$-Calculus}.
\newblock \bibinfo{journal}{\emph{Bulletin of the Belgian Mathematical
  Society}} \bibinfo{volume}{8}, \bibinfo{number}{2} (\bibinfo{year}{2001}),
  \bibinfo{pages}{359--391}.
\newblock


\end{thebibliography}

\end{document}